\documentclass
[
    twoside,                 
    openany,               
    cleardoublepage = empty, 
    fontsize = 12 pt,        
    american,                
    captions = tableheading, 
    numbers = noenddot,      
    footheight = 35 pt,      
]
{scrbook}

\newif\ifprintVersion   
\newif\iffancyTheorems  
\newif\ifboldNumberSets 
\newif\ifbachelorThesis 

\printVersionfalse
\fancyTheoremstrue
\boldNumberSetstrue
\bachelorThesistrue

\newcommand*{\printTitle}{}

\newcommand*{\printAuthor}{}

\newcommand*{\printSubject}{}

\newcommand*{\printKeywords}{}

\newlength{\extraborderlength}
\newcommand*{\extraBorder}[1]{\setlength{\extraborderlength}{#1}}

\newlength{\mybindingcorrection}
\newcommand*{\bindingCorrection}[1]{\setlength{\mybindingcorrection}{#1}}

    \extraBorder{3 mm}
    
    \bindingCorrection{10 mm}

\usepackage[utf8]{inputenc} 
\usepackage[T1]{fontenc}    

\usepackage{calc} 

\newlength{\myparindent}
\newlength{\myparskip}
\setfootnoterule{7 cm}

\deffootnote[1.2 em]{1.2 em}{0 em}{\makebox[1.4 em][l]{\textbf{\thefootnotemark}}}

\makeatletter%
    \@removefromreset{footnote}{chapter}%
\makeatother

\usepackage[dvipsnames]{xcolor} 

\definecolor{stroke1}{HTML}{2574A9} 

\colorlet{captionlabel}{black}
\colorlet{footerpagenr}{black}
\colorlet{footerchapter}{stroke1}
\colorlet{footerchaptername}{black}
\colorlet{footersection}{stroke1}
\colorlet{footersectionname}{black}
\colorlet{chapternumber}{stroke1}

\newcommand*{\magicratio}{1 * \ratio{193 mm}{210 mm}}

\newlength{\mypaperwidth}
\newlength{\mypaperheight}
\newlength{\mybodywidth}
\newlength{\mybodyheight}
\newlength{\myoutermargin}
\ifprintVersion
\else
\fi

\newlength{\mytopmargin}
\ifprintVersion
\else
\fi

\newlength{\myinnermargin}
\ifprintVersion
\else
\fi

\newlength{\mybottommargin}
\ifprintVersion
\else
\fi

\newcommand{\goldenratio}{1.618}

\newlength{\myheadsep} 
\ifprintVersion
\else
\fi

\newlength{\myfootskip} 
\ifprintVersion
\else
\fi

\newlength{\mymargininnersep} 
\newlength{\mymarginoutersep} 
\ifprintVersion
\else
\fi

\newlength{\mymarginwidth} 
\newlength{\mymarginwidthwithinnersep} 
\usepackage
[
    \ifprintVersion
        paperwidth = \mypaperwidth + 2\extraborderlength + \mybindingcorrection,
        paperheight = \mypaperheight + 2\extraborderlength,
    \else
        paperwidth = \mypaperwidth,
        paperheight = \mypaperheight,
    \fi
    textwidth = \mybodywidth,
    textheight = \mybodyheight,
    outer = \myoutermargin,
    top = \mytopmargin,
    headsep = \myheadsep,
    footskip = \myfootskip,
    marginparsep = \mymargininnersep,
    marginparwidth = \mymarginwidth,
]
{geometry} 

\usepackage
[
]
{scrlayer-scrpage} 

\KOMAoptions
{%
    headwidth = \textwidth + \mymarginwidthwithinnersep,%
    footwidth = \myoutermargin : \textwidth,%
}

\automark[chapter]{chapter}
\automark*[section]{}

\lehead%
{%
    \begin{minipage}[b]{\mymarginwidth}%
        \small\raggedleft\normalfont\textsf{\textbf{\color{footerchapter}\chaptername\ \thechapter}}
    \end{minipage}
}
\cehead{\hspace*{\mymarginwidthwithinnersep}\parbox{\textwidth}{\raggedright\leftmark}}

\rohead%
{%
    \Ifstr{\rightmark}{\leftmark}%
    {%
        \begin{minipage}[b]{\mymarginwidth}%
            \small\raggedright\normalfont\textsf{\textbf{\color{footersection}Chapter\ \thechapter}}%
        \end{minipage}%
    }%
    {%
        \begin{minipage}[b]{\mymarginwidth}%
            \small\raggedright\normalfont\textsf{\textbf{\color{footersection}Section\ \thesection}}%
        \end{minipage}%
    }%
}
\cohead{\hspace*{-\mymarginwidthwithinnersep}\parbox{\textwidth}{\raggedleft\rightmark}}

\lefoot*%
{%
    \vspace*{1 ex}%
    {\color{stroke1}\rule{\myoutermargin - \mymargininnersep}{0.5 mm}}\\
    \begin{minipage}[b]{\myoutermargin - \mymargininnersep}%
        \raggedleft\normalfont\color{footerpagenr}\textbf{\thepage}%
    \end{minipage}%
}
\rofoot*%
{%
    {\color{stroke1}\rule{\myoutermargin - \mymargininnersep}{0.5 mm}}\\
    \begin{minipage}[b]{\myoutermargin - \mymargininnersep}%
        \raggedright\normalfont\color{footerpagenr}\textbf{\thepage}%
    \end{minipage}%
}

\usepackage{caption}
\usepackage{tikz} 

\newlength{\tmpa}
\newlength{\tmpb}

\renewcommand*{\partlineswithprefixformat}[3]%
{%
    #2
    \thispagestyle{empty}
    \ifprintVersion
        \setlength{\tmpa}{0.618\mypaperwidth + \mybindingcorrection + \extraborderlength}%
        \setlength{\tmpb}{0.382\mypaperheight + \extraborderlength}%
    \else
        \setlength{\tmpa}{0.618\mypaperwidth}%
        \setlength{\tmpb}{0.382\mypaperheight}%
    \fi
    \begin{tikzpicture}[overlay, remember picture]%
        \node [inner sep = 0, outer sep = 0, anchor = north] at (current page.north west)%
        {%
            \begin{tikzpicture}[overlay, remember picture]%
            \draw[color = stroke1, line width = 0.7 mm] (\tmpa, 0) -- (\tmpa, -\tmpb);%
            \end{tikzpicture}%
        };%
        \node (align) [align = right, below = \tmpb - 2 ex, inner sep = 0, outer sep = 0, anchor = north west] at (current page.north west)%
        {%
            \hspace{\tmpa}\hspace{0.5 em}\partname\ \thepart\\[1 ex]
            \color{stroke1}#3%
        };%
    \end{tikzpicture}%
}
\RedeclareSectionCommand%
[%
    font = \normalfont\Huge\sffamily,
    prefixfont = \normalfont\Huge\sffamily,
]
{part}

\usepackage{etoolbox}

\newbool{chapterHasANumber}
\newbool{chapterHasAStar}
\renewcommand*{\chapterlinesformat}[3]%
{%
    \Ifnumbered{#1}{\setbool{chapterHasANumber}{true}}{\setbool{chapterHasANumber}{false}}%
    \Ifstr{#2}{}{\setbool{chapterHasAStar}{true}}{\setbool{chapterHasAStar}{false}}%
    \ifboolexpr{bool{chapterHasANumber} and not bool{chapterHasAStar}}%
    {%
        \begin{tikzpicture}[overlay, remember picture]%
            \node [right = \myinnermargin, below = \mytopmargin, inner sep = 0, outer sep = 0, anchor = north west] (numbernode) at (current page.north west)%
            {%
                \hspace{\myinnermargin}%
                \sffamily\fontsize{60}{60}\selectfont%
                \color{chapternumber}%
                \thechapter%
            };%
            \node [inner sep = 0, outer sep = 0, anchor = north west] at (numbernode.south west)%
            {%
                \begin{tikzpicture}[overlay, remember picture]%
                    \draw[color = stroke1, line width = 0.7 mm] (\myinnermargin, -1 ex) -- (\paperwidth, -1 ex);%
                \end{tikzpicture}%
            };%
            \node (align) [text width = \textwidth - 2 cm, align = right, right = \myinnermargin + \mybodywidth, inner sep = 0, outer sep = 0, anchor = east] at (numbernode.west)%
            {%
                #3%
            };%
        \end{tikzpicture}%
    }%
    {%
        \begin{tikzpicture}[overlay, remember picture]%
            \node [right = \myinnermargin, below = \mytopmargin, inner sep = 0, outer sep = 0, anchor = north west] (numbernode) at (current page.north west)%
            {%
                \hspace{\myinnermargin}%
                \sffamily\fontsize{60}{60}\selectfont%
                \color{white}%
                \thechapter%
            };%
            \node [inner sep = 0, outer sep = 0, anchor = north west] at (numbernode.south west)%
            {%
                \begin{tikzpicture}[overlay, remember picture]%
                    \draw[color = stroke1, line width = 0.7 mm] (\myinnermargin, -1 ex) -- (\paperwidth, -1 ex);%
                \end{tikzpicture}%
            };%
            \node (align) [align = left, right = \myinnermargin, inner sep = 0, outer sep = 0, anchor = south west] at (numbernode.south west)%
            {%
                #3%
            };%
        \end{tikzpicture}%
    }%
}
\RedeclareSectionCommand%
[%
    font = \color{stroke1}\normalfont\huge\sffamily,
    afterskip = 20 pt,
]
{chapter}

\BeforeStartingTOC[toc]{\pagestyle{plain}}
\AfterStartingTOC{\thispagestyle{plain}}

\usepackage[style=phys, biblabel=brackets, 
backend=biber, 
natbib=true]{biblatex}

\usepackage
[
    babel = true, 
]
{microtype}           
\usepackage{csquotes} 

\usepackage{amsmath}
\usepackage{amssymb}
\usepackage{amsthm}
\usepackage{thmtools}
\usepackage{mathtools}
\usepackage{thm-restate}
\usepackage{dsfont}        
\usepackage{braceMnSymbol} 
\usepackage{smartdiagram}
\usepackage{verbatim}
\usepackage{tikz}
\usetikzlibrary{calc, shapes, arrows, positioning}
\usepackage{caption}
\usepackage{epigraph}
\usepackage{biblatex}
\usepackage{ragged2e}
\usepackage
[
    ttscale = 0.85, 
]
{libertine} 
\usepackage
[
    libertine,    
    slantedGreek, 
    vvarbb,       
    libaltvw,     
]
{newtxmath} 
\usepackage{url} 
\usepackage{bm}  

\usepackage{graphicx} 
\usepackage
[
    subrefformat = simple, 
    labelformat = simple,  
]
{subcaption}         
\usepackage{wrapfig} 

\columnsep = \mymargininnersep

\usepackage{array}     
\usepackage{booktabs}  
\usepackage{longtable} 
\usepackage{pdflscape} 

\usepackage{enumitem} 

\usepackage
[
    ruled,         
    vlined,        
    linesnumbered, 
]
{algorithm2e} 

\usepackage{xspace}   
\usepackage
[
    shortcuts, 
]
{extdash}             
\usepackage{setspace} 

\usepackage{xparse}    
\usepackage{footnote}  
\usepackage{afterpage} 
\usepackage
[
    textsize = scriptsize, 
]
{todonotes}            

\usepackage{algpseudocode}
\usepackage{mathtools}

\usepackage{listings}
\usepackage{xcolor}
\usepackage[titletoc]{appendix}
\usepackage{float}
\usepackage{blindtext}

\algnotext{EndFor}
\algnotext{EndIf}

\usepackage{braket}

\usepackage
[
    bookmarks = true,                 
    bookmarksopen = false,            
    bookmarksnumbered = true,         
    pdfstartpage = 1,                 
    pdftitle = {{\printTitle}},       
    pdfauthor = {{\printAuthor}},     
    pdfsubject = {{\printSubject}},   
    pdfkeywords = {{\printKeywords}}, 
    breaklinks = true,                
    \ifprintVersion
        hidelinks,                    
    \else
    colorlinks = true,            
    allcolors = stroke1,          
    \fi
]
{hyperref} 

\usepackage
[
    noabbrev,   
    nameinlink, 
]
{cleveref} 
\ifbachelorThesis

\fi

\makeatletter
    \def\IfEmptyTF#1%
    {%
        \if\relax\detokenize{#1}\relax%
            \expandafter\@firstoftwo%
        \else%
            \expandafter\@secondoftwo%
        \fi%
    }
\makeatother

\NewDocumentCommand{\mathOrText}{m}
{%
    \ensuremath{#1}\xspace%
}

\let\originalleft\left
\let\originalright\right
\renewcommand{\left}{\mathopen{}\mathclose\bgroup\originalleft}
\renewcommand{\right}{\aftergroup\egroup\originalright}

\makeatletter
    \DeclareRobustCommand{\bfseries}%
    {%
        \not@math@alphabet\bfseries\mathbf%
        \fontseries\bfdefault\selectfont%
        \boldmath%
    }
\makeatother

\xspaceaddexceptions{]\}}

\allowdisplaybreaks

\crefname{ineq}{inequality}{inequalities}
\creflabelformat{ineq}{#2{\upshape(#1)}#3} 

\crefname{term}{term}{terms}
\creflabelformat{term}{#2{\upshape(#1)}#3}

\let\oldfootnote\footnote

\newlength{\spaceBeforeFootnote} 
\newlength{\spaceAfterFootnote}  

\RenewDocumentCommand{\footnote}{o o o m}%
{%
    \IfNoValueTF{#1}%
    {%
        \oldfootnote{#4}%
    }%
    {%
        \setlength{\spaceBeforeFootnote}{\IfEmptyTF{#1}{0}{#1} em}%
        \IfNoValueTF{#2}%
        {%
            \hspace*{\spaceBeforeFootnote}\oldfootnote{#4}%
        }%
        {%
            \setlength{\spaceAfterFootnote}{\IfEmptyTF{#2}{0}{#2} em}%
            \hspace*{\spaceBeforeFootnote}\IfNoValueTF{#3}{\oldfootnote{#4}}{\oldfootnote[#3]{#4}}\hspace*{\spaceAfterFootnote}%
        }%
    }%
}

\makesavenoteenv{figure}
\makesavenoteenv{table}
\makesavenoteenv{tabular}

\iffancyTheorems
    \declaretheoremstyle
    [
        spaceabove = \topsep,
        spacebelow = \topsep,
        headfont = \bfseries,
        headformat = \textcolor{stroke1}{$\blacktriangleright$} \NAME~\NUMBER \NOTE,
        notefont = \bfseries,
        notebraces = {(}{)},
        bodyfont = \normalfont,
        postheadspace = 0.5 em,
        qed = \textcolor{stroke1}{\bfseries$\blacktriangleleft$},
    ]
    {myTheoremStyle}
    
    \declaretheorem
    [
        style = myTheoremStyle,
        name = Conjecture,
        numberwithin = chapter,
    ]
    {conjecture}
    \declaretheorem
    [
        style = myTheoremStyle,
        name = Proposition,
        sharenumber = conjecture,
    ]
    {proposition}
    \declaretheorem
    [
        style = myTheoremStyle,
        name = Claim,
        sharenumber = conjecture,
    ]
    {claim}
    \declaretheorem
    [
        style = myTheoremStyle,
        name = Lemma,
        sharenumber = conjecture,
    ]
    {lemma}
    \declaretheorem
    [
        style = myTheoremStyle,
        name = Corollary,
        sharenumber = conjecture,
    ]
    {corollary}
    \declaretheorem
    [
        style = myTheoremStyle,
        name = Theorem,
        sharenumber = conjecture,
    ]
    {theorem}
    \declaretheorem
    [
        style = myTheoremStyle,
        name = Definition,
        sharenumber = conjecture,
    ]
    {definition}
    \declaretheorem
    [
        style = myTheoremStyle,
        name = Example,
        sharenumber = conjecture,
    ]
    {example}
    \declaretheorem
    [
        style = myTheoremStyle,
        name = Remark,
        sharenumber = conjecture,
    ]
    {remark}
\else
    \theoremstyle{plain}

\fi

\NewDocumentCommand{\functionTemplate}{m m m m o}%
{%
    \IfNoValueTF{#5}%
    {%
        \mathOrText{#1\left#2{#4}\right#3}%
    }%
    {%
        \mathOrText{#1#5#2{#4}#5#3}%
    }%
}

\newcommand*{\leftBracketType}{(}
\newcommand*{\rightBracketType}{)}

\NewDocumentCommand{\createFunction}{m m o o}%
{%
    \renewcommand*{\leftBracketType}{\IfNoValueTF{#3}{(}{#3}}%
    \renewcommand*{\rightBracketType}{\IfNoValueTF{#4}{)}{#4}}%
    \NewDocumentCommand{#1}{o o}%
    {%
        \IfNoValueTF{##1}%
        {%
            \mathOrText{#2}%
        }%
        {%
            \functionTemplate{#2}{\leftBracketType}{\rightBracketType}{##1}[##2]%
        }%
    }%
}

\DeclareDocumentCommand{\probabilisticFunctionTemplate}{m m O{} o}
{%
    \functionTemplate{#1}%
    {\lbrack}%
    {\rbrack}%
    {#2\IfEmptyTF{#3}{}{\ \IfNoValueTF{#4}{\left}{#4}\vert\ \vphantom{#2}#3\IfNoValueTF{#4}{\right.}{}}}%
    [#4]%
}

\ifboldNumberSets

    \newcommand*{\indicatorFunctionSymbol}{\mathbf{1}}
\else

    \newcommand*{\indicatorFunctionSymbol}{\mathds{1}}
\fi

\RenewDocumentCommand{\Pr}{m O{} o}%
{%
    \probabilisticFunctionTemplate{\mathrm{Pr}}{#1}[#2][#3]%
}

\NewDocumentCommand{\E}{m O{} o}%
{%
    \probabilisticFunctionTemplate{\mathrm{E}}{#1}[#2][#3]%
}

\NewDocumentCommand{\Var}{m O{} o}%
{%
    \probabilisticFunctionTemplate{\mathrm{Var}}{#1}[#2][#3]%
}

\DeclareDocumentCommand{\bigO}{m o}%
{%
    \functionTemplate{\mathrm{O}}{(}{)}{#1}[#2]%
}

\DeclareDocumentCommand{\smallO}{m o}%
{%
    \functionTemplate{\mathrm{o}}{(}{)}{#1}[#2]%
}

\DeclareDocumentCommand{\bigTheta}{m o}%
{%
    \functionTemplate{\upTheta}{(}{)}{#1}[#2]%
}

\DeclareDocumentCommand{\bigOmega}{m o}%
{%
    \functionTemplate{\upOmega}{(}{)}{#1}[#2]%
}

\DeclareDocumentCommand{\smallOmega}{m o}%
{%
    \functionTemplate{\upomega}{(}{)}{#1}[#2]%
}

\DeclareDocumentCommand{\eulerE}{o}%
{%
    \mathOrText{\mathrm{e}\IfNoValueTF{#1}{}{^{#1}}}%
}

\DeclareDocumentCommand{\poly}{m o}%
{%
    \functionTemplate{\mathrm{poly}}{(}{)}{#1}[#2]%
}

\createFunction{\id}{\mathrm{id}}

\NewDocumentCommand{\ind}{m o o}%
{%
    \IfNoValueTF{#2}%
    {%
        \mathOrText{\indicatorFunctionSymbol_{#1}}%
    }%
    {%
        \functionTemplate{\indicatorFunctionSymbol_{#1}}{(}{)}{#2}[#3]%
    }%
}

\DeclareDocumentCommand{\dom}{m o}%
{%
    \functionTemplate{\mathrm{dom}}{(}{)}{#1}[#2]%
}

\DeclareDocumentCommand{\rng}{m o}%
{%
    \functionTemplate{\mathrm{rng}}{(}{)}{#1}[#2]%
}

\DeclareDocumentCommand{\d}{o}%
{%
    \mathrm{d}\IfNoValueTF{#1}{}{^{#1}}%
}

\DeclareDocumentCommand{\set}{m m o}%
{%
    \mathOrText{\IfNoValueTF{#3}{\left}{#3}\{#1\ \IfNoValueTF{#3}{\left}{#3}\vert\ \vphantom{#1}#2\IfNoValueTF{#3}{\right.}{}\IfNoValueTF{#3}{\right}{#3}\}}%
}      

\newtheorem{prop}{PROPOSITION}[section]
\newtheorem{defi}{DEFINITION}[section]

\newtheorem{thm}{THEOREM}[section]

\algnotext{EndFor}
\algnotext{EndIf}

\DeclareMathOperator*{\argmin}{arg\,min}

\newcommand{\ip}[2]{\langle #1 \vert #2 \rangle}
\newcommand{\kb}[2]{\vert #1 \rangle \langle #2 \vert}

\newcommand{\kt}[1]{\vert #1 \rangle}

\hypersetup{colorlinks,%
citecolor=stroke1,%
filecolor=black,%
linkcolor=stroke1,%
urlcolor=stroke1
}

\definecolor{codegreen}{rgb}{0,0.6,0}
\definecolor{codegray}{rgb}{0.5,0.5,0.5}
\definecolor{codepurple}{rgb}{0.58,0,0.82}
\definecolor{backcolour}{rgb}{0.95,0.95,0.92}

\lstdefinestyle{mystyle}{
    backgroundcolor=\color{backcolour},   
    commentstyle=\color{red},
    keywordstyle=\color{codegreen},
    numberstyle=\tiny\color{codegray},
    stringstyle=\color{codepurple},
    basicstyle=\ttfamily\footnotesize,
    breakatwhitespace=false,         
    breaklines=true,                 
    captionpos=b,                    
    keepspaces=true,                 
    numbers=left,                    
    numbersep=5pt,                  
    showspaces=false,                
    showstringspaces=false,
    showtabs=false,                  
    tabsize=2
}

\begin{document}
    
    \frontmatter

\begin{titlepage}
\topmargin -30mm \textheight 240mm
\oddsidemargin 8mm \textwidth 134mm


\thispagestyle{empty}

\begin{figure}
	\includegraphics[scale=.3]{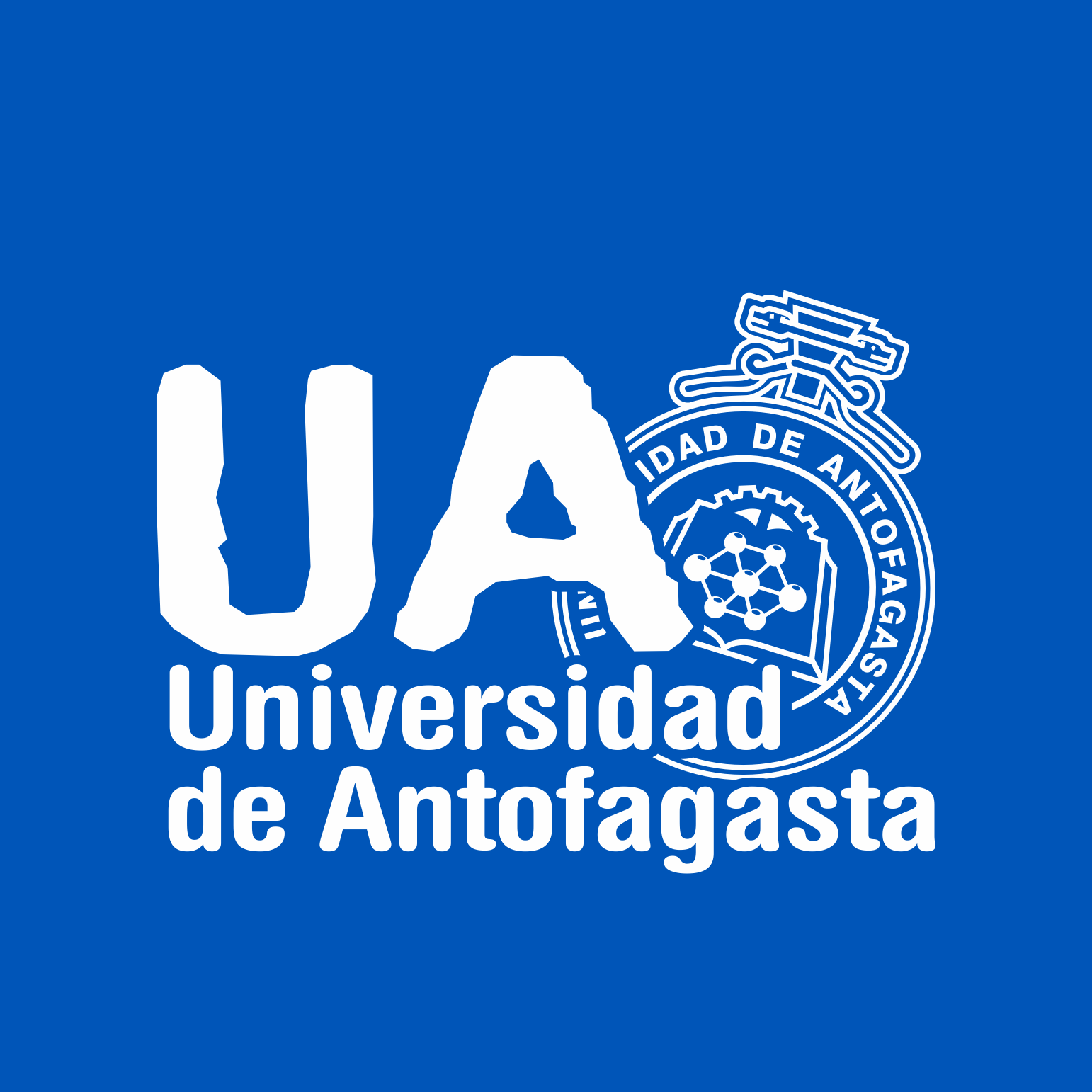}
\end{figure}

\vspace*{-34mm}
\begin{flushright}
	\includegraphics[scale=.3]{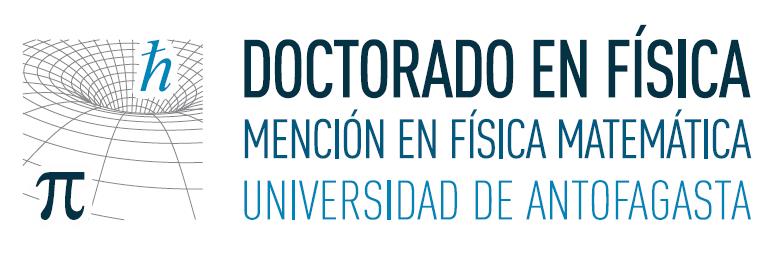}
\end{flushright}

\vspace*{3mm}
\begin{center}
	\begin{large}
		{\Large Universidad de Antofagasta}\\
		Facultad de Ciencias B\'asicas
		
		\vspace*{10mm}
		Doctorado en F\'isica Menci\'on F\'isica -- Matem\'atica
		
		\vspace*{30mm}
		\textbf{ \LARGE Design of optimization tools for quantum information theory }
		
		\vspace*{10mm}
		\textbf{Daniel Uzcátegui Contreras }
		
		\vspace*{10mm}
	\end{large}	
\end{center}

\vspace*{15mm}
\begin{large}
	Profesor Tutor: 
	\hspace{15mm} 
	Dardo Miguel Goyeneche

\vspace*{30mm}
\begin{center}
	Antofagasta, Chile, enero 2022
\end{center}
\end{large}
\newpage
\thispagestyle{empty}

\begin{figure}
	\includegraphics[scale=.3]{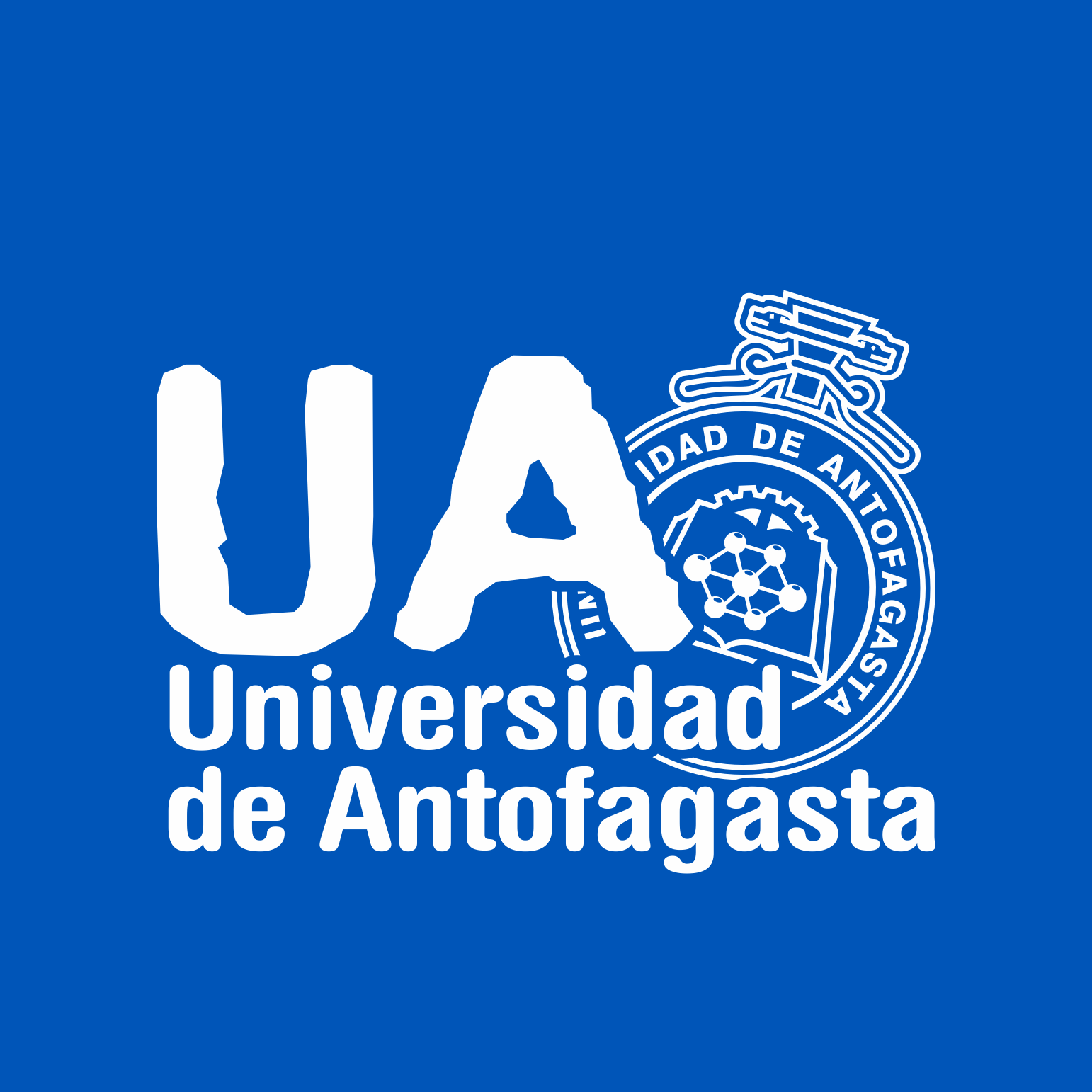}
 \end{figure}

\vspace*{-32mm}
\begin{flushright}
	\includegraphics[scale=.3]{ core/title_page/logodoctorado.jpg}
\end{flushright}

\vspace*{3mm}
\begin{center}
 \begin{large}
 	{\Large Universidad de Antofagasta}\\
 	Facultad de Ciencias B\'asicas
 	
 	\vspace*{10mm}
 	Doctorado en F\'isica Menci\'on F\'isica -- Matem\'atica
 	
 	\vspace*{20mm}
 	\textbf{\LARGE Design of optimization tools for quantum information theory  }
 	
 	\vspace*{10mm}
 	\textbf{Daniel Uzcátegui Contreras }
 	
 	\vspace*{15mm}
 	
 	Tesis presentada ante la ilustre Universidad de Antofagasta para\\ optar al grado acad\'emico de Doctor en F\'isica\\ Menci\'on F\'isica -- Matem\'atica
 	
 \end{large}
\end{center}

\begin{large}
\vspace*{10mm}
\begin{tabbing}
\= Profesor Tutor: 
\hspace{10mm} 
                      \= Dardo Miguel Goyeneche \\[10mm]
\> Comit\a'e de Evaluaci\a'on: \> \\
\> Álvaro Restuccia \> \\   
\> Freddy Lastra \> \\   
\> Algo Delgado \> \\   
\> Leonardo Neves \> \\   
\> Gustavo Cañas \> \\   
	
\end{tabbing}

\vspace*{10mm}
\begin{center}
	Antofagasta, Chile, enero 2022
\end{center}
\end{large}

\end{titlepage}

\restoregeometry
    
    \pagestyle{plain}
    
    \addchap*{}
    \begin{flushright}
\textit{Dedicated to my parents: María and Mario.}
\end{flushright}

    
    \addchap{Acknowledgements}

Here, I would like to thank many people who has helped me to make this thesis possible.\\

First and foremost, I would like to thank my advisor, Prof. Dardo Goyeneche, for all his guidance, encouragement and continuous support on both academic and personal level. I greatly appreciate his patience to help me go through this journey.\\

I would also like to thank all my professors in the program, for their support and for providing me with many tools for my further development. \\

I would like to express my appreciation to my fellow students, for all their help and support, and for all the special moments we shared. \\

My gratitude to the Universidad de Antofagasta, to the Departamento de Física of the Universidad de Antofagasta and to all the people who works to make our PhD program possible.

    
    \addchap{List of publications and ongoing projects}
    This thesis is based on the following publications and projects:

\begin{itemize}
    \item[{[A]}] S. Gómez, D. Uzcátegui, I. Machuca, E. S. Gómez, S. P. Walborn, G. Lima, and D. Goyeneche., \textit{Optimal strategy to certify quantum nonlocality,} \href{https://doi.org/10.1038/s41598-021-99844-2}{Sci Rep 11, 20489 (2021).}
    
    \item[{[B]}] D. Uzcátegui, G. Senno and D. Goyeneche., \textit{Fast and simple quantum state estimation}, 2021 \href{http://dx.doi.org/10.1088/1751-8121/abdba2}{J. Phys. A: Math. Theor. \textbf{54} 085302.}
    
    \item[{[C]}] D. Uzcátegui, G. Senno and D. Goyeneche., \textit{An algorithm for the Quantum Marginal Problem}. In preparation.
\end{itemize}

\vspace{4mm}
\noindent Other related publications during my PhD studies: 

\begin{itemize}
    \item[{[D]}] Contreras, Daniel Uzcátegui; Goyeneche, Dardo; Turek, Ondřej and Václavíková, Zuzana., \textit{Circulant matrices with orthogonal rows and off-diagonal entries of absolute value 1,} \href{https://doi.org/10.2478/cm-2021-0005}{Communications in Mathematics, vol.29, no.1, 2021, pp.15-34.}

\end{itemize}

    \addchap{Abstract}

In this thesis, we present optimization tools for different problems in quantum information theory. First, we introduce an algorithm for \textit{quantum estate estimation}. The algorithm consists of orthogonal projections on intersecting hyperplanes, which are determined by the probability distributions and the measurement operators. We show its performance, in both runtime and fidelity, considering realistic errors. Second, we present a technique for certifying \textit{quantum non-locality}. Given a set of bipartite measurement frequencies, this technique finds a Bell inequality that maximizes the gap between the local hidden variable and the quantum value of a Bell inequality. Lastly, to study the quantum marginal problem, we introduce an operator and develop an algorithm, which takes as inputs a set of quantum marginals and eigenvalues, and outputs a density matrix, if exists, compatible with the prescribed data. 

 
    \setuptoc{toc*}{totoc}
    \tableofcontents
    \listoffigures
    \pagestyle{headings}
    \mainmatter
    
    \addchap{Introduction}
    \pagestyle{plain}

The development of Quantum Mechanics in the early decades of the 20th century brought significant progress in the understanding of the microscopic world. One of the elements of the theory is the \textit{state};  incorporated by Erwin Schrödinger in his equation, the state is a crucial piece for describing quantum systems. Thus, the importance of knowing the quantum state became evident since the beginning of quantum mechanics. In view of this, Wolfang Pauli posed the question of whether the results of measuring momenta and coordinates of a quantum system would be enough to determine its state \cite{Pauli1933}; this is known as the \textit{Pauli's problem}. Since then, many \textit{quantum state estimation} techniques have been developed. Determining the state is ubiquitous for the experimental verification of fundamental aspects of nature predicted by quantum theory and, with the advent of quantum technologies and its promising applications, it is also important for characterization of quantum devices. The increasing number of quantum systems implemented in these devices, e.g. quantum computers, demands for more efficient and scalable methods to determine the quantum state. \\

The implications of quantum mechanics also opened a debate about fundamental issues; most famously Albert Einstein, one of its founders, in a paper co-authored with Boris Podolsky and Nathan Rosen (EPR) \cite{EPR_1935}, argued that quantum mechanics was not a \textit{complete theory}, in the sense that it does not make predictions with certainty. According to Einstein, to complete quantum mechanics it was necessary the introduction of \textit{local deterministic hidden variables (LHV)} \cite{BELOUSEK1996437}.  Nearly thirty years later, John Bell settled many of the concerns posed in the EPR paper \cite{bellTheorem}: he provided, in the form of inequalities, an operational test to determine whether a set of correlations arise within the LHV frame. The most remarkable of Bell's findings is that some predictions from quantum theory can violate these inequalities, thus giving origing to the phenomenon of \textit{nonlocality}. Nowadays, \textit{Bell Nonlocality}\footnote{\textit{Nonlocality} is used to refer to both: the phenomenon and the area of study.} is an area whose main aim is to study and characterize Bell type inequalities. Nonlocality is considered a valuable resource in quantum information theory, with applications in quantum key distribution protocols for secure communication \cite{Acin2007}, random number certification \cite{Pironio2010}, among others. \\


Physical quantum systems formed by many bodies are challenging to understand. Particularly, the relationship between the parts and a whole, namely, to determine whether a set of subsystems, each of them described by a quantum state, is compatible with a global quantum state describing all the parts. This problem is known as  the \textit{Quantum Marginal Problem (QMP)} and has been studied since the beginning of quantum mechanics. The QMP is relevant in many fields, especially in quantum chemistry to understand the formation of molecules \cite{coleman1963}, and in condensed matter physics to study the properties of solids \cite{Yu2021, Huber_2016}. Despite the fact that a solution to the QMP in its most general form has been elusive, important contributions has been made when considering sub-classes of the problem.\\

In this thesis, we present methods to tackle some of the problems mentioned above. These methods are intended to be implemented in practical applications, e.g. in labs focused on quantum information experiments, and also to study problems in quantum information theory. In chapter 1, we introduce the basic concepts required to understand subsequent ideas. In Chapter 2, we develop an algorithm for \textit{quantum state estimation}, show its properties and performance; the main virtues of this method are simplicity and  performance, when compared with state-of-the-art techniques. In Chapter 3, we present a method for \textit{quantum non-locality certification}, based on the optimization of a function using experimental data. This optimization seeks to maximize the gap between the LHV upper bound and the quantum mechanical maximum violation  of a Bell inequality using the available data and, thus, not requiring a new experiment. Finally, in Chapter 4, we present an algorithm for the \textit{quantum marginal problem}, designed to find a multipartite quantum state describing a global system from the knowledge of some reduced parts of the system. All the codes of these tools are in the appendix and are also available in Github.
\newpage
\pagestyle{headings}

    \chapter{Preliminaries}\label{cap1}
    In this chapter we introduce some notions and concepts that are necessary for the development of the ideas in the next chapters. We do not intend to discuss those notions in much depth, but to point out directly to the mathematical properties and physical meaning and illustrate with examples to make the concepts more clear. 

\section{Quantum States}\label{subsection 1.1}
In the Quantum realm, the \emph{state} of a physical system is mathematically described by an element of finite-dimensional Hilbert space $\mathcal{H}$ \cite{NielseChuang}. In Dirac's notation, a $d$-dimensional quantum state can be expanded in an orthonormal basis $\{ \ket{i} \}$ as

\begin{equation}\label{eq1:ket}
    \kt{\psi} = \sum_{i = 0}^{d-1} c_i \ket{i}.
\end{equation}
\noindent Eq. \eqref{eq1:ket} is termed a \emph{pure state}  and the symbol $\kt{\ldots}$ is called a \textit{ket}. $\bra{\psi}$ is an element of the dual space $\mathcal{H}^*$ called \emph{bra}. For $\ket{\psi}, \ket{\phi} \in \mathcal{H}$, the operation $\ip{\psi}{\phi}$ is called the \emph{inner product} and it produces a complex number. Quantum states are normalized, which means that $\| \kt{\psi} \|_{2} = \sqrt{\ip{\psi}{\psi}} = 1$, with $\| \ldots \|_{2}$ the $l_2$-norm.

\subsubsection{Qubits}
A \emph{Qubit}, the fundamental unit of quantum information theory, is a two-level quantum system. Electronic or nuclear spins of an atom and the polarization of a photon are examples of qubits found in nature. For a qubit, the state \eqref{eq1:ket} is $\ket{\psi} = c_0 \ket{0}  + c_1\ket{1}$, with $c_0, c_1 \in \mathbb{C}$. A parametrization of this state is
\begin{equation}\label{qubit}
    \ket{\psi} = \cos(\theta/2) \ket{0}  + e^{i\phi}\sin(\theta/2) \ket{1},
\end{equation}
\noindent with $0 \leq \theta \leq \pi$ and $0 \leq \phi \leq 2\pi$. Geometrically, the pair $(\theta, \phi)$ corresponds to a point on a unit sphere, as shown in Fig. \ref{fig:BlochSphere}. Note that any point on the sphere is a valid pure quantum state.
\begin{figure}
 \centering
 \includegraphics[scale=0.5]{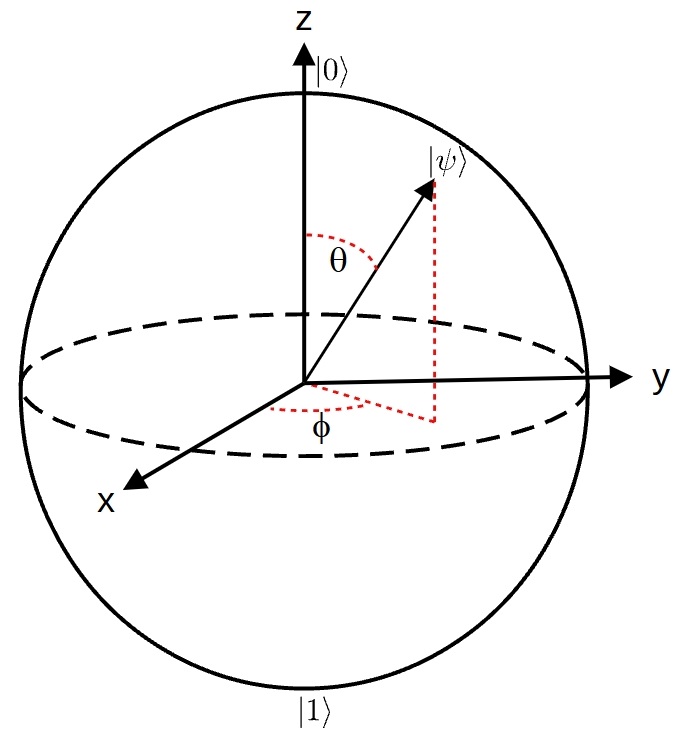}
 \caption[Bloch Sphere]{Bloch Sphere representation for Qubits. The state $| 0 \rangle$ is at the north Pole ($\theta  = 0$) and $| 1 \rangle$ is at the south pole ($\theta = \pi$).    }
 \label{fig:BlochSphere}
\end{figure}

\subsubsection{The Density Operator}
A quantum state can also be expressed as a \emph{Density Operator} or \emph{Density Matrix}.  For the state in Eq. \eqref{eq1:ket}, its density matrix representation $\rho$ is computed as $\rho = \kb{\psi}{\psi} $. Density matrices can result from a statistical mixture of pure states $\kt{\psi_i}$ 
\begin{equation}
 \label{eq3:mixed_states}
 \rho = \sum_{i} p_i \kb{\psi_i}{\psi_i},
\end{equation}
\noindent where $\sum_i p_i = 1$ and $p_i \geq 0$. There is no restriction on the number of terms in the right hand side of Eq. \eqref{eq3:mixed_states}. States in the form of Eq. \eqref{eq3:mixed_states} are called \emph{Mixed States}. A density matrix satisfying $\rho^2 = \rho$ is pure, thus pure states can be considered a special case of mixed states. The purity of a density matrix is quantified by $\mathrm{Tr}(\rho^2)$. Density matrices are \emph{bounded operators}\footnote{Bounded means that, for a norm $\|\ldots\|$, density operators satisfy:\\ $$\|\rho \|:= \text{sup}\left\{\| \rho  \kt{\psi}\| \;\; \big| \;\; \kt{\psi} \in \mathcal{H} \text{  and  } \|\; \kt{\psi}\; \|=1 \right\} < \infty$$} acting on a Hilbert space and have the following properties:

\begin{enumerate}
\item[i)] They are linear
\begin{equation}
    \rho\left( \sum_i c_i \kt{\psi_i} \right) = \sum_i c_i \left(\rho\kt{\psi_i}\right), \quad \text{ with }\ket{\psi_i} \in \mathcal{H} \text{ and } c_i \in \mathbb{C},
\end{equation}
\item[ii)] hermitian $\rho = \rho^{\dagger}$,
\item[iii)] positive-semidefinite, denoted $\rho \geq 0$. This is, 
\begin{equation}
    \bra{\varphi} \rho \ket{\varphi} \geq 0 \quad \text{for all} \quad \ket{\varphi} \in \mathcal{H},
\end{equation}
\item[iv)] and normalized $\mathrm{Tr}\left( \rho \right) = 1.$
\end{enumerate}
\noindent The set of \emph{Density Operators} acting on a Hilbert space $\mathcal{H}$ is denoted $B(\mathcal{H})$. For $\rho_{\alpha} \in B(\mathcal{H})$, with $\alpha$ the state's label, and probabilities $p_{\alpha}$ satisfying $\sum_{\alpha}p_{\alpha} = 1$, the state
\begin{equation}\label{convexCombination}
\rho = \sum_{\alpha} p_{\alpha} \rho_{\alpha},
\end{equation}
\noindent is also an element of $B(\mathcal{H})$. This is, $B(\mathcal{H})$ is a \emph{Convex set} and Eq. \eqref{convexCombination} is a convex combination of the states $\rho_{\alpha}$. The border of this convex set is determined by the set of pure states.

\subsubsection{The Bloch vector}
The density matrix of a $d$-dimensional quantum system may be written as \cite{Hioe_1981}
\begin{equation}\label{blochVecotr_d}
\rho = \dfrac{\hat{I}}{d} + \dfrac{1}{2}\sum_{i = 1}^{d^2-1}x_i\sigma_i,
\end{equation}
\noindent with $\{ \sigma_i \}$ the set of $d^2-1$ generators of the special unitary group ($SU(d)$) and $\hat{I}$ the identity operator of dimension $d$. The generators $\sigma_i$ satisfy the relations $\mathrm{Tr}(\sigma_i \sigma_j) = 2\delta_{ij}$, where $\delta_{ij}$ is the Kronecker delta function. Given a density matrix $\rho$, we can use these properties to calculate the entries of the vector $\Vec{r} = (x_1, \ldots, x_{d^2-1})$ as $x_i = \mathrm{Tr}\big( \rho \sigma_i \big)$; $\Vec{r}$ is known as the \emph{Bloch vector}. For $d = 2$, Eq. \eqref{blochVecotr_d} becomes
\begin{equation}\label{blochVector}
\rho = \dfrac{1}{2}\left(\hat{I} + \sum_{i=1}^{3}x_i\sigma_i \right),
\end{equation}
\noindent with $\{ \sigma_i \}_{i=1}^{3}$ the Pauli Matrices 
\begin{equation}\label{PauliMatrices}
    \sigma_1 = 
    \begin{bmatrix}
      0 & 1 \\
      1 & 0
     \end{bmatrix}, \quad
     \sigma_2 = 
    \begin{bmatrix}
      0 & -i \\
      i & 0
     \end{bmatrix} \quad \text{and } \quad
     \sigma_3 = 
     \begin{bmatrix}
      1 & 0\\
      0 & -1
     \end{bmatrix},
\end{equation}
\noindent and $x_i \in [-1,1]$ for $i =1,2,3$. The Bloch vector has magnitude $\| \Vec{r} \| \leq 1$, where  $\| \Vec{r} \| = 1$ and $\| \Vec{r} \| < 1$ hold for pure and mixed states, respectively. Thus, a qubit density matrix corresponds to a  point in a unit ball called the \emph{Bloch sphere}\footnote{For historical reasons it is not called the \emph{Bloch Ball}.} whose position in the ball is given by $\Vec{r}$; pure states are located on the surface and mixed states inside the ball. For $\Vec{r} = (0,0,0)$, Eq. \eqref{blochVector} becomes $\rho =\hat{I}/2$, which is called the \textit{maximally mixed state}.

\section{Composite systems}
The most general pure state vector of a quantum system composed of $N$ bodies can be expressed using the \emph{Tensor Product} \cite{HorodeckiS}:
\begin{equation}
    \label{eq9: N-party state}
    \ket{\psi} = \sum_{i_1 = 0}^{d_1-1}\ldots \sum_{i_N = 0}^{d_N -1} a_{i_1 \ldots i_N} \ket{i_1} \otimes \ldots \otimes \ket{i_N} \quad a_{i_1 \ldots i_N} \in \mathbb{C},
\end{equation}
\noindent with the set $\{ \ket{i_1} \otimes \ldots \otimes \ket{i_N}\}$ forming an orthonormal basis in the composite space  $\mathcal{H}^{d_1} \otimes \ldots \otimes \mathcal{H}^{d_N}$. Conventionally, to make the notation cleaner, $\ket{i_1} \otimes \ldots \otimes \ket{i_N}$ is often written just as $\ket{i_1 \ldots i_N}$.  It is not always possible to write the state $\ket{\psi}$ of an $N$-body quantum system as a \emph{product state}
\begin{equation}
    \label{eq12: separable state}
    \ket{\psi} = \ket{\varphi_1} \otimes \ket{\varphi_2} \otimes \ldots \otimes \ket{\varphi_N},
\end{equation}
\noindent with $\ket{\varphi_i}$ describing the state of the $i$-th subsystem. States which can be factorized as Eq. \eqref{eq12: separable state} are called \emph{separable}. Physically, separable states are \emph{uncorrelated}, this is, measurement outcomes on individual parts of the system do not depend on measurement outcomes obtained on other bodies. Non-separable pure states, those that cannot be expressed as Eq. \eqref{eq12: separable state}, are \emph{entangled}. In the density matrix representation, a mixed state $\rho$ is entangled when it cannot be written as a convex combination of product states \cite{Werner1989}, this is
\begin{equation}
    \label{entangledDensityMatrix}
    \rho \neq \sum_i p_i \rho_1^i \otimes \ldots \otimes \rho_N^i,
\end{equation}
with $p_i\geq0$ and $\sum_i p_i=1$. In general, determining whether a density matrix is separable or not is a hard task. Nonetheless, there exist some criteria to decide separability in some special cases, some of which will be discussed later in this chapter.

For a 2-body system, with local dimensions $d_1 = 2$ and $d_2 = 2$, Eq. \eqref{eq9: N-party state} becomes 
\begin{equation}
    \label{eq11: two body system}
    \ket{\psi} = a_{00} \ket{00} + a_{01} \ket{01} + a_{10}\ket{10} + a_{11} \ket{11}.
\end{equation}
\noindent  By, for example, setting $a_{00} = a_{11} = 1/\sqrt{2}$ and $a_{01} = a_{10} = 0$ in Eq. \eqref{eq11: two body system}, one obtains the maximally entangled state $\ket{\psi} = \left( \ket{00} + \ket{11} \right)/\sqrt{2}$. With coefficients $a_{00} = 1$ and $a_{11} = a_{01} = a_{10} = 0$, Eq. \eqref{eq11: two body system} becomes $\ket{\psi} =  \ket{00} = \ket{0}\otimes\ket{0}$, which is a separable state .

\subsection{Partial Trace and Marginals}
Given the state of a multipartite quantum system, one might be interested in describing just a part of it, namely, to know the density matrix of a reduced part of the whole system. In general, for a quantum system formed by $N$ bodies, with labels in the set $\mathcal{J} = \{ 1,\ldots,N \}$, and state described by the density matrix $\rho_{\mathcal{J}} \in B(\mathcal{H}^{d_{\mathcal{J}}})$, the state $\rho_{\mathcal{I}}$ of a reduced part of the system can be calculated as 
\begin{equation}
    \rho_{\mathcal{I}} = \mathrm{Tr}_{\mathcal{I}^{c}} \left( \rho_{\mathcal{J}} \right) = \sum_i (\bra{i_{I}} \otimes \mathbb{1}_{I^{c}}) \rho_{\mathcal{J}} (\kt{i_{I}} \otimes \mathbb{1}_{I^{c}}),
    \label{reducedSystem}
\end{equation}
\noindent where $\mathcal{I} \subset \mathcal{J}$, $\mathcal{I}^{c}$ is the complement of $\mathcal{I}$ with respect to $\mathcal{J}$ (see Fig. \ref{marginals}) and $<\{ \kt{i_{\mathcal{I}}} \}$ is an orthonormal basis for the subsystem $\mathcal{I}$. Reductions of a density matrix are called \emph{marginals}. Similar to the trace operation, partial trace is independent of the choice of basis. 

\begin{figure}[h!]
    \centering
    \includegraphics[scale = 0.77]{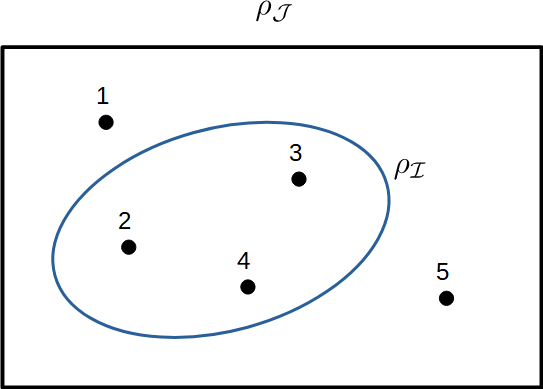}
    \caption[Quantum Marginals]{Illustration of quantum marginals. Enclosed by the rectangle we have a quantum system formed by {\color{red}5} bodies and described by $\rho_{\mathcal{J}}$, with $\mathcal{J} = \{1,2,3,4,5 \}$. The subsystem formed by the bodies $2$, $3$ and $4$, is described by the marginal density matrix $\rho_{\mathcal{I}}$.}
    \label{marginals}
\end{figure}

For an observable $\hat{M}$ (see section \ref{section_1.4}) acting on $B(\mathcal{H}^{d_{\mathcal{I}}})$, the expectation value is computed as $\langle \hat{M} \rangle = \mathrm{Tr} ( \hat{M} \rho_{\mathcal{I}} )$. Although we can calculate $\rho_{\mathcal{I}}$, the subsystem $\mathcal{I}$ is not separated from the whole system $\mathcal{J}$ and thus, for physical consistency,  $\langle \hat{M} \rangle$ has to satisfy
\begin{equation}
    \langle \hat{M} \rangle = \mathrm{Tr} ( \hat{M} \rho_{\mathcal{I}} ) = \mathrm{Tr} \left( ( \hat{M} \otimes \mathbb{1}_{\mathcal{I}^{c}} ) \rho_{\mathcal{J}} \right).
    \label{partialTraceUnique}
\end{equation}
\noindent It turns out that the partial trace $\mathrm{Tr}_{\mathcal{I}^{c}}$, defined in Eq.\eqref{reducedSystem}, is the \emph{unique} linear operation that maps $\rho_{\mathcal{J}}$ into $\rho_{\mathcal{I}}$ such that Eq.\eqref{partialTraceUnique} is satisfied for all quantum states $\rho_{\mathcal{J}}$.

\section{Separability criteria}

\subsubsection{The Schmidt Decomposition}

\noindent For a bipartite system, the state in Eq. \eqref{eq9: N-party state} is given by
\begin{equation}
    \label{BipartiteState}
    \ket{\psi} = \sum_{i_A= 0}^{d_1-1}\sum_{i_B = 0}^{d_2 -1} a_{i_A i_B} \ket{i_A} \ket{i_B},
\end{equation}
\noindent where  $\{ \ket{i_A} \}$ and $\{ \ket{i_B} \}$ are sets of orthonormal bases for the systems $A$ and $B$, respectively. Any bipartite pure state \eqref{BipartiteState} can also be written in the \emph{Schmidt Decomposition} 
\begin{equation}
    \label{SchmidtDecomposition}
    \ket{\psi} = \sum_{i = 0}^{r-1} \lambda_i \ket{\alpha_i} \ket{\beta_i},
\end{equation}
\noindent The number $r$ in \eqref{SchmidtDecomposition} is called the \emph{Schmidt rank} or \emph{Schmidt number}, and corresponds to the number of non-zero singular values $\{ \lambda_i \}$ of the \emph{singular value decomposition} $U S V^{\dagger}$ of the Coefficients Matrix $[ a_{i_A i_B} ]$, with $a_{j_A j_B} = \bra{j_A}\otimes\bra{j_B}\ket{\psi}$ \footnote{For this we use the fact that $\bra{j_A}\otimes\bra{j_B} \ket{i_A}\otimes\ket{i_B} = \delta_{j_A i_A}\delta_{j_B i_B}$.}. The sets of vectors $\{ \ket{\alpha_i} \}$ and $\{\ket{\beta_i} \}$ are the columns of the unitary matrices $U$ and $V$, respectively; they form two orthonormal bases. $S = [ \Lambda \; \; \mathbf{0} ]^T$, with $\Lambda$ a square diagonal matrix with entries $( \lambda_0, \ldots, \lambda_{r-1} )$. If $r = 1$, then $\lambda_0 = 1$ and Eq. \eqref{SchmidtDecomposition} becomes $\ket{\psi} = \ket{\alpha_0} \ket{\beta_0} $. This is, vector states with Schmidt rank equal to one are separable. On the other hand, states with $r > 1$ are non-separable (or entangled).

\subsubsection{Positive Partial Transposition}

In a product basis, a bipartite density matrix $\rho_{AB}$ can be expanded as

\begin{equation}
    \rho_{AB} = \sum_{m,n}^{d_A} \sum_{\mu, \nu}^{d_B} \rho_{m \mu; n \nu} \kt{m}\bra{n} \otimes \kt{\mu}\bra{\nu}.
\end{equation}


The operation of transposing only one of the subsystems , for example $B$, is called \emph{Partial Transposition} and is denoted $\rho_{AB}^{T_{B}}$, with coefficients $\rho_{m\mu;n\nu}^{T_{B}} = \rho_{m\nu;n\mu}$.

The state of a bipartite separable density matrix is $\rho_{AB} = \sum_{i} p_i\rho_{A}^i\otimes\rho_{B}^i$. For this case, the partial transpose $\rho_{AB}^{T_{B}} = \sum_{i} p_i\rho_{A}^i\otimes\rho_{B}^{iT}$ is always a positive matrix, since  $\rho_{B}^{iT} \geq 0$. For entangled states, however, $\rho_{AB}^{T_{B}}$ might result in a matrix with negative eigenvalues. The same occurs for $\rho_{AB}^{T_{A}}$. 

Asher Peres \cite{Peres1996} showed that, for the bipartite case, partial transposition resulting in a positive matrix is a necessary condition for separability; this is called the \emph{Positive Partial Transpose (PPT)} criterion. Later, Horodecki et al. \cite{HORODECKI19961} proved that for $2\times2$ and $2\times3$ systems, positivity of a bipartite density matrix under partial transposition is necessary and sufficient for separability. That is why the PPT criterion is also known as the \emph{Peres–Horodecki criterion}.

\subsubsection{Entanglement Witnesses}

\noindent The Schmidt Decomposition and the PPT criterion are criteria to determine separability in the bipartite case. A more general approach is the concept of an \emph{entanglement witness} \cite{TERHAL2002313}. Let $S_{sep}$ and $S_{ent}$ be the sets of all separable and entangled states, respectively. An observable $\hat{W}$ (see section \ref{section_1.4}) is an entanglement witness if and only if it satisfies the following properties:
\begin{enumerate}
    \item[i)] $\mathrm{Tr}\left( \hat{W} \sigma \right) \geq 0$ for all $\sigma \in S_{sep}$, 
    \item[ii)] $ \mathrm{Tr}\left( \hat{W} \rho_{ent} \right) < 0$ for at least one entangled state $\rho_{ent}$.
\end{enumerate}
\noindent which are necessary and sufficient conditions to determine entanglement. These properties have a geometrical interpretation (see Fig. \ref{Witnesses}). All the states $\rho$ satisfying $\mathrm{Tr}(\hat{W} \rho) = 0$ form a hyperplane that separates the entire set of states in two parts: the left of the hyperplane ($\mathrm{Tr}(\hat{W} \rho) \geq 0$) contains the set $S_{sep}$, to the right ($\mathrm{Tr}(\hat{W} \rho) < 0$) are the entangled states \emph{detected} by $W$. From Fig. \ref{Witnesses} is clear that some witnesses detect more entangled states than others, especially the hyperplanes tangent to $S_{sep}$: these witnesses are called \emph{optimal} \cite{Lewenstein2000}. In Ref. \cite{HORODECKI19961} Horodecki et al. proved that there exists a witness for each $\rho \in S_{ent}$, but constructing  witnesses is not trivial.

\begin{figure}
    \centering
    \includegraphics[scale=0.24]{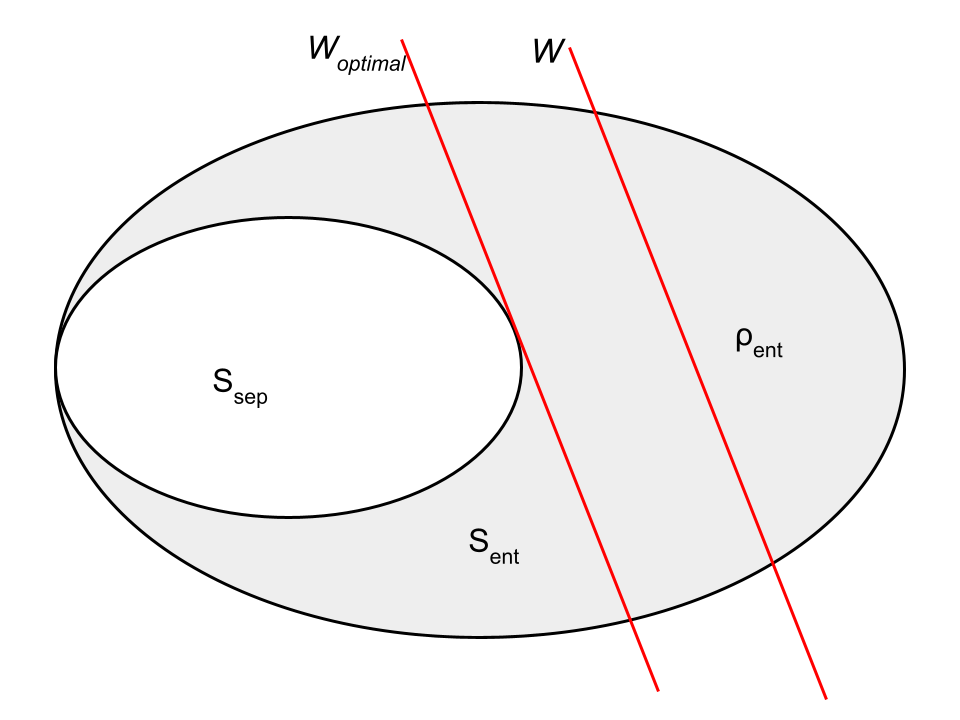}
    \caption[Entanglement Witness]{Schematic representation of the space of states. The red lines correspond to witnesses. $\hat{W}_{optimal}$ touches $S_{sep}$ at only one state: a separable state $\rho$ that satisfies $\mathrm{Tr}(\hat{W} \rho) = 0$.}
    \label{Witnesses}
\end{figure}
Since witnesses are observables, they are an important tool to experimentally detect entanglement. Thus, if in some experiment we measure $ \mathrm{Tr}\left( \hat{W} \rho \right) < 0$, we can assure that $\rho$ is entangled. 

\section{Observables and Measurements}
\label{section_1.4}

\subsection{Observables}
\label{observables}
\noindent Measurable properties of quantum systems are called \emph{observables} and are represented by hermitian linear operators acting on a Hilbert Space. The spectral decomposition of an observable $\hat{A}$ can be written as 
\begin{equation}
    \label{eq6:spectral decompositon}
    \hat{A} = \sum_{n = 0}^{d-1} \lambda_n \kb{\lambda_n}{\lambda_n},
\end{equation}
\noindent with $\{ \lambda_i \}_{i=0}^{d-1}$ the set of eigenvalues (or \textit{spectra}) and $\{ \ket{\lambda_i} \}_{i=0}^{d-1}$ the set of eigenvectors forming an orthonormal basis. The hermitian property $\hat{A} = \hat{A}^{\dagger}$ guarantees that the eigenvalues $\lambda_0, \lambda_1, \ldots, \lambda_{d-1}$ are all real numbers. The eigenvalues correspond to the possible outcomes when a measurement of $\hat{A}$ is performed on a quantum system. If the system to be measured is in the state $\rho$, then the probability of measuring the eigenvalue $\lambda_i$ is given by the Born Rule\cite{Born1926}
\begin{equation}
    \label{eq7: Born Rule}
    p(\lambda_i) = \mathrm{Tr}\big( \rho \hat{\Pi}_i \big).
\end{equation}
\noindent with $\hat{\Pi}_i = \kb{\lambda_i}{\lambda_i}$ . The completeness property $\sum_{i=0}^{d-1} \hat{\Pi}_i = \hat{I}$, with  $\hat{I}$ the identity operator, ensures that $\sum_i p(\lambda_i) = 1$ for any quantum state $\rho$. In practice, to estimate $p(\lambda_i)$, one would have to repeat the measurement procedure on an ensemble of identically prepared quantum systems to compute the frequency with which the eigenvalue $\lambda_i$ is obtained. For a vector state $\ket{\psi}$, the expectation value $\langle \hat{A} \rangle$ of an observable $\hat{A}$ is computed as $\langle \hat{A} \rangle = \bra{\psi} \hat{A} \ket{\psi}$. For a density operator $\rho$, $\langle \hat{A} \rangle = \mathrm{Tr} ( \rho \hat{A} )$.

\subsection{General Quantum Measurements}

\noindent Operators $\hat{\Pi}_i$ from Eq. \eqref{eq7: Born Rule}, which are called \textit{von Neumann Measurements}, are a particular case of a more general formulation called \emph{Positive Operator-Valued Measure (POVM)}. A set $\{ \hat{E}_i \}$ of operators form a POVM if they satisfy the following properties 
\begin{enumerate}
    \item[i)] They are hermitian $\hat{E}_i = \hat{E}_i^{\dagger}$,
    \item[ii)] positive semidefinite $\hat{E}_i \geq 0$,
    \item[iii)] and complete  $\sum_i \hat{E}_i = \hat{I}$.
\end{enumerate}
Operators $\hat{E}_i$ admit the following decomposition
\begin{equation}
    \hat{E}_i = \hat{M}_i\hat{M}_i^{\dagger},
\end{equation}
\noindent with $\hat{M}_i$ a \textit{measurement operator}. In general, the probability $p_i$ of measuring the outcome $i$ associated with the operator $\hat{M}_i$ is given by the Born's rule 
\begin{equation}\label{born_rule_povm}
    p_i= \mathrm{Tr}\big( \hat{M}_i \rho \hat{M}_i^{\dagger} \big) = \mathrm{Tr}\big( \hat{E}_i \rho\big).
\end{equation}
Von Neumann measurements $\hat{\Pi}_i$, besides the three properties above, are orthonormal; this is, $\hat{\Pi}_i \hat{\Pi}_j  = \hat{\Pi}_i \delta_{ij}$. Thus, when $\hat{M}_i = \hat{\Pi}_i$, the POVM element $\hat{E}_i = \hat{\Pi}_i \hat{\Pi}_i = \hat{\Pi}_i$. Von Neumann measurements are also known as \emph{Projective Valued Measure (PVM)} or \emph{projectors}. Unlike PVMs, whose maximum number of possible operators $\hat{\Pi}_i$ is equal to the dimension of the Hilbert space $d$, there is no restriction on the number of elements in the set $\{ \hat{E}_i \}$ forming a POVM.

In general, if a quantum system is initially in a state described by the density matrix $\rho$ and the measurement of this quantum system produces the outcome associated to the measurement operator $\hat{M}_i$, then the post-measurement state $\rho'$ is given by
\begin{equation}\label{projection_postulate}
    \rho' = \dfrac{\hat{M}_i \rho \hat{M}_i^{\dagger}}{\mathrm{Tr}(\hat{M}_i \rho \hat{M}_i^{\dagger})},
\end{equation}
\noindent Eq. \eqref{projection_postulate} is called the \textit{projection postulate}. 

POVMs are ubiquitous for many problems in quantum information theory; it has been shown that POVMs  can perform much better than PVMs in many tasks, such as quantum state discrimination \cite{Bergou_2007}, quantum state estimation \cite{Renes2004}, quantum key-distribution \cite{Renes_2004_2} and randomness certification \cite{Acin_2016}. 

\subsubsection{Mutually Unbiased Bases (MUB)}

Let $\{ \ket{\alpha_i} \}$ and $\{ \ket{\beta_j} \}$ be two orthonormal bases of a $d$-dimensional Hilbert space. These bases are said to be \emph{mutually unbiased} if they satisfy 
\begin{equation}\label{mub1}
    | \ip{\alpha_i}{\beta_j} |^{2} = \dfrac{1}{d}, \quad \text{for every } i,j = 0,\ldots,d-1
\end{equation}
\noindent Also, a set of $m$ bases are MUB if they are pairwise MUB. Eq. \eqref{mub1} can also be written in terms of projectors as 
\begin{equation}\label{mub2}
    \mathrm{Tr}(P_i Q_j) = \dfrac{1}{d},
\end{equation}
\noindent with $P_i = \kb{\alpha_i}{\alpha_i}$ and $Q_j = \kb{\beta_j}{\beta_j}$. For example, for $d = 2$ the bases 
\begin{equation}
    \left \{ \dfrac{\ket{0} + \ket{1}}{\sqrt{2}}, \dfrac{\ket{0} - \ket{1}}{\sqrt{2}} \right \}, \left \{ \dfrac{\ket{0} + i\ket{1}}{\sqrt{2}}, \dfrac{\ket{0} - i\ket{1}}{\sqrt{2}} \right \}, \;\text{and} \;  \left \{ \ket{0}, \ket{1} \right \}, \nonumber
\end{equation}
\noindent corresponding to the eigenvectors of $\sigma_1$, $\sigma_2$ and $\sigma_3$, form three MUB.

MUB find applications in many problems in quantum information theory, including quantum state estimation \cite{Ivonovic_1981, WF1989}, entanglement detection \cite{Spengler_2012} and quantum cryptography \cite{Cerf_2002}. For the case of $d$ being a prime number, Ivonovic \cite{Ivonovic_1981} showed, by explicit construction, that there exist $d + 1$ MUB. An alternative construction is given in \cite{Bandyopadhyay2002ANP}, by defining the following unitary operators
\begin{eqnarray} 
X \ket{j} &=& \ket{j + 1}, \label{Xoperator} \\
Z \ket{j} &=& \omega^j \ket{j}, \label{Zoperator}
\end{eqnarray}
\noindent with $\{ \ket{0}, \ldots, \ket{d-1} \}$ an orthonormal basis and $\omega = \exp(2\pi i/d)$. It is clear from Eqs. \eqref{Xoperator} and \eqref{Zoperator} that $X(Z)^k \ket{j} = (\omega^k)^j \ket{j+1}$. For $d$ a prime number, S. Bandyopadhyay et al. showed (see theorem 2.3 in \cite{Bandyopadhyay2002ANP}) that the eigenvectors bases of the operators $Z$, $X$, $XZ$, $X(Z)^2, \ldots, X(Z)^{d-1}$ form a set of $d+1$ MUBs. 
Wooters and Fields \cite{WOOTTERS1989363} proved that $d+1$ is the maximum number of possible orthonormal bases satisfying \eqref{mub1}, although it is not known whether a maximal set of $d+ 1$ mutually unbiased bases exists in arbitrary dimension $d$. In the same work the authors also proved that there exist $d+1$ MUB when $d = p^n$, with $p$ a prime number and $n$ a positive integer. 

\subsubsection{Symmetric Informationally Complete POVM  (SIC-POVM)}

Let $\{ \ket{\psi_i} \}$ be  a set of $d^2$ normalized vectors belonging to a $d$-dimensional Hilbert space. These vectors are \textit{symmetric} if
\begin{equation}\label{sic-povm}
    | \ip{\psi_i}{\psi_j} |^2 = \dfrac{d \delta_{ij} + 1}{d + 1}, \quad \text{for } i \neq j
\end{equation}
\noindent for all the $\ket{\psi_i}$ in the set. Equivalently, a set of $\{ \hat{\Pi}_i \}$ of $d^2$ subnormalized rank-1 projectors are symmetric if
\begin{equation}\label{symmetric_property_sicpovm}
    \mathrm{Tr}\left( \hat{\Pi}_i \hat{\Pi}_j \right) = \dfrac{d \delta_{ij} + 1}{d + 1}, \quad \text{for } i \neq j
\end{equation}
\noindent for all $\hat{\Pi}_i = \kb{\psi_i}{\psi_i}$. 


A POVM is said to be \textit{informationally complete (IC)} if a quantum state is completely determined by the probabilities obtained from the Born rule \eqref{born_rule_povm}. Thus, an informationally complete POVM $\{ \hat{\Pi}_i \}_{i=0}^{d^2}$ satisfying the completeness property $\sum_i d^{-1}\hat{\Pi}_{i} = \hat{I}$, with elements $\hat{\Pi}_i$ satisfying the symmetric property \eqref{symmetric_property_sicpovm}, defines a \textit{Symmetric Informationally Complete POVM (SIC-POVM)}. Exact solutions for these measurements exist in dimensions $d = 2 - 53$ and 63 more dimensions in the range  $ 57 \leq d \leq 5799$ \cite{sic_povm_list}. Numerical solutions have been found for $d = 2-193$ and $d = 204, 224, 255, 288, 528, 1155$ and $2208$ \cite{Scott_2010, scott2017sics, Grassl_2017, sic_povm_list}. For example, the vectors $\kt{\psi_0} = \kt{0}$, $\kt{\psi_1} = (\kt{0} + \sqrt{2}\kt{1})/\sqrt{3}$, $\kt{\psi_2} = (\kt{0} + \sqrt{2} e^{i \frac{2 \pi}{3}} \kt{1})/\sqrt{3}$ and $\kt{\psi_3} = (\kt{0} + \sqrt{2} e^{i \frac{4\pi}{3}} \kt{1})/\sqrt{3}$ form a SIC-POVM for a qubit system. 

A \textit{complex projective $t$-design} \cite{HOGGAR1982} is formed by a finite set of $m$ normalized vectors $\{ \ket{\psi_i} \}$ satisfying 
\begin{equation}
\dfrac{1}{m^2}
 \sum_{i,j}|\ip{\psi_i}{\psi_j}|^{2t} = \binom{d + t - 1}{t}^{-1},
\end{equation}
\noindent where $\ket{\psi_i} \in \mathcal{H}$. It turns out that maximal set of $d+1$ MUB and SIC-POVMs are complex projective $2$-designs \cite{Klappenecker_2005, scott2017sics}. For these two type of measurements, any $d$-dimensional density matrix $\rho$ admits a decomposition
\begin{equation}\label{tomo_sicpovm}
    \rho = (d + 1)\sum_{i = 1}^{m} p_i \hat{\Pi}_i - \hat{I},
\end{equation}

\noindent with $m=d+1$ for MUB and $m=d^2$ for SIC-POVM. Thus, we can use \eqref{tomo_sicpovm} to reconstruct a quantum state given the experimental probabilities $p_i$ obtained from having performed measurements $\hat{\Pi}_i$. Another advantage of SIC-POVMs is that they minimize the statistical error in the quantum state reconstruction process \cite{Scott_2006}. 


\section{Bell inequalities}

A \textit{Bell scenario} is defined as the arrangement consisting of a given number of parties, the number of measurement settings and the number of possible outcomes for each measurement. In a Bell scenario involving two spatially separated parts (Alice and Bob), a physical system is prepared in a quantum state at some point midway between them, and then part of it is sent to Alice and the remaining part to Bob (see Fig. \ref{fig:bell_scenario}). On her part of the system, Alice can perform one out of $m$ distinct measurements, labeled by $x$, which produce one out of $N$ distinct outcomes, labeled by $a$ (and similarly Bob, with measurements labeled by $y$ and outcomes labeled by $b$). Here, Alice and Bob are free to choose their respective measurements in an arbitrary manner, but they are not allowed to communicate their outcomes during the experiment. On the other hand, they can agree on any pre-established strategy. After carrying out a finite number of identically prepared experiments, we can use the collected data to estimate the joint probabilities $p(a,b|x,y)$ satisfying $p(a,b|x,y) \geq 0$ and $\sum_{a,b}p(a,b|x,y) = 1$.
\begin{figure}
    \centering
    \includegraphics[scale = 0.45]{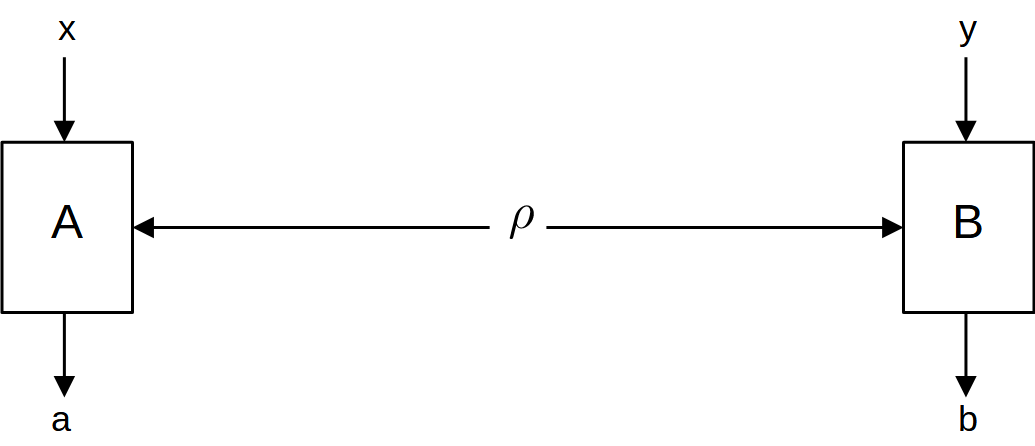}
    \caption[Two-party Bell scenario]{Bipartite Bell scenario. Alice (A) can choose among distinct measurements labeled by $x = 0, \ldots, m-1$, with possible outcomes labeled by $a = 0, \ldots , N-1$. Bob (B) can choose among distinct measurements labeled by $y = 0, \ldots, m-1$, with possible outcomes labeled by $b = 0, \ldots , N-1$.}
    \label{fig:bell_scenario}
\end{figure}

Bell inequalities are defined as 
\begin{equation}\label{BellInequality}
    \sum_{xyab}s_{xy}^{ab}p(a,b|x,y) \leq C,
\end{equation}
\noindent with $-1\leq s_{xy}^{ab} \leq 1$. The set $\{ p(ab|xy) \}$ of joint probabilities satisfying the \emph{locality} condition
\begin{equation}\label{locality-condition}
p(a,b|x,y) = \int d\lambda q(\lambda) p_{\lambda}(a|x)p_{\lambda}(b|y),
\end{equation}
\noindent is called the \emph{local set} and is denoted $\mathcal{L}$. In a local theory, all the predictions $p(ab|xy)$ are ``predetermined'' and this is encoded in the \emph{hidden variable} $\lambda$, whose values are distributed according to the probability density $q(\lambda)$. Equivalently, the condition \eqref{locality-condition} can be expressed as a \textit{local deterministic model}, where the joint probabilities $p(a,b|x,y)$ are statistically independent, meaning that $p(a,b|x,y) = p(a|x)p(b|y)$, and the local probabilities  $p(a|x)$ and $p(b|y)$ can only take the values $0$ and $1$. The upper bound $\mathcal{C}$ in Eq. \eqref{BellInequality}, called the local bound, corresponds to the maximum value that the left-hand side of \eqref{BellInequality} can achieve considering all possible local deterministic strategies. The set $\mathcal{L}$ forms a polytope. Formally, polytopes are defined as the \emph{convex hull} of a finite set \cite{Grunbaum2003}. We may also think of polytopes as a set of half-spaces intercepting each other and enclosing a region of an $n$-dimensional space.

In 1964, J.S Bell \cite{bellTheorem} showed that some predictions $p(a,b|x,y)$ involving quantum systems prepared in a entangled state are incompatible with a local theory\footnote{\emph{Locality}, \emph{local theory} and \emph{local hidden-variable} are frequently used as synonymous in literature.}. This is, for an experiment involving a quantum system in the state $\rho$, there exist predictions given by
\begin{equation}\label{quantum-correlations}
    p(a,b|x,y) = \mathrm{Tr}(\rho M_{a}^x \otimes M_{b}^y),
\end{equation}
that violate Bell inequalities and, therefore, do not satisfy Eq. \eqref{locality-condition}. $M_{a}^x$ and $M_{b}^y$ are the measurements operators of Alice and Bob, respectively. The set of correlations given by Eq. \eqref{quantum-correlations} is denoted $\mathcal{Q}$ and it contains the local set ($\mathcal{L} \subset \mathcal{Q}$). 
\begin{figure}
    \centering
    \includegraphics[scale = 0.675]{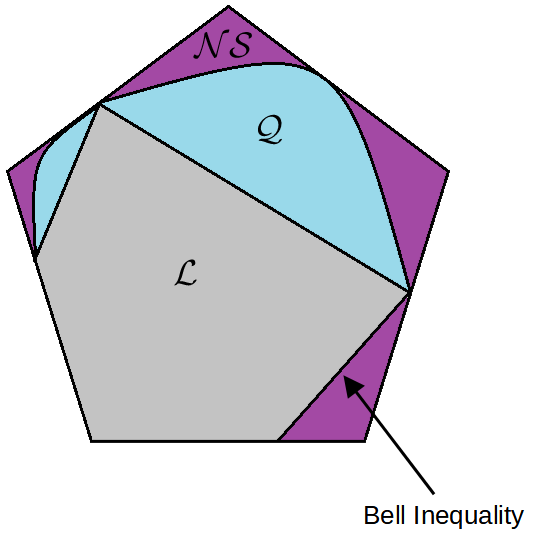}
    \caption[The no-signaling ($\mathcal{NS}$), quantum ($\mathcal{Q}$) and local ($\mathcal{L}$) sets]{Here we depict the $\mathcal{L} \subset \mathcal{Q} \subset \mathcal{NS}$ relation between the no-signaling ($\mathcal{NS}$), quantum ($\mathcal{Q}$) and local ($\mathcal{L}$) sets. $\mathcal{L}$ is a polytope and its facets are Bell inequalities given by Eq. \eqref{BellInequality}. Also, unlike the set $\mathcal{Q}$, $\mathcal{NS}$ is a polytope.}
    \label{fig:local-polytope}
\end{figure}
In Fig. \ref{fig:local-polytope} we show, in addition to $\mathcal{Q}$ and $\mathcal{L}$, the \emph{no-signaling} set $\mathcal{NS}$, which is formed by the entire set joint probabilities $\{p(a,b|x,y)\}$ satisfying the no-signaling constraints \cite{Popescu1994}

\begin{eqnarray}\label{no-signaling}
\sum_{b=0}^{N-1}p(a,b|x,y)=\sum_{b=0}^{N-1}p(a,b|x,y')=:p_A(a|x), \nonumber \\
\sum_{a=0}^{N-1}p(a,b|x,y)=\sum_{a=0}^{N-1}p(a,b|x',y)=:p_B(b|y).
\end{eqnarray}


\noindent Although quantum mechanics agrees the no-signalling principle, there exist probabilities $p(a,b|x,y)$ satisfying \eqref{no-signaling} that do not belong to $\mathcal{Q}$ \cite{Popescu1994}, which implies that $\mathcal{Q} \subset \mathcal{NS}$. Note that Alice's marginal probabilities $p_{A}(a|x)$ are independent of Bob's choice of $y$ (analogously for Bob's marginals $p_{B}(b|y)$). Thus, the no-signaling constraints ensure that no instantaneous exchange of messages between Alice and Bob is possible, preventing Alice's results to be affected by Bob's choice of measurements (and vice versa). The sets $\mathcal{L}$, $\mathcal{Q}$ and $\mathcal{NS}$ are \emph{convex}, meaning that if $p(a,b|x,y)$ and $p(a',b'|x',y')$ belong to one of these sets, then $\beta p(a,b|x,y) + (1 - \beta) p(a',b'|x',y')$ also belongs to the set, for any $0 \leq \beta \leq 1$.

For a bipartite Bell scenario with $x,y \in \{ 0,1 \}$ and $a,b \in \{ 0,1 \}$, a very important inequality is the \emph{CHSH} \cite{CHSH}:
\begin{equation}\label{chsh}
    p(0,0|0,0) + p(0,0|0,1) + p(0,0|1,0) - p(0,0|1,1)  - p_{A}(0|0) - p_{B}(0|0) \leq 0.
\end{equation}

\noindent Considering probabilities $p(a,b|x,y)$ resulting from quantum measurements,  inequality \eqref{chsh} can be violated up to the maximum value of $1/\sqrt{2} - 1/2$.


Bell inequalities can also be written using \emph{correlations}. For the bipartite scenario, the correlations are given by the expectations values of the form $\langle \hat{A}_x \otimes \hat{B}_y \rangle$, with $\hat{A}_x$ and $\hat{B}_y$ Alice's and Bob's observables, respectively. Measurements of $\hat{A}_x$ and $\hat{B}_y$ can produce one of two possible outcomes, namely their eigenvalues, labeled $+1$ and $-1$. Using correlations the CHSH inequality can be equivalently written as
\begin{equation}\label{chsh-correlators}
 \langle \hat{A}_0 \otimes \hat{B}_0 \rangle   + \langle \hat{A}_0 \otimes \hat{B}_1 \rangle   +
 \langle \hat{A}_1 \otimes \hat{B}_0 \rangle   -
 \langle \hat{A}_1 \otimes \hat{B}_1 \rangle   \leq 2.
\end{equation}
\noindent Considering, for example, the maximally entangled state $\ket{\psi} = \left( \ket{01} - \ket{10} \right)/\sqrt{2}$ and observables $\hat{A}_x = \hat{x}\cdot\Vec{\sigma}$ and $\hat{B}_y = \hat{y}\cdot\Vec{\sigma}$, with $\hat{x}$ and $\hat{y}$ unitary vectors indicating the direction in which the measurements are performed, we obtain $\langle \hat{A}_x \otimes \hat{B}_y \rangle = - \hat{x} \cdot \hat{y} $.  For certain directions $\hat{x}$ and $\hat{y}$, is possible to achieve $Q = 2\sqrt{2}$, see \cite{Brunner2014}.  It is always possible to transform an inequality from the correlations notation into probabilities. In the case of a bipartite scenario, we can do that by applying the transformation $\langle \hat{A}_x\hat{B}_y \rangle = \sum_{a,b}(-1)^{a + b}p(a,b|x,y)$.


\newpage
\noindent \text{}
\thispagestyle{plain}
    
    \chapter{Quantum State Estimation}\label{cap2}
    In this chapter, we present an iterative method to study the multipartite state estimation problem. We demonstrate convergence for any given set of informationally complete generalized quantum measurements in every finite dimension. Our method exhibits fast convergence in high dimension and strong robustness under the presence of realistic errors, both in state preparation and measurement stages. In particular, for mutually unbiased bases and tensor product of generalized Pauli observables, it converges in a single iteration.

\section{Introduction}

\emph{Quantum state estimation (QSE)} is the process of estimating the density matrix from measurements performed over an ensemble of identically prepared quantum systems. In the early days of quantum theory, W. Pauli posed the question of whether position and momentum probability distributions univocally determine the state of a quantum particle \cite{Pauli1933}, something that holds in classical mechanics. However, quantum states belong to an abstract Hilbert space whose dimension exponentially increases with the number of particles of the system. Thus, more information than classically expected is required to determine the state. Since then, there has been an increasing interest to estimate the state of a quantum system from a given set of measurements and several methods appeared. For instance, standard state tomography \cite{AJK05} reconstructs $d$-dimensional density matrices from $O(d^3)$ rank-one \emph{Projective Valued Measures (PVM)}, whereas  \emph{mutually unbiased bases (MUB)} \cite{Ivonovic_1981, WF1989} and \emph{Symmetric Informationally Complete (SIC) Positive Operator Valued Measures (POVM)} \cite{Renes2004} do the same task with $O(d^2)$ rank-one measurement projectors. In general, any tight quantum measurement \cite{Scott_2006}, equivalently any complex projective $2$-design, is informationally complete \cite{HOGGAR1982}. 

Quantum state estimation finds applications in communication systems \cite{Molina-Terriza2004}, dissociating molecules \cite{Skovsen_2003} and characterization of  optical devices \cite{D_Ariano_2002}. It is a standard tool for verification of quantum devices, e.g. estimating fidelity of two photon CNOT gates \cite{O_Brien2003}, and has been used to characterize quantum states of trapped ions \cite{Haeffner_2006}, cavity fields \cite{Sayrin2012}, atomic ensembles \cite{Christensen2012} and photons \cite{MAURODARIANO2003205}.

Aside from the experimental procedure of conducting a set of informationally complete measurements on a system, QSE requires an algorithm for reconstructing the state from the measurement statistics. From a variety of techniques proposed, the approaches featuring in the majority of experiments are variants of linear inversion (LI) and maximum-likelihood quantum state estimation (MLE) \cite{Paris2004}. As its name suggests, with LI one determines the state of the quantum system under consideration by inverting the measurement map solving a set of linear equations with the measurement data as input. For relevant families of informationally-complete set of measurements, analytical expressions for the inverse maps are known, significantly speeding up the whole reconstruction effort, see e.g. \cite{Guta2020}. MLE aims to estimate the state that maximizes the probability of obtaining the given experimental data set, among the entire set of density matrices, as we will see in section \ref{QSTtechniques}. Within the different implementations of MLE, those currently achieving the best runtimes are variants of a projected-gradient-descent scheme, see \cite{Superfast2007,Bolduc2017}. Algorithms based on variants of linear inversion \cite{Kaznady2009,Acharya_2019} are typically faster than those implementing MLE when the inversion process is taken from already existing inversion formulas \cite{Guta2020}. On the other hand, when restrictions on the rank of the state being reconstructed apply, techniques based on the probabilistic method of compressed-sensing have proven to be very satisfactory \cite{Gross2010,Cramer2010,Acharya_2017}. Also, the statistics based on five rank-one projective measurements is good enough to have high fidelity reconstruct of rank-one quantum states, even under the presence of errors in both state preparation and measurement stages \cite{Goyeneche2015}.\\  
In Section \ref{QSTtechniques}, we give a brief introduction to some of the most efficient quantum state estimation techniques. In Section \ref{sec:pio}, we introduce the main ingredient of our algorithm: the \emph{Physical Imposition Operator}, a linear operator having an intuitive geometrical interpretation. In Section \ref{sec:algorithm}, we present our iterative algorithm for quantum state estimation based on the physical imposition operator and prove its convergence. In Section  \ref{sec:ultra-convergence}, we show that for a wide class of quantum measurements, which include mutually unbiased bases and tensor product of generalized Pauli observables for $N$ qudit systems, convergence is achieved in a single iteration. In Section \ref{sec:simulations}, we numerically study the performance of our algorithm in terms of runtime and fidelity estimation.

\section{Quantum state estimation techniques}
\label{QSTtechniques}
In this section we review two very well-known techniques for QSE. 
\subsubsection{Semidefinite Programming}
Let us consider the set $ \mathbf{S}^n = \{ X \in \mathbb{R}^{n\times n} | X = X^T \} $, this is, $\mathbf{S}^n$ is the set of $n \times n$ real symmetric matrices. A convex optimization problem with a linear objective function and unknown variable $X$ is a \emph{semidefinite program (SDP)} \cite{boyd2004convex} if it can be stated as
\begin{alignat}{2}
\label{SDP-primal}
&\text{maximize}   & \quad & \mathrm{Tr}\big( C X  \big), \nonumber \\
&\text{subject to} & \quad & \mathrm{Tr}\big( A_i X  \big)  = b_i, \quad i = 1,2,...,m   \\
&                  & \quad & X \geq 0, \nonumber 
\end{alignat}
\noindent with $X, C, A_1,\ldots, A_m \in \mathbf{S}^n$ and $b_i \in \mathbb{R}$. An SDP written as \eqref{SDP-primal} is in its \emph{standard (or primal) form}. An SDP involving matrices with complex entries can be transformed to a real one using the fact that \cite{boyd2004convex}
\begin{align}
    X \geq 0 \Longleftrightarrow 
    \begin{bmatrix}
      \mathfrak{R}X & - \mathfrak{I}X \\
      \mathfrak{I}X & \mathfrak{R}X
     \end{bmatrix} \geq 0,
\end{align}
\noindent with $\mathfrak{R}X$ and $\mathfrak{I}X$ the real and imaginary part of $X$, respectively\footnote{In most of convex optimization literature the symbol "$\succeq$" is used in place of "$\geq$" when referring to positive semidefinite matrices. For example, instead of $X \geq 0$, it is written $X \succeq 0$.}. There is a problem associated to the primal \eqref{SDP-primal}, which consists in solving 
\begin{alignat}{2}
\label{SDP-dual}
&\text{minimize}   & \quad & b^T y, \nonumber \\
&\text{subject to} & \quad & \sum_{j = 1}^{m}y_j A_j - C \geq 0,   \\
&                  & \quad & y \geq 0. \nonumber 
\end{alignat}
\noindent  Eq. \eqref{SDP-dual} is called the \emph{dual} of the primal \eqref{SDP-primal}, with $y \in \mathbb{R}^m$ the unknown variable. Semidefinite programming is an extension of \emph{linear programming} and relies on sophisticated methods, such as the \emph{interior-point methods} \cite{boyd2004convex}, to solve convex optimization problems.  Semidefinite programming has been applied in many fields, including, among others, control theory and combinatorial optimization, for example, in the \emph{MAX-CUT} problem \cite{maxcut}. Semidefinite programming also finds applications in Quantum Information Theory, for example, in discrimination of Quantum States\cite{SDP_POVM}, the Quantum Marginal Problem \cite{Yu2021,Aloy2020TheQM} and Bell Inequalities \cite{Navascues2008}. The problem of determining a density matrix $\rho$ (see Eq. \eqref{born_rule_povm}) given the experimental probabilities $p_i$ obtained from a set of $m$ measurements $\hat{M}_i$, can be formulated as an SDP:
\begin{alignat}{2}\label{SDP primal}
&\text{maximize}   & \quad & \mathrm{Tr}\big( \rho  \big) = 1 , \nonumber \\
&\text{subject to} & \quad & \mathrm{Tr}\big( \hat{M}_i \rho  \big)  = p_i, \quad i = 1,2,...,m   \\
&                  & \quad & \rho \geq 0. \nonumber 
\end{alignat}

\noindent which is maximizing the constant function $1$.

\subsubsection{Maximum Likelihood Estimation (MLE)}

\noindent The modern version of the MLE technique was developed and widespread by Robert Fisher between 1912 and 1922 \cite{Hald1999}. Since then, MLE has successfully been applied for \emph{parameter estimation} and \emph{mathematical modeling} in many fields, ranging from psychology \cite{Myung2003}, finance \cite{Phillips2007}, machine learning \cite{Goodfellow-et-al-2016}, among others. Given some observed data $y= [y_0,  \ldots , y_{m-1}]^T$, with the $y_i$s obtained from some experiment, MLE consists in finding the parameters $w = [w_0, \ldots , w_{m-1}]^T $ that maximizes the Likelihood function $L\left(y|w\right)$; the vector $w$ is called the MLE \emph{estimate}. In practice, for computational convenience, maximizing the log-likelihood $\log\left( L\left(y|w\right) \right)$ is preferred. The Likelihood function $L\left(y|w\right)$ corresponds to \emph{the model} and must be carefully chosen to be representative of the \emph{phenomena or process} that generated the data $y$.

In Quantum Mechanics, MLE is perhaps the most known technique for QSE \cite{Hradil2004} and consists in maximizing the likelihood function
\begin{equation}
    \label{MLE function}
    L\left( f | \rho \right) = \prod_k \mathrm{Tr}\left(\rho \hat{M}_k \right)^{Nf_k},
\end{equation}
\noindent this is, MLE consists in determining the density matrix $\rho$ that maximizes the likelihood function $L\left( f | \rho \right)$ given the experimental relative frequencies $f= [f_0,  \ldots , f_{m-1}]^T$. The frequency $f_i = n_i/N$ is the ratio between the number $n_i$ of occurrences in the detector $\hat{M}_i$ and the number $N$ of identically and independently prepared quantum systems over which the measurements are performed. Considering the log-likelihood function $\log\left( L\left( f | \rho \right) \right)$, the QSE problem can be stated as 
\begin{alignat}{2}
\label{MLE qst}
&\text{minimize}   & \quad & -\dfrac{1}{N}\log\left( L\left( f | \rho \right) \right) = -\sum_{k=0}^{m-1}f_k \mathrm{Tr}\left(\rho \hat{M}_k \right),  \\
&\text{subject to} & \quad &  \rho \geq 0 \quad \text{and} \quad \mathrm{Tr}\left(\rho \right) = 1.  \nonumber
\end{alignat}
\noindent Methods, such as projected-gradient descent \cite{Superfast2007}, have been applied to solve the constrained optimization problem in Eq.\eqref{MLE qst}.

\section{Imposing physical information}\label{sec:pio}
Consider an experimental procedure $\mathcal{P}$ that prepares a quantum system in some \emph{unknown} state $\rho$. Let us  assume that, given some prior knowledge about $\mathcal{P}$, our best guess for $\rho$ is the state $\rho_0$, which could be even the maximally mixed state in absence of prior information. Next, we perform a POVM measurement $A$ composed by $m_A$ outcomes, i.e $A=\{E_i\}_{i\leq m_A}$ on an ensemble of systems independently prepared according to $\mathcal{P}$, obtaining the outcome statistics $\vec{p}=\{p_i\}_{i\leq m_A}$. Given this newly acquired information, \\

\emph{how can we update $\rho_0$ to reflect our new state of knowledge about the system?} \\

To tackle this question, we introduce the \emph{physical imposition operator}, an affine map that updates $\rho_0$ with a matrix containing the probabilities $p_i$. 

\begin{defi}[Physical imposition operator]\label{def:PIOO}
Let $A=\{E_i\}_{i\leq {m_A}}$ be a POVM acting on a $d$-dimensional Hilbert space $\mathcal{H}$ and let $\vec{p}\in\mathbb{R}^{m_A}$ be a probability vector. The physical imposition operator (PIO) associated to $E_i$ and $p_i$ is the map 

\begin{equation}\label{def:pio}
T^{p_i}_{E_i}(\rho)=\rho+\frac{(p_i-\mathrm{Tr}[\rho E_i])E_i}{\mathrm{Tr}(E_i^2)},
\end{equation}
for every $i\leq m_A$.
\end{defi}

In order to clarify the meaning of the physical imposition operator \eqref{def:pio} let us assume for the moment that $A$ is a projective measurement. In such a case, operator $T^{p_i}_{E_i}(\rho)$ takes a quantum state $\rho$, removes the projection along the direction $E_i$, i.e. it removes the physical information about $E_i$ stored in the state $\rho$, and imposes a new projection along this direction, weighted by the probability $p_i$. Here, $p_i$ can be either taken from experimental data or simulated by Born's rule, with respect to a target state to reconstruct. Note that operator $\rho'=T^{p_i}_{E_i}(\rho)$ reflects the experimental knowledge about the quantum system. As we will show in Section \ref{sec:algorithm}, a successive iteration of PIO along an informationally complete set of quantum measurements allows us to univocally reconstruct the quantum state. After several impositions of all involved PIO, the sequence of quantum states successfully converges to a quantum state containing all the imposed physical information, as we demonstrate in Theorem \ref{thm:convergence}. To simplify notation, along the work we drop the superscript $p_i$ in $T_{E_i}^{p_i}$ when the considered probability $p_i$ is clear from the context.


Let us now state some important facts about PIOs that easily arise from \eqref{def:pio}. From now on, $\mathfrak{D}(\rho,\sigma):=\mathrm{Tr}[(\rho-\sigma)^2]$ denotes the Hilbert-Schmidt distance between states $\rho$ and $\sigma$.\\

\begin{prop}\label{prop:pioproperties}
The following properties hold for any POVM $\{E_i\}_{i\leq m_A}$ and any quantum states $\rho,\sigma$ acting on $\mathcal{H}$: 
\begin{enumerate}
\vspace{-1mm}
\item Imposition of physical information: $\mathrm{Tr}[T^{p_i}_{E_i}(\rho)E_i]=p_i.$
\vspace{-2.5mm}
\item Composition: 
\vspace{-3.5mm}
\begin{equation}
    T^{p_j}_{E_j}\circ T^{p_i}_{E_i}(\rho)=T^{p_i}_{E_i}(\rho)+T^{p_j}_{E_j}(\rho)-\rho-\dfrac{\bigl(p_i-\mathrm{Tr}(\rho E_i)\bigr)\mathrm{Tr}(E_iE_j)E_j}{\mathrm{Tr}(E_j^2)\mathrm{Tr}(E_i^2)}. \nonumber
\end{equation}
\vspace{-7.5mm}
\item Non-expansiveness: $\mathfrak{D}(T^{p_j}_{E_j}(\rho),T^{p_j}_{E_j}(\sigma))\leq\mathfrak{D}(\rho,\sigma).$
\end{enumerate}
\end{prop}
\begin{proof}
Items \emph{1} and \emph{2} easily arise from \eqref{def:pio}. In order to show the non-expansiveness stated in Item \emph{3}, let us apply \eqref{def:pio} to two quantum states $\rho$ and $\sigma$, acting on $\mathcal{H}$, i.e.
\begin{equation}\label{pio1Ap}
T^{p_i}_{E_i}(\rho)=\rho+\frac{(p_i-\mathrm{Tr}[\rho E_i])E_i}{\mathrm{Tr}(E_i^2)},
\end{equation}
\begin{equation}\label{pio2Ap}
T^{p_i}_{E_i}(\sigma)=\sigma +\frac{(p_i-\mathrm{Tr}[\sigma E_i])E_i}{\mathrm{Tr}(E_i^2)}.
\end{equation}
Subtracting \eqref{pio1Ap} from \eqref{pio2Ap}
\begin{equation}
T_{E_i}(\rho) - T_{E_i}(\sigma) = (\rho - \sigma) -  \dfrac{\mathrm{Tr}[(\rho - \sigma) E_i]E_i}{\mathrm{Tr}(E_i^2)},
\end{equation}
where we have dropped the upper index $p_i$ from $T^{p_i}_{E_i}$. Now, let us compute $\mathfrak{D}(T_{E_j}(\rho),T_{E_j}(\sigma))^2 = \mathrm{Tr}\bigl[ \bigl( T_{E_i}(\rho) - T_{E_i}(\sigma) \bigr)\bigl( T_{E_i}(\rho) - T_{E_i}(\sigma) \bigr)^{\dagger} \bigr]$. Thus,
\begin{align}
\begin{split}
    \mathfrak{D}(T_{E_j}(\rho),T_{E_j}(\sigma))^2 ={}& \mathfrak{D}(\rho, \sigma)^2 -  2\dfrac{\mathrm{Tr}[(\rho - \sigma) E_i]\mathrm{Tr}[(\rho - \sigma) E_i]}{\mathrm{Tr}(E_i^2)} \\
         & \quad \quad  + \dfrac{\bigl( \mathrm{Tr}[(\rho - \sigma) E_i]\bigr)^2 \mathrm{Tr}(E_i^2)}{\bigl( \mathrm{Tr}(E_i^2) \bigr)^2} \nonumber
\end{split}\\
     ={}& \mathfrak{D}(\rho, \sigma)^2 -  \dfrac{\bigl( \mathrm{Tr}[(\rho - \sigma) E_i]\bigr)^2}{\mathrm{Tr}(E_i^2)},
\end{align}
where $\mathfrak{D}(\rho, \sigma)^2 = \mathrm{Tr}\bigl[(\rho - \sigma)(\rho - \sigma)^{\dagger} \bigr]$. Therefore, $\mathfrak{D}(T_{E_j}(\rho),T_{E_j}(\sigma)) \leq \mathfrak{D}(\rho, \sigma)$ and item \emph{3} holds. 
\end{proof}
\vspace{-1mm}
Some important observations arise from Prop. \ref{prop:pioproperties}. First, for $j=i$ in the above item \emph{2} we find that 
\begin{equation}\label{projection}
T^2_{E_i}(\rho)=T_{E_i}(\rho).
\end{equation}
Note that any quantum state $\sigma=T_{E_i}(\rho)$ is a fixed point of $T_{E_i}$, i.e. $T_{E_i}(\sigma)=\sigma $, which simply arises from \eqref{projection}. Roughly speaking, quantum states already having the physical information we want to impose are fixed points of the map $T_{E_j}$. This key property allows us to apply dynamical systems theory \cite{strogatz_2000} to study the quantum state estimation problem. We consider the alternating projection method, firstly studied by Von Neumann \cite{VonNeumann_1949} for the case of two alternating projections and generalized by Halperin to any number of projections \cite{halperin_1962}, which strongly converges to a point onto the intersection of a finite number of affine subspaces \cite{PANG2015419}. 

In Theorem \ref{thm:convergence}, we will show that composition of all physical imposition operators associated to an informationally complete POVM produces a map having a unique attractive fixed point, i.e., the solution to the quantum state tomography problem. The uniqueness of the fixed point guarantees a considerable speed up of the method in practice, as any chosen seed monotonically approaches to the solution of the problem. 

To simplify notation, we consider a single physical imposition operation $\mathcal{T}_A$ for an entire POVM A, defined as follows
\begin{equation}\label{pio2}
\mathcal{T}_A=T_{E_{m_A}}\circ\dots\circ T_{E_1}.
\end{equation}
For any PVM $A$, operator $\mathcal{T}_A$ reduces, up to an additive term proportional to the
identity (omitted), to
\begin{equation}\label{piopvm}
\mathcal{T}_A(\rho)=\sum_{i=1}^{m_A}T_{E_i}(\rho),
\end{equation}
what follows from considering \eqref{pio2} and Prop. \ref{prop:pioproperties}. The additive property \eqref{piopvm} holding for PVM measurements plays an important role, as it helps to reduce
the runtime of our algorithm. This additive property holding for PVM measurements plays an important role, as it helps to reduce the runtime of our algorithm. Precisely, this fact allows us to apply \emph{Kaczmarz method}
\cite{Kaczmarz_1937}. Kaczmarz method considers projections over the subspace generated by the intersection of all associated hyperplanes, defined by the linear system of equations (Born's rule).

\begin{defi}[Generator state]
Given a POVM $A=\{E_i\}_{i\leq m_A}$ and a probability vector $\vec{p}\in\mathbb{R}^{m_A}$, a quantum state $\rho_{gen}$ is called \emph{generator state} for $\vec{p}$ if $\mathrm{Tr}(\rho_{gen} E_i)=p_i$, for every $i\leq m_A$ . 
\end{defi}
Note that $\rho_{gen}$ is a fixed point of $\mathcal{T}_{E_i}$, according to \eqref{def:pio} and \eqref{pio2}. State $\rho_{gen}$ plays an important role to implement numerical simulations, as it guarantees to generate sets of probability distributions compatible with the existence of a positive semidefinite solution to the quantum state tomography problem.

To conclude this section, note that map $\mathcal{T}_A$ defined in \eqref{piopvm} has a simple interpretation in the Bloch sphere for a qubit system, see Fig. \ref{orthogonal_projection}. The image of $\mathcal{T}_A$, i.e. $\mathcal{T}_A[\textrm{Herm}(\mathcal{H}_2)]$, is a plane that contains the disk $$D^{\vec{p}}_{A}=\{z=p_2-p_1\mid z=\mathrm{Tr}(\rho \sigma_z),\,p_i=\mathrm{Tr}(\rho E_i),\,\rho\geq0,\mathrm{Tr}(\rho)=1\},$$ i.e., the disk contains the full set of generator states $\rho_{gen}$. Note that $\mathcal{T}_A$ is not a completely positive trace preserving (CPTP) map, as $\mathcal{T}_A[\textrm{Herm}(\mathcal{H}_2)]$ extends beyond the disk $D^{\vec{p}}_{A}$, i.e. outside the space of states. Indeed, for any state $\rho$ that is not a convex combination of projectors $E_i$, there exists a probability distribution $\vec{p}$ such that $\mathcal{T}_A(\rho)$ is not positive semi-definite. Roughly speaking, any point inside the Bloch sphere from Fig. \ref{orthogonal_projection} but outside the blue vertical line is projected by $\mathcal{T}_A$ outside the sphere, for a sufficiently small disk $D^{\vec{p}}_A$.
\begin{figure}[h]
    \centering
      \includegraphics[scale=0.32]{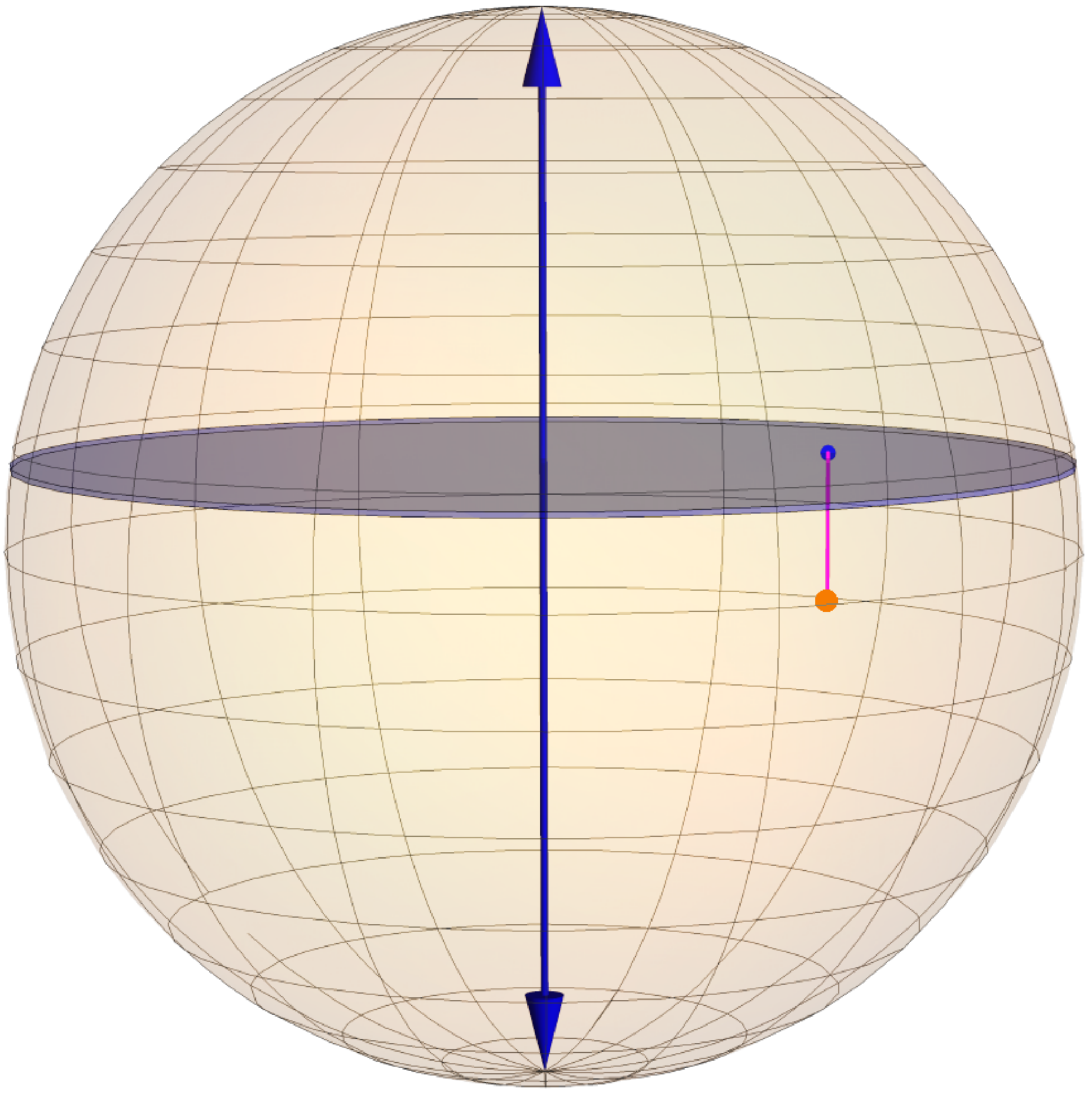}
     \caption[QSE: Orthogonal projection in the Bloch sphere]{Bloch sphere representation for a single qubit system and PVM measurements. The blue arrows define the eigenvectors of $\sigma_z$. The grey disk represents the entire set of quantum states $\rho_{gen}$ satisfying equations $p_j=\mathrm{Tr}(\rho_{gen}E_j)$, $j=0,1$, where $E_0$ and $E_1$ are rank-one eigenprojectors of an observable and $\{p_j\}$ the set of probabilities experimentally obtained. The action of $\mathcal{T}_A$ over the initial state $\rho_0$ (orange dot) is the orthogonal projection to the plane that contains the disk (blue dot) [color online].}
     \label{orthogonal_projection}
\end{figure}

\section{Algorithm for quantum state estimation}\label{sec:algorithm}

In the practice of quantum state estimation, one collects a set of probability distributions $\vec{p_1},\dots,\vec{p_{\ell}}$ from a set of $\ell$ POVM measurements $A_1,\dots,A_{\ell}$, with $A_{k}=\{E_i^{(k)}\}_{i\leq m_{k}}$, implemented over an ensemble of physical systems identically prepared in a quantum state $\rho_{gen}$. The statistics collected allows a unique state reconstruction when considering an  \emph{informationally-complete} (IC) sets of observables $A_1,\dots,A_\ell$ under the absence of sources of errors. Our algorithm for quantum state estimation, Algorithm \ref{alg:pio1} below, defines a sequence of hermitian operators $\rho_n$, not necessarily composed by quantum states, that converges to the unique quantum state that is solution to the reconstruction problem, i.e. the quantum state $\rho_{gen}$. For the moment, we assume error-free state estimation in our statements. The algorithm applies to any finite dimensional Hilbert space $\mathcal{H}$, and any informationally complete set of quantum measurements.

\begin{algorithm}[h]\caption{Quantum state estimation algorithm.}\label{alg:pio1}
\begin{algorithmic}
\Require dimension $d\in\mathbb{N}$, POVMs $A_1,\dots, A_{\ell}$ acting on $\mathcal{H}$, \\ 
experimental frequencies $\vec{f}_1,\dots,\vec{f}_{\ell}\in \mathbb{R}^m$ and accuracy $\epsilon\in [0,1]$.
\Ensure estimate $\rho_{ est}\in\mathcal{B}(\mathcal{H})$.
\State{$\rho_{0} = \mathbb{I}/d$}
\State{$\rho =  \mathcal{T}_{A_{\ell}}\circ\cdots\circ \mathcal{T}_{A_1}(\rho_{0})$}\\
\textbf{repeat} \\
\hspace*{0.5cm} $\rho_{ old} = \rho$ \\
\hspace*{0.5cm} $\rho = \mathcal{T}_{A_{\ell}}\circ\cdots\circ\mathcal{T}_{A_1}(\rho_{old})$ \\
\textbf{until} $\mathfrak{D}(\rho,\rho_{old})\leq\epsilon$ \\
\Return{ $\argmin_{ \rho_{ est} \in \mathcal{B}(\mathcal{H}) } \mathfrak{D}(\rho,\rho_{ est})$ }
\end{algorithmic}
\end{algorithm}
In Algorithm \ref{alg:pio1}, $\mathcal{B}( \mathcal{H}) $ denotes the set of density operators over $\mathcal{H}$. The presence of errors in real probabilities may lead algorithm \ref{alg:pio1} to find a matrix $\rho$ with negative eigenvalues. In that case, the last step finds the estimate state $\rho_{est}$ for which $\mathfrak{D}(\rho,\rho_{ est})$ is minimum; after finding the spectral decomposition $\rho = U diag ( \Vec{\lambda} )U^{\dagger}$, the best estimate is $\rho_{est} = U diag \left(  \left[\Vec{\lambda} - x_0 \Vec{1} \right]^+\right)U^{\dagger}$ \cite{Guta2020}, where the function $\left[\cdot\right]^+$ has components $\left[ \Vec{y} \right]_i^+ = max\{y_i,0 \}$ and $x_0$ is such that $f(x_0)=0$, with $f(x) = 2 + \mathrm{Tr}(\rho) - dx + \sum_{i=1}^d|\lambda_i - x|$.

Theorem \ref{thm:convergence} below asserts the convergence of Algorithm \ref{alg:pio1} when the input frequencies are exact, i.e. Born-rule, probabilities of an IC set of POVMs.
\begin{thm}\label{thm:convergence}
Let $A_1,\dots,A_{\ell}$ be a set of  informationally complete POVMs acting on a Hilbert space $\mathcal{H}$, associated to a set of probability distributions $\vec{p_1},\dots,\vec{p_{\ell}}$ compatible with the existence of a quantum state. Therefore, Algorithm \ref{alg:pio1} converges to the unique solution to the quantum state tomography problem.
\end{thm}
\begin{proof}
First, from item \emph{1} in Prop. \ref{prop:pioproperties} the generator state $\rho_{gen}$ is a fixed point of each imposition operator $\mathcal{T}_{A_i}$, for every chosen informationally complete POVM measurement $A_1,\dots,A_{\ell}$. Hence, $\rho_{gen}$ is a fixed point of the composition of all involved operators. Moreover, this fixed point is unique, as there is no other quantum state having the same probability distributions for the considered measurements, as $A_1,\dots,A_{\ell}$ are informationally complete. Here, we are assuming error-free probability distributions. Finally, convergence of our sequences is guaranteed by Theorem 1 in \cite{halpern_1967}, which states that successive iterations of non-expansive projections converge to a common fixed point of the involved maps.
\end{proof}
Here, compatibility refers to the existence of a quantum state associated to exact probability distributions $\vec{p_1},\dots,\vec{p_{\ell}}$ what is guaranteed when probabilities come from a generator state $\rho_{gen}$. Theorem \ref{thm:convergence} asserts, in other words, that the composite map $\mathcal{T}_{A_{\ell}}\circ\cdots\circ \mathcal{T}_{A_1}$ defines a dynamical system having a unique attractive fixed point. The successive iterations of Algorithm \ref{alg:pio1} define a \emph{Picard sequence} \cite{kuczma_choczewski_ger_1990}:
\begin{align}\label{eq:picard-sequence}
    \rho_0&=\mathbb{I}/d,\nonumber\\
    \rho_n&= \mathcal{T}_{A_{\ell}}\circ\dots\circ \mathcal{T}_{A_1}(\rho_{n-1}),~n\geq 1.
\end{align} 
Note that for arbitrarily chosen set of measurements, the composition of physical imposition operators depends on its ordering. According to Theorem \ref{thm:convergence}, this ordering does not affect the success of the convergence in infinitely many steps. However, in practice one is restricted to a finite sequence, where different orderings produce different quantum states as an output. Nonetheless, such difference tends to zero when the state $\rho_n$ is close to the attractive fixed point, i.e. solution to the state tomography problem. According to our experience from numerical simulations, we did not find any advantage from considering a special ordering for composition of operators.

Figure \ref{Fig2} shows the convergence of $\rho_n$ in the Bloch sphere representation for a single qubit system and three PVMs taken at random. For certain families of measurements, e.g. mutually unbiased bases and tensor product of Pauli matrices, the resulting Picard sequences and, therefore, Algorithm \ref{alg:pio1} converges in a single iteration, see Prop. \ref{prop_singlestep}. That is, $\rho_n = \rho_{1}$ for every $n\geq 1$. We numerically observed this same behaviour for the $3^N$ product Pauli eigenbases in the space of $N$-qubits, with $1\leq N\leq 8$, conjecturing that it holds for every $N\in\mathbb{N}$, see Section \ref{sec:simulations-pauli}. 

\begin{figure}[t]
    \centering
     \includegraphics[scale=0.32]{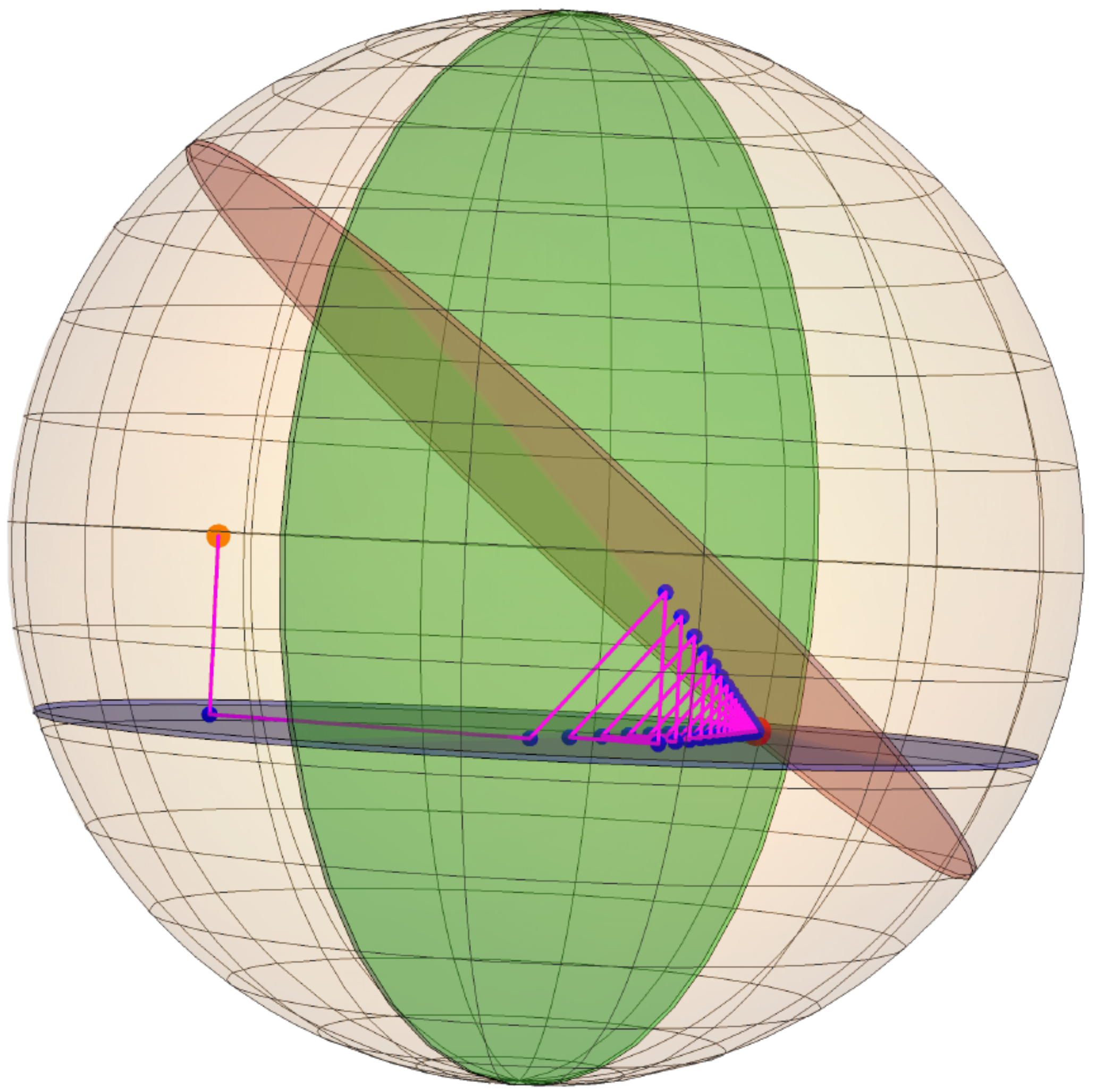}
     \caption[QSE: Convergence of Algorithm \ref{alg:pio1} in the Bloch sphere.]{Graphical representation of the convergence of Algorithm \ref{alg:pio1} in the Bloch sphere for a single qubit system. We show convergence for three incompatible PVMs $A_1,A_2$ and $A_3$, defining disks $D_1$ (grey), $D_2$ (green) and $D_3$ (red) in the Bloch sphere. The initial state $\rho_0$  (orange dot), which we have chosen different from $\mathbb{I}/2$ only for graphical purposes, is first projected to $D_1$. The corresponding point in $D_1$ is then projected to $D_2$ and that projection is later projected to $D_3$. The iteration of this sequence of projections successfully converges to the generator state $\rho_{gen}$ (red dot), the unique solution to the quantum state tomography problem [color online].} 
     \label{Fig2}
\end{figure}


\subsection{Single iteration convergence}\label{sec:ultra-convergence}

When considering maximal sets of mutually unbiased bases, the Picard sequences featuring in Algorithm \ref{alg:pio1} converge in a single iteration. This is so because the associated imposition operators commute for MUB. This single-iteration convergence is easy to visualize in the Bloch sphere (see Fig. \ref{pio_mubs}) for a qubit system, as the three disks associated to three MUB are mutually orthogonal, and orthogonal projections acting over orthogonal planes keep the impositions within the intersection of the disks. The same argument also holds in every dimension. Let us formalize this result.
\begin{figure}[t]
    \centering
     \includegraphics[scale=0.35]{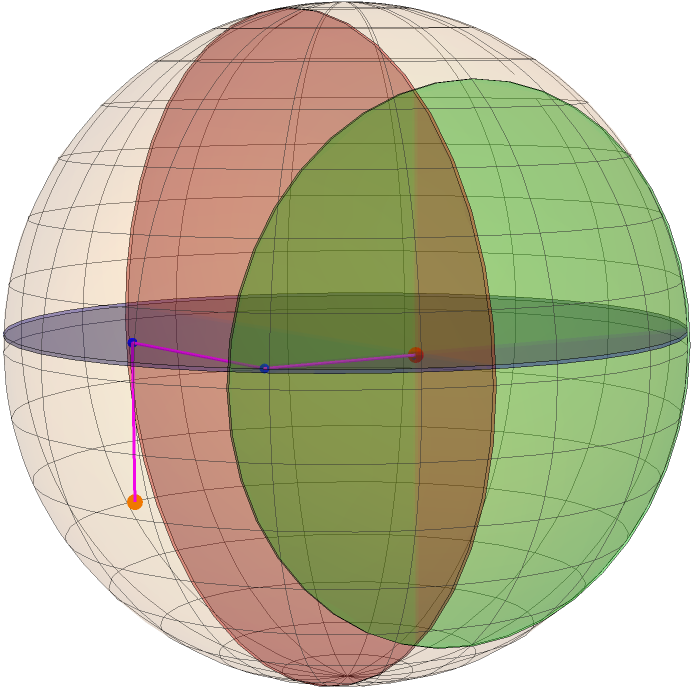}
     \caption[QSE: Convergence of the algorithm \ref{alg:pio1} for MUBs]{Here we show the convergence of the algorithm when MUBs are considered. We see that it only requires to project once on each plane to converge to the solution.} 
     \label{pio_mubs}
\end{figure}
\begin{prop}\label{MUBcommute}
Let $T_A$ and $T_B$ be two physical imposition operators associated to two mutually unbiased bases $A$ and $B$. Therefore,
\begin{equation}
T_B\circ T_A=T_A\circ T_B=T_A+T_B-\mathbb{I},
\end{equation}
with $\mathbb{I}$ the identity operator. In particular, note that $T_A$ and $T_B$ commute.
\end{prop}
\begin{proof}
First, it is simple to show that  $\mathcal{T}_A(\rho)=\rho_0+\sum_{j=0}^{m_A-1}\Pi_j(\rho-\rho_0)\Pi_j$ for any PVM $A$, where $\Pi_j=E_j$ are the subnormalized rank-one PVM elements\footnote{We thank Felix Huber for highlighting this property}. Thus, we have
\begin{align}
\begin{split}
    T_B\circ T_A(\rho_0) ={}& \rho_0 + \sum_{j=0}^{m_A-1}\Pi^A_j(\rho-\rho_0)\Pi^A_j \\
          & \quad \quad + \sum_{k=0}^{m_B-1}\Pi^B_k\left[\rho-\left(\rho_0+\sum_{j=0}^{m_A-1}\Pi^A_j(\rho-\rho_0)\Pi^A_j\right)\right]\Pi^B_k, \nonumber
\end{split}  \\
\begin{split}
     ={}& \rho_0+\sum_{j=0}^{m_A-1}\Pi^A_j(\rho-\rho_0)\Pi^A_j+\sum_{k=0}^{m_B-1}\Pi^B_k(\rho-\rho_0)\Pi^B_k\\
        &  \quad \quad + \sum_{j,k}\Pi^B_k\Pi^A_j(\rho-\rho_0)\Pi^A_j\Pi^B_k. \nonumber
\end{split}
\end{align}
On the other hand, from item 2 in Prop. \ref{prop:pioproperties}
\begin{eqnarray*}
\sum_{j,k}\Pi^B_k\Pi^A_j(\rho-\rho_0)\Pi^A_j\Pi^B_k&=&\sum_{j,k}\mathrm{Tr}(\Pi^A_j\Pi^B_k)\mathrm{Tr}\bigl((\rho-\rho_0)\Pi^A_j\bigr)\Pi^B_k,\\
&=&\gamma(A,B)\sum_{j,k}\mathrm{Tr}\bigl((\rho-\rho_0)\Pi^A_j\bigr)\Pi^B_k,\\
&=&\gamma(A,B)\mathrm{Tr}(\rho-\rho_0) \mathbb{1} = 0.
\end{eqnarray*}
Therefore, we have
\begin{eqnarray}
T_B\circ T_A(\rho_0)&=&\rho_0+\sum_{j=0}^{m_A-1}\Pi^A_j(\rho-\rho_0)\Pi^A_j+\sum_{k=0}^{m_B-1}\Pi^B_k(\rho-\rho_0)\Pi^B_k, \nonumber \\
&=&T_A(\rho_0)+T_B(\rho_0)-\rho_0,
\end{eqnarray}
for any initial state $\rho_0$. So, we have 
$T_B\circ T_A=T_A\circ T_B=T_A+T_B-\mathbb{I}$.
\end{proof}

Also, it is easy to see from Item \emph{2}, Prop. \ref{prop:pioproperties} that operators $T_{E_i}$ commute when considering $E_i$ equal to the tensor product local Pauli group. In this case, operators $E_i$ do not form a POVM but given that they define an orthogonal basis in the matrix space, they are an informationaly complete set of observables. Let us now show the main result of this section:
\begin{prop}\label{prop_singlestep}
Algorithm \ref{alg:pio1} converges, in a single iteration, to the unique solution of the quantum state tomography problem for  product of generalized Pauli operators and also for $d+1$ mutually unbiased bases, in any prime power dimension $d$.
\end{prop}
\begin{proof}
For generalized Pauli operators, commutativity of imposition operators comes from orthogonality condition $\mathrm{Tr}(E_iE_j)$, see item \emph{2} in Prop. \ref{prop:pioproperties}. Thus, we have
\begin{eqnarray}\label{seq}
\rho_n&=&(T_{E_{d^2}}\circ\cdots\circ T_{E_1})^n(\rho_0),\nonumber\\
&=&T^n_{E_{d^2}}\circ\cdots\circ T^n_{E_1}(\rho_0),\nonumber\\
&=&T_{E_{d^2}}\circ\cdots\circ T_{E_1}(\rho_0),
\end{eqnarray}
where the second step considers commutativity and the last step the fact that every $T_j$, $j=1,\dots,d+1$ is a projection. On the other hand, from Theorem \ref{thm:convergence} we know that $\rho_n\rightarrow\rho_{gen}$ when $n\rightarrow\infty$, for any generator state $\rho_{gen}$. From combining this result with (\ref{seq}) we have 
\begin{equation}
T_{E_{d^2}}\circ\cdots\circ T_{E_1}(\rho_0)=\rho_{gen},
\end{equation}
for any seed $\rho_0$ and any generator state $\rho_{gen}$, in any prime power dimension $d$. For MUB the result holds in the same way, where commutativity between the associated imposition operators associated to every PVM  arises from, see Prop. \ref{MUBcommute}. 
\end{proof}
We observe from simulations that the speedup predicted by Prop. \ref{prop_singlestep} has no consequences in the reconstruction fidelity of our method, which is actually higher than the one provided by MLE.

\vspace{-2.7mm}
\section{Numerical study}\label{sec:simulations}
Theoretical developments from Sections \ref{sec:pio} and \ref{sec:algorithm} apply to the ideal case of error free probabilities coming from an exact generator state $\rho_{gen}$. In practice, probabilities are estimated from frequencies, carrying errors due to finite statistics. Moreover, the states being prepared in each repetition of the experiment are affected by unavoidable systematic errors. These sources of errors imply that the output of Algorithm \ref{alg:pio1} is typically outside the set of quantum states when considering experimental data. We cope with this situation by finding the closest quantum state to the output, called $\rho_{est}$ in Hilbert-Schmidt (a.k.a. Frobenius) distance , for which there are closed-form expressions \cite{Guta2020}. 
In the following, we provide numerical evidence for robustness of our method in the finite-statistics regime with white noise affecting the generator states, i.e. errors at the preparation and measurement stages. That is, we consider noisy states of the form $\tilde{\rho}(\lambda)=(1-\lambda)\rho+\lambda\mathbb{I}/d$, where $\lambda$ quantifies the amount of errors. We understand there are more sophisticated techniques to consider errors, e.g. ill-conditioned measurement matrices \cite{Bolduc2017}. Nonetheless, we believe the consideration of another model to simulate a small amount of errors would not substantially change the exhibited results.
We reconstruct the state for $N$-qubit systems with $1\leq N\leq 8$, by considering the following sets of measurements: a) Mutually unbiased bases, b) Tensor product of local Pauli bases and c) A set of $d+1$ informationally complete bases taken at random with Haar distribution. The last case does not have a physical relevance but illustrates the performance of our algorithm for a set of measurements defined in an unbiased way. As a benchmark, we compare the performance of our method with the conjugate gradient, accelerated-gradient-descent (CG-AGP) implementation of Maximum Likelihood Estimation (MLE) \cite{Superfast2007}. Computations were conducted on an Intel core i5-8265U laptop with 8gb RAM. For the CG-AGP algorithm, we used the implementation provided by authors of Ref. \cite{Superfast2007}, see Ref. \cite{qMLE_gitub}. We provide an implementation of our Algorithm \ref{alg:pio1} in Python \cite{pio_gitub}, see Appendix  \hyperref[app:A]{A}, together with the code to run the simulations presented in the current section.

\subsection{Mutually unbiased bases}\label{sec:simulations-MUB}
Figure \ref{fig:MUB} shows performance of Algorithm \ref{alg:pio1} in the reconstruction of $N$-qubit density matrices from the statistics of a maximal set of $2^N+1$ MUBs. We consider a generator state $\rho_{gen}$ in dimension $d$, taken at random according to the Haar measure distribution, with the addition of a $10\%$ level of white noise, i.e.
$\tilde{\rho}(\lambda)=(1-\lambda)\rho_{gen}+\lambda\mathbb{I}/2^N$, with $\lambda=0.1$. Here, it is important to remark that fidelities are compared with respect to the generator state $\rho_{gen}$, so that the additional white noise reflects the presence of systematic errors in the state preparation process. Probabilities are estimated from frequencies, i.e. $f_j=\mathcal{N}_j/\mathcal{N}$ with $\mathcal{N}_j$ the number of counts for outcome $j$ of some POVM and $\mathcal{N}=\sum_j \mathcal{N}_j$ the total number of counts. Our simulations consider $\mathcal{N}=100\times 2^N$ samples per measurement basis. Our figure of merit is the fidelity $F(\rho_n,\rho_{gen})=\mathrm{Tr}{[\sqrt{\sqrt{\rho_{gen}}\rho_n\sqrt{\rho_{gen}}}]}^2$ between the reconstructed state after $n$ iterations $\rho_n$ and the generator state $\rho_{gen}$. Runtime of the algorithm is averaged over 50 independent runs, each of them considering a generator state $\rho_{gen}$ chosen at random according to the Haar measure.

\begin{figure}[h!]
    \centering
     \subfloat[\label{qst_performance_a}]{%
       \includegraphics[scale=0.85]{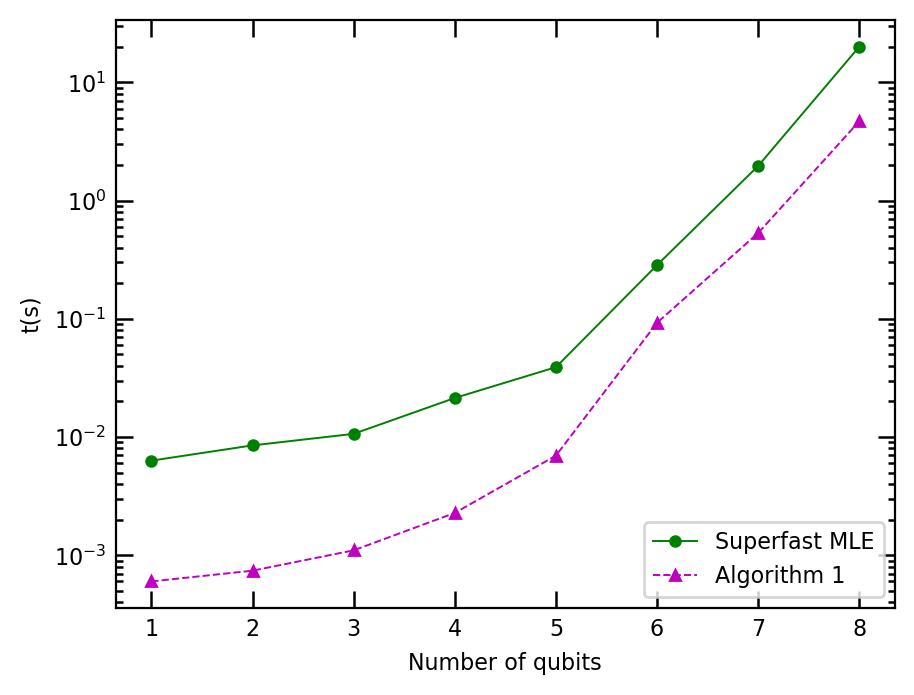}
     }\hspace{-0.15cm}
     \subfloat[\label{qst_performance_b}]{%
       \includegraphics[scale=0.85]{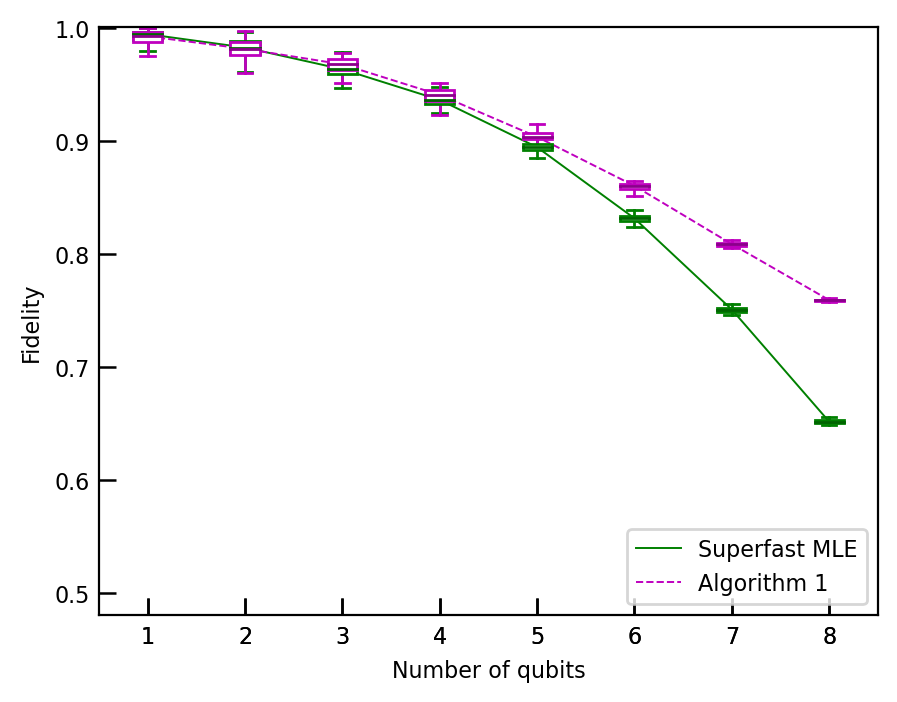}
     }
     \caption[QSE: Performance of Algorithm \ref{alg:pio1} for $d+1=2^N+1$ MUB.]{Performance of Algorithm \ref{alg:pio1} and the CG-AGP Superfast MLE method from \cite{Superfast2007}, for the reconstruction of $N$-qubit states from a maximal set of $d+1=2^N+1$ mutually unbiased basis (MUB) in dimension $d=2^N$. Generator state $\rho$ is chosen at random by considering the Haar measure distribution, subjected to $10\%$ of white noise and finite statistics satisfying Poissonian distribution. For simulations, we consider  $100\times2^N$ samples. Fig. \ref{qst_performance_a} considers runtime of the algorithm in seconds, averaged over 50 trials, whereas \ref{qst_performance_b} shows bloxplots for the 50 trials datasets of the fidelity between the target and obtained state, the lines cross the medians. It is worth to mention that our simulations were performed in Python whereas Ref. \cite{Superfast2007} considers Matlab.}
     \label{fig:MUB}
\end{figure}

Table \ref{tab:fidelitiesMubs} shows errors in the fidelity of both algorithm \ref{alg:pio1} and the Superfast algorithm for different number of qubits. To determine the standard errors we used the \textit{bootstrapping} technique \cite{EfroTibs93}.
Here, we apply 50 trials of Monte Carlo simulations using the bootstrap countings, which are distributed according to the probability distributions obtained from $\tilde{\rho}(\lambda)$. Then, the state is estimated from the simulated countings and the fidelity with respect to $\rho_{gen}$ is calculated. 

\begin{table}[h!]
\centering
    \caption{Errors in the estimation of both algorithm 1 and the Superfast algorithm when considering MUB. $\Bar{F}$ and $\delta F$ are the mean value and the standard error of the fidelity, respectively.}
   \vspace{-6mm}
\begin{center}
\begin{tabular}{|| c | c | c ||} 
 \hline
  Number of qubits & $\Bar{F} \pm \delta F$ (Algorithm 1) & $\Bar{F} \pm \delta F$ (Superfast) \\ [0.5ex] 
\hline\hline
1 & $0.9898 \pm 0.0009$ & $0.9921 \pm 0.0007$ \\
2 & $0.9656 \pm 0.0014$ & $0.9703 \pm 0.0018$ \\
3 & $0.9544 \pm 0.0012$ & $0.9617 \pm 0.0011$ \\
4 & $0.9004 \pm 0.0010$ & $0.9361 \pm 0.0010$ \\
5 & $0.8611 \pm 0.0007$ & $0.8949 \pm 0.0006$ \\
6 & $0.8062 \pm 0.0005$ & $0.8340 \pm 0.0005$ \\
7 & $0.7587 \pm 0.0003$ & $0.7505 \pm 0.0004$ \\
8 & $0.7201 \pm 0.0001$ & $0.6504 \pm 0.0002$ \\
[1ex] 
 \hline
\end{tabular}
\end{center}
\label{tab:fidelitiesMubs}
\end{table}

\clearpage

\subsection{$N$-qubit Pauli bases}\label{sec:simulations-pauli}
Here, we consider the reconstruction of $N$-qubit density matrices from the $3^N$ PVMs determined by all possible combinations of the the tensor product of single qubit Pauli eigenbases, for $N=1,\dots,8$, which are informationally complete. Similarly to the case of MUBs, Picard sequences $\rho_n=T_{Pauli}^n(\rho_0)$ converge in a single iteration when product of Pauli measurements are considered, for any generator state $\rho_{gen}$ and any initial state $\rho_0$. Figure \ref{fig:pauli} shows performance of a single iteration of these Picard sequences, where the generator state $\rho_{gen}$ is taken at random, according to the Haar measure. Algorithm CG-AGP exploits the product structure of the $N$-qubit Pauli bases to speedup its most computationally expensive part: the computation of the probabilities given by the successive estimates in the MLE optimization. It does so by working with reduced density matrices which, in turn, imply an efficient use of memory. In order to have a fair comparison with our method, we decided to include the time to compute the $N$-qubit observables from the single Pauli observables in the total runtime of our algorithm. In practice, however, one would preload them into memory, as they are, of course, not a function of the input, i.e. of the observed probabilities.

Errors in the fidelity of both algorithm \ref{alg:pio1} and the Superfast, when using Pauli bases, are shown in table \ref{tab:fidelitiesPaulis}. The errors were determined using the bootstrapping technique (see section \ref{sec:simulations-MUB} ).

\begin{table}[h!]
\centering
   \caption[Error analysis]{ Errors in the estimation of both algorithm 1 and the Superfast algorithm when considering Pauli bases. $\Bar{F}$ and $\delta F$ are the mean value and the standard error of the fidelity, respectively}
   \vspace{-5mm}
\begin{center}
\begin{tabular}{|| c | c | c ||} 
 \hline
  Number of qubits & $\Bar{F} \pm \delta F$ (Algorithm 1) & $\Bar{F} \pm \delta F$ (Superfast) \\ [0.5ex] 
\hline\hline
1 & $0.9761 \pm 0.0022$ & $0.9783 \pm 0.0021$ \\
2 & $0.9792 \pm 0.0010$ & $0.9871 \pm 0.0008$ \\
3 & $0.9528 \pm 0.0010$ & $0.9737 \pm 0.0006$ \\
4 & $0.9246 \pm 0.0007$ & $0.9632 \pm 0.0007$ \\
5 & $0.8913 \pm 0.0004$ & $0.9476 \pm 0.0005$ \\
6 & $0.8524 \pm 0.0002$ & $0.9270 \pm 0.0003$ \\
7 & $0.8066 \pm 0.0001$ & $0.9017 \pm 0.0002$ \\
8 & $0.7712 \pm 0.0001$ & $0.8708 \pm 0.0001$ \\
[1ex] 
 \hline
\end{tabular}
\end{center}
\label{tab:fidelitiesPaulis}
\end{table}

\begin{figure}[h!]
    \centering
     \subfloat[\label{fig:pauli-a}]{%
       \includegraphics[scale=0.85]{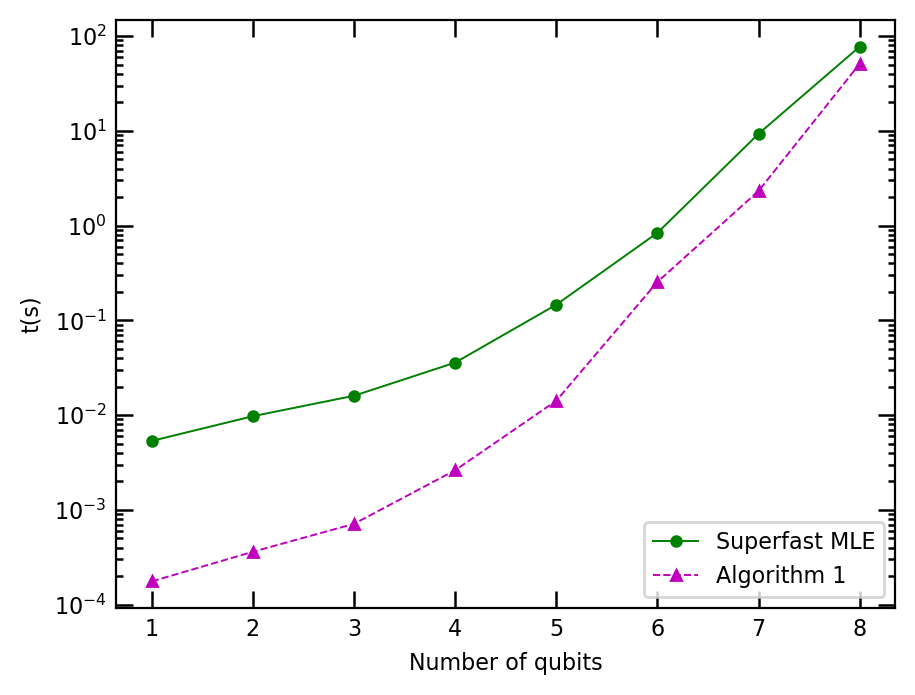}
     }\hspace{-0.15cm}
     \subfloat[\label{fig:pauli-b}]{%
       \includegraphics[scale=0.85]{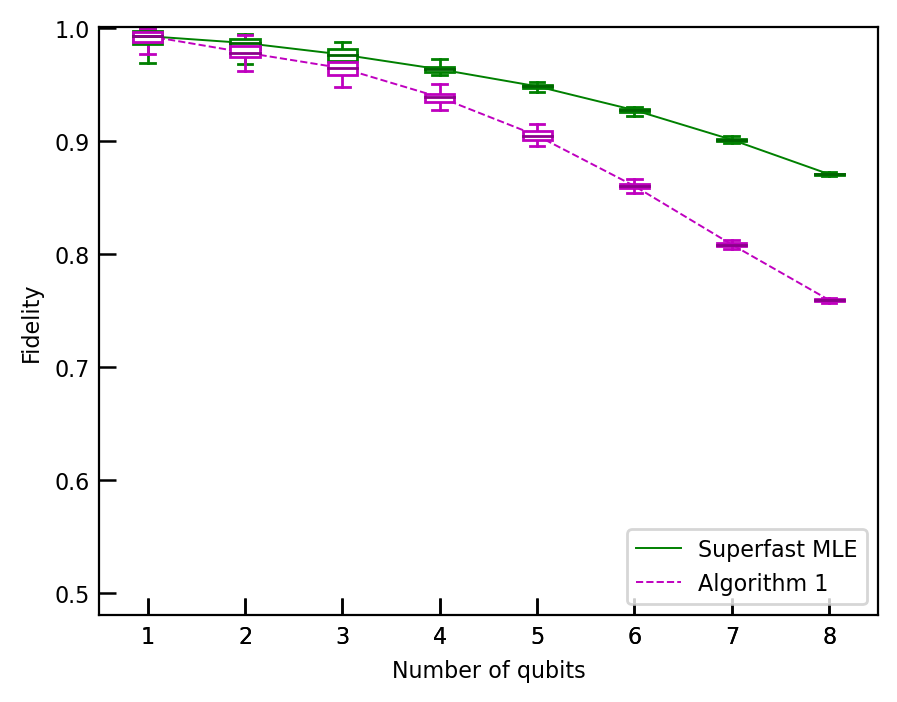}
     }
     \caption[QSE: Performance of Algorithm \ref{alg:pio1} for local Pauli observables]{Performance of Algorithm \ref{alg:pio1} and the CG-AGP Superfast MLE \cite{Superfast2007}, for the reconstruction of $N$-qubit states from $3^N$ PVM given by products of the eigenbases of local Pauli observables $\sigma_X$, $\sigma_Y$ and $\sigma_Z$. Generator states $\rho$ are chosen at random (Haar measure), subjected to $10\%$ of white noise and finite statistics satisfying Poissonian distribution, considering  $100\times2^N$ samples per PVM. Fig. \ref{fig:pauli-a} considers runtime of the algorithm in seconds, whereas \ref{fig:pauli-b} shows bloxplots for the 50 trials datasets of the fidelity between the target and obtained state, the lines cross the medians. We consider simulations in Python, whereas Ref. \cite{Superfast2007} considers Matlab.}
     \label{fig:pauli}
\end{figure}



\clearpage
\subsection{Random measurements for $N$-qubit systems}\label{sec:simulations-random}
The simulations in the preceding subsections correspond to informationally complete sets of measurements for which Algorithm \ref{alg:pio1} converges in a single iteration. To test whether the advantage over \cite{Superfast2007} hinges critically on this fact, we have numerically tested our algorithm with sets of PVMs selected at random, with respect to the Haar measure. In Fig. \ref{fig:random-bases} we show that in this case, the advantage in fidelity increases substantially, compared to Figs. \ref{fig:MUB} and \ref{fig:pauli}. Table \ref{tab:fidelitiesPaulis} shows the errors in the fidelity of both algorithm \ref{alg:pio1} and the Superfast for the case of random measurements.

\vspace{-1.5mm}
\begin{table}[h!]
\centering
   \caption[Error analysis]{ Errors in the estimation of both algorithm 1 and the Superfast algorithm when considering random bases. $\Bar{F}$ and $\delta F$ are the mean value and the standard error of the fidelity, respectively.}
   \vspace{-5mm}
\begin{center}
\begin{tabular}{|| c | c | c ||} 
 \hline
  Number of qubits & $\Bar{F} \pm \delta F$ (Algorithm 1) & $\Bar{F} \pm \delta F$ (Superfast) \\ [0.5ex] 
\hline\hline
1 & $0.9744 \pm 0.0017$ & $0.9871 \pm 0.0017$ \\
2 & $0.9050 \pm 0.0021$ & $0.9037 \pm 0.0082$ \\
3 & $0.8641 \pm 0.0016$ & $0.8793 \pm 0.0050$ \\
4 & $0.8719 \pm 0.0010$ & $0.8459 \pm 0.0019$ \\
5 & $0.8391 \pm 0.0007$ & $0.7973 \pm 0.0014$ \\
6 & $0.7997 \pm 0.0003$ & $0.7352 \pm 0.0007$ \\
7 & $0.7628 \pm 0.0002$ & $0.6605 \pm 0.0005$ \\
8 & $0.7331 \pm 0.0001$ & $0.5790 \pm 0.0002$ \\
[1ex] 
 \hline
\end{tabular}
\end{center}
\label{tab:fidelitiesRandom}
\end{table}

\vspace{-5mm}
Finally, we would like to mention the \emph{Projective Least Squares} (PLS) quantum state reconstruction \cite{Guta2020}. This method outperforms both in runtime and fidelity our Algorithm \ref{alg:pio1}. This occurs when the linear inversion procedure required by the method \emph{is not} solved but taken from analytically existing reconstruction formula. Existing inversion formulas are known for complex projective 2-designs, measurement composed by stabilizer states, Pauli observables and uniform/covariant POVM, see \cite{Guta2020}. However, when taking into account the cost of solving the linear inversion procedure, our method has a considerable advantage over PLS. For instance, PLS does not have such efficient speed up for a number of physically relevant observables for which there is no explicit inversion known, including the following cases: a) discrete Wigner functions reconstruction for arbitrary dimensional bosons and fermions quantum systems from discrete quadratures, that can be treated as observables by considering Ramsey techniques \cite{Leonhardt1995}, b) reconstruction of single quantized cavity mode from magnetic dipole measurements with Stern-Gerlach aparatus \cite{Walser1996}, c) minimal 




\begin{figure}[h!]
    \centering
     \subfloat[\label{fig:random-a}]{%
       \includegraphics[scale=0.85]{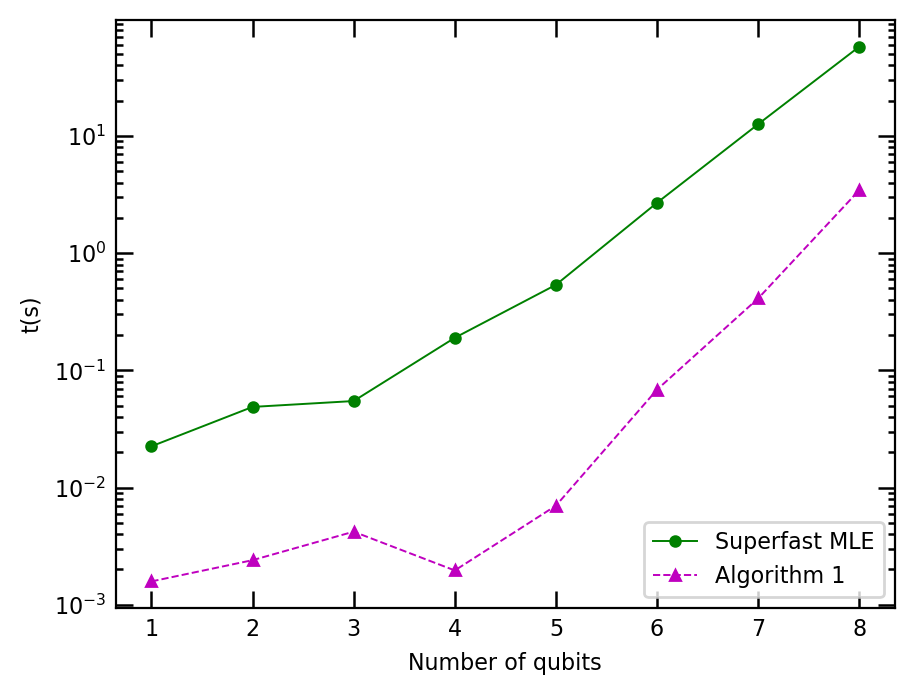}
     }\hspace{-0.15cm}
     \subfloat[\label{fig:random-b}]{%
       \includegraphics[scale=0.85]{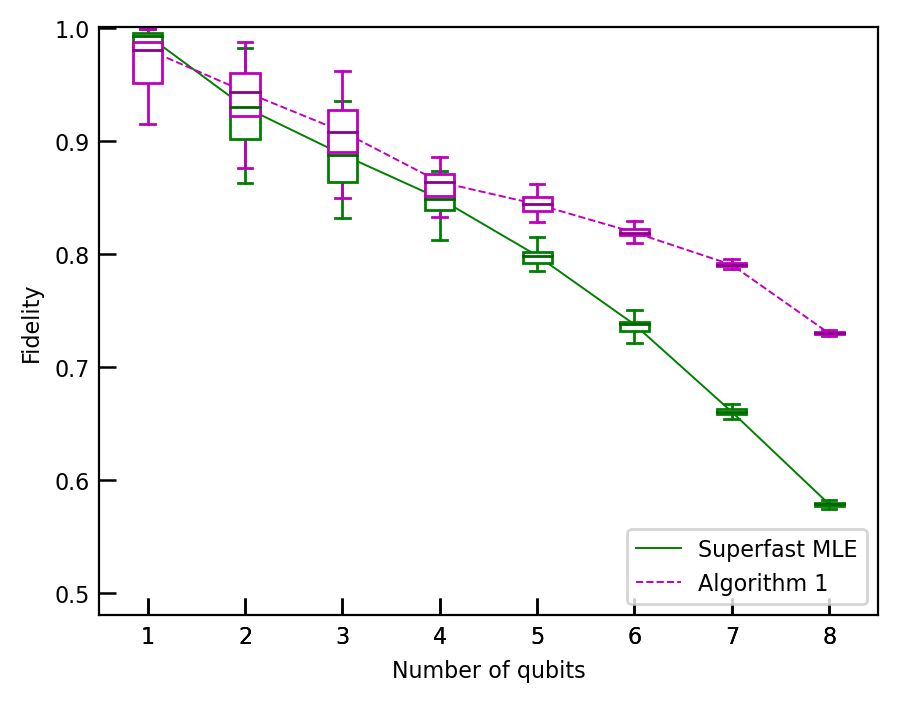}
     }
     \caption[QSE: Performance of Algorithm \ref{alg:pio1} for $d+1=2^N+1$ basis chosen Haar-random]{Performance of Algorithm \ref{alg:pio1} and the CG-AGP Superfast MLE method from \cite{Superfast2007} for the reconstruction of $N$-qubit states from a a set of $d+1=2^N+1$ basis chosen Haar-random in dimension $d=2^N$. Generator state $\rho$ is chosen under the Hilbert-Schmidt measure, subjected to $10\%$ of white noise. Measurement statistics are estimated from $\mathcal{N}=100\times2
    ^N$ identical copies. Fig. \ref{fig:random-a} considers the runtime of the algorithm in seconds, averaged over 50 trials, whereas \ref{fig:random-b} shows bloxplots for the 50 trials datasets of the fidelity between the target and obtained state, the lines cross the medians.}
     \label{fig:random-bases}
\end{figure}

\clearpage
\noindent state reconstruction of $d$-dimensional quantum systems from POVM consisting on $d^2$ elements, inequivalent to SIC-POVM \cite{Weigert2006}, d) spin $s$ density matrix state reconstruction from Stern-Gerlach measurements \cite{Amiet_1999}, e) Quantum state tomography for multiparticle spin $1/2$ systems \cite{D_Ariano_2003}, neither reduced to mutually unbiased bases nor local Pauli measurements.

For both algorithm 1 and the CG-AGP algorithm, considering any of the measurements studied in this section, fidelity decreases when the amount of white noise is increased, as expected.

\section{Discussion and conclusions}\label{sec:conclusions}
We introduced an iterative method for quantum state estimation of density matrices from any informationally complete set of quantum measurements in any finite dimensional Hilbert space. We demonstrated convergence to the unique solution for any informationally complete or overcomplete set of POVMs, see Theorem \ref{thm:convergence}. The method, based on dynamical systems theory, exhibited a simple and intuitive geometrical interpretation in the Bloch sphere for a single qubit system, see Figs. \ref{orthogonal_projection} and \ref{Fig2}. Algorithm \ref{alg:pio1} revealed a fast convergence for a wide class of measurements, including mutually unbiased bases and tensor product of generalized Pauli observables for an arbitrary large number of particles having $d$ internal levels. Furthermore, numerical simulations revealed strong robustness under the presence of realistic errors in both state preparation and measurement stages, see Figs. \ref{fig:MUB} to \ref{fig:random-bases}. We provided an easy to use code developed in Python to implement our algorithm, see \cite{qMLE_gitub}.

As interesting future lines of research, we pose the following list of open issues: (\emph{i}) Find an upper bound for fidelity reconstruction of Algorithm \ref{alg:pio1} as a function of errors and number of iterations; (\emph{ii}) Characterize the full set of quantum measurements for which Algorithm \ref{alg:pio1} converges in a single iteration; (\emph{iii}) Extend our method to quantum process tomography.\bigskip

    \chapter{Quantum non-locality}\label{cap3}
    Certification of quantum nonlocality plays a central role in practical applications like device-independent quantum cryptography and random number generation protocols. In this chapter, we introduce a technique to find a Bell inequality with the largest possible gap between the quantum prediction and the classical local hidden variable limit for a given set of  measurement frequencies. Our method represents an efficient strategy to certify quantum nonlocal correlations from experimental data without requiring extra measurements, in the sense that there is no Bell inequality with a larger gap than the one resulting from our method. Furthermore, we also reduce the photodetector efficiency required to close the detection loophole. We illustrate our technique by improving the detection of quantum nonlocality from experimental data obtained with weakly entangled photons.

\section{Introduction} 
Quantum nonlocality plays a fundamental role in flourishing quantum technologies, such as device \cite{Mayers1998,Acin2007} and semi-device \cite{Pawlowski2011} independent quantum cryptography, device-independent quantum secure direct communication against collective attacks \cite{Zhou2020}, and genuine random number generation \cite{Pironio2010}, as well as fundamental aspects of quantum physics. In these applications, certification of nonlocality is typically required. In the ideal case, for any set of nonlocal correlations, there exists a Bell inequality that is violated. However,  certification of nonlocality can be hard to achieve in practice due to the presence of experimental errors. This is especially true when the optimal quantum state, i.e. the state producing the maximal violation of a given Bell inequality, is weakly entangled \cite{Acin2012}. This problem plays a relevant role even when considering tight Bell inequalities, as also these inequalities might be maximally violated by partially entangled quantum states \cite{Collins2002,Vidick2011}.  Tight Bell inequalities are particularly useful as they are known to maximize the randomness that can be certified in a Bell scenario. For instance, a recent experiment deal with quantum nonlocality certification by using near-ideal two-qubit states for weakly entangled quantum systems \cite{Gomez2019}.  

Bell inequalities can be used to certify that a set of correlations cannot be described by a local hidden variable (LHV) model. Some bipartite Bell inequalities can have a large ratio between the quantum and LHV limits, equal to  $\sqrt{n}/\log{n}$, for $n$ settings and $n$ outputs in $n$ dimensional Hilbert spaces 
\cite{Junge2011LargeVO}. From the experimental perspective, a larger theoretical violation increases the chance to certify quantum nonlocality in the laboratory. Nonetheless, sometimes experiments are not conclusive to certify nonlocality. Under such situation, one can simply choose another Bell inequality with a larger gap between the LHV and quantum values, thus increasing the chances for success. However, this change typically involve a new experimental implementation, as the optimal settings of the new Bell inequality most likely differ from the original one. This procedure consumes additional time and resources in the laboratory. Thus, a fundamental question arises:
\begin{quote}
\emph{Can we certify quantum nonlocality from experimental data that previously failed to violate a target Bell inequality?}    
\end{quote}
To start studying the problem, an asymptotically optimal analysis can be done to reject local realism of a given statistical data \cite{Zhang2011,Zhang2013,Liang2019}. In this work, we find necessary and sufficient conditions to provide a conclusive answer to the above question, for any bipartite scenario. In addition, we present an optimization method that finds a Bell inequality that maximizes the chances to detect quantum nonlocality, among the entire set of inequalities of a given scenario. This method is particularly useful to certify nonlocality when considering weakly entangled quantum states. For instance, we successfully certify nonlocality for quantum states having smaller concurrence than those studied in a recent work \cite{Christensen2015}. Our technique finds a wide range of practical applications including communication complexity problems, where the advantage in communication is an increasing function of the quantum-LHV value gap \cite{Brukner2004,Tavakoli2020}.

\section{Method}\label{sec:optimization}
In this section, we introduce a method that provides the largest possible gap between the quantum and LHV predictions for any given set of experimental data. In particular, this procedure allows us to determine whether a set of experimental data is genuinely nonlocal or not, i.e. whether there is a Bell inequality that can certify quantum nonlocality from the noisy data. Our method represents an efficient certification of nonlocal correlations, that can be applied to experimental data without requiring extra measurements. In other words, we produce a Bell inequality that maximizes the chances to detect quantum nonlocality from a given set of statistical data among the entire set of Bell inequalities.

Let us express Eq. \eqref{BellInequality} in the following form
\begin{equation}\label{Bellineq}
\sum_{x,y=0}^{m-1}\sum_{a,b=0}^{d-1} s^{ab}_{xy}\,p(a,b|x,y)\leq\mathcal{C},
\end{equation}
where $p(a,b|x,y)$ is the probability of obtaining outcomes $a,b\in\{0,\dots,d-1\}$ when inputs $x,y\in\{0,\dots,m-1\}$ are chosen by two observers, Alice and Bob, respectively.  Here, $\mathcal{C}$ denotes the maximal value of the left-hand side in \eqref{Bellineq} that can be achieved by considering \emph{local hidden variable} (LHV) theories, whereas quantum mechanics might predict a violation of this inequality  \cite{bellTheorem}. Without loss of generality, we can restrict our attention to parameters within the set $-1\leq s^{ab}_{xy}\leq1$, for every $a,b=0,\dots,d-1$ and $x,y=0,\dots,m-1$. In order to obtain the LHV value $\mathcal{C}$, we have to optimize the left hand side in \eqref{Bellineq} over all possible local deterministic strategies. Here, locality means statistical independence of simultaneous and distant events, i.e. $p(a,b|x,y)=p(a|x)p(b|y)$, and deterministic means that all probabilities take values $0$ or $1$ only, restricted to the normalization conditions.

Quantum joint probability distributions satisfy the no-signaling principle, i.e. information cannot be instantaneously transmitted between distant parties. In particular, the outcome of one party cannot reveal information about the input of the other. That is expressed in the no-signaling conditions \eqref{no-signaling}, that we rewrite in the following form
\begin{equation}\label{ns1}
\sum_{b=0}^{d-1}p(a,b|x,y)=\sum_{b=0}^{d-1}p(a,b|x,y')=:p_A(a|x),
\end{equation}
and
\begin{equation}\label{ns2}
\sum_{a=0}^{d-1}p(a,b|x,y)=\sum_{a=0}^{d-1}p(a,b|x',y)=:p_B(b|y),
\end{equation}
for every $x\neq x'$ and $y\neq y'$, where $p_A(a|x)$ and $p_B(b|y)$ are the marginal probability distributions associated to Alice and Bob, respectively.\\
\indent Let us now consider a set of relative frequencies  $f(a,b|x,y)$ of occurrence for outcomes $a,b$ when $x,y$ is measured by Alice and Bob, respectively, obtained from experimental data. The no-signaling constraints \eqref{ns1} and \eqref{ns2} are not exactly satisfied for experimental data. However, they can be recovered by minimizing the Kullback-Leible divergence \cite{Lin2018}:
\begin{equation}\label{KLdivergence}
D_{KL}(\vec{f}||\vec{P})=\sum_{a,b,x,y}f(x,y)f(a,b|x,y)\log_2\left[\frac{f(a,b|x,y)}{p(a,b|x,y)}\right],
\end{equation}
where $f(x,y)$ is the relative frequency of implementing a measurement $x$ by Alice and $y$ by Bob, and $p(a,b|x,y)$ the optimization variables, consisting of a joint probability distribution within the framework of quantum mechanics. The minimization procedure \eqref{KLdivergence} is equivalent to maximizing the likelihood of producing the observed frequency $p(a,b|x,y)$, see Appendix D1 in \cite{Lin2018}.

The quantum prediction of a Bell inequality (\ref{Bellineq}), defined by coefficients $s^{ab}_{xy}$, is given by  
\begin{equation}\label{Qexp}
\mathcal{Q} =\sum_{x,y=0}^{m-1}\sum_{a,b=0}^{d-1} s^{ab}_{xy}\,p(a,b|x,y),
\end{equation}
having associated an error propagation $\Delta\mathcal{Q}$; see section \ref{error_prop} for a detailed treatment of errors. An experimentally obtained probability distribution $p(a,b|x,y)$, associated to errors $\Delta p(a,b|x,y)$, is certainty nonlocal if $\mathcal{Q}-\mathcal{C}>\Delta\mathcal{Q}$, for a given Bell inequality. However, sometimes quantum nonlocality cannot be revealed due to the relatively high amount of errors. This especially occurs when a weakly entangled quantum state produces the maximal violation of the inequality, where the gap between the lhv value and the maximal violation is very small. Under such situation, the method introduced here provides a new Bell inequality that increases the chances to prove quantum nonlocality for a given set of probability distributions $p(a,b|x,y)$, associated to experimental errors $\Delta p(a,b|x,y)$. The method consists in solving the following nonlinear problem:
\begin{equation}\label{max}
R=\max_{s}\frac{\mathcal{Q}-\Delta \mathcal{Q}+dm}{\mathcal{C}+dm},
\end{equation}
for a fixed set of statistical data, where the optimization is implemented over all coefficients $s^{ab}_{xy}$ defining a Bell inequality \eqref{Bellineq}. The shifting factor $dm$ introduced in \eqref{max} avoids divergence of the function $R$; otherwise, the output inequality would be any such that the LHV value $\mathcal{C}$ vanishes. Optimization \eqref{max} is typically hard to implement, due to the presence of a large amount of local maximum values. 

To solve this problem, we implement the \emph{Sequential Least Squares Programming} (SLSQP) \cite{kraft1988} algorithm, using routines from the scientific library (Scipy) of the Python Programming Language \cite{scipy}. SLSQP is an efficient method to numerically solve constrained nonlinear optimization problems with bounds, well suited to solve the following problem
\begin{equation}
\begin{tabular}{ccc}
\text{maximize}   & $R(\vec{s}),$  &  \text{with } $\vec{s} = \{ s^{ab}_{xy} \}$ \\
\text{subject to}   &  $-1 \leq s^{ab}_{xy} \leq 1$. &  \\
\end{tabular}    \nonumber
\end{equation}
Our strategy consists in running this optimization for a given number of trials. In the first trial, a random seed real vector $\vec{x}_0$, with entries taken in the range $[-1,1]$, is given to the routine. After the first optimization procedure the algorithm outputs a vector $\vec{s}_0$. Then, the seed for the second trial is taken as $\vec{x}_1 = (\vec{s}_0 + \vec{x}_0)/2$. We update the seed as $\vec{x}_{k+1} = (\vec{s}_{k} + \vec{x}_{k})/2$ during the trials and keep the output $\vec{s}_{k} = \vec{s}$ for which $R(\vec{s})$ is the greatest. The solution to the optimization problem posed above is stored in vector $\vec{s}$. This approach is highly efficient to avoid getting stuck in local maximum values, although there is no way to certify that the reached solution is the global maximum. The code for this optimization procedure was written in Python, see Appendix  \hyperref[app:B]{B}, and is available at GitHub \cite{code_nlc}.\\

\begin{prop}\label{prop1}
A Bell inequality having LHV value $\mathcal{C}$, for which a quantum value $\mathcal{Q}$ is achieved with errors $\Delta\mathcal{Q}$, is genuinely nonlocal if and only if $R>1$ in Eq.(\ref{max}).
\end{prop}
\begin{proof}
Let $s^{ab}_{xy}$ be the parameters of the Bell inequality producing the maximum value $R$ in \eqref{max}. Therefore, it is simple to show that $\mathcal{Q}-\Delta\mathcal{Q}-\mathcal{C}=(R-1)(\mathcal{C}+dm)$. Therefore, the statistical data is nonlocal, i.e. $\mathcal{Q} - \Delta\mathcal{Q} - \mathcal{C} > 0$, if and only if $R>1$. Given that $R$ takes the maximal possible value among all Bell inequalities of the scenario, if $R\leq1$ then there is no Bell inequality that can detect quantum nonlocality for the given statistical data. 
\end{proof}
Here, genuinely nonlocal means that there exists a Bell inequality \eqref{Bellineq}, associated to coefficients $s^{ab}_{xy}$, such that the amount of experimental errors do not overpass the gap between the quantum violation $\mathcal{Q}$ and the LHV prediction $\mathcal{C}$. Therefore, $R>1$ is equivalent to saying that there is a way to experimentally certify quantum nonlocality for a given set of experimental data. On the other hand, if $R\leq1$  then there is no way to tell whether there exists a linear Bell inequality, as defined in Eq. \eqref{BellInequality}, that detects nonlocality for the given probabilities. 

Note that the gap of a Bell inequality $\mathcal{Q} - \Delta\mathcal{Q} - \mathcal{C}$ can be artificially enlarged by a multiplicative factor $\kappa>1$ by considering a Bell inequality having a rescaled set of  coefficients $\tilde{s}^{ab}_{xy}=\kappa s^{ab}_{xy}$. In order to avoid this scaling problem, it is convenient to refer to the \emph{relative gap}, given by the gap of a Bell inequality for which $\mathcal{C}=1$, which can be assumed without loss of generality. Thus, the global maximum value of $R$ in \eqref{max} implies the largest possible relative gap, as we show below.
\begin{prop}\label{prop2}
The Bell inequality associated to the global maximum of the function $R$, introduced in \eqref{max}, produces the largest possible relative gap between the LHV and quantum values, among all linear Bell inequalities of a given scenario. 
\end{prop}
\begin{proof}
Without loss of generality, we can restrict our attention to Bell inequalities for which $\mathcal{C}=1$. Indeed, this can always be done by considering the following rescaling of a given Bell inequality: $\mathcal{Q}'=\mathcal{Q}/\mathcal{C}$,  $\Delta\mathcal{Q}'=\Delta\mathcal{Q}/\mathcal{C}$ and $\mathcal{C}'=1$. Here, we used the fact that error propagation is linear as a function of the coefficients $s^{ab}_{xy}$ defining the Bell inequality.  After the rescaling, optimization of $R$ is equivalent to maximize $\mathcal{Q}'-\Delta\mathcal{Q}'$ along all Bell inequalities satisfying $\mathcal{C}'=1$, according to \eqref{max}. This is equivalent to maximizing the gap $\mathcal{Q}'-\Delta\mathcal{Q}'-1=\mathcal{Q}'-\Delta\mathcal{Q}'-\mathcal{C}'$.
\end{proof}

In section \ref{sec:tilted}, we apply our optimization method to the so-called \emph{tilted Bell inequality} \cite{Acin2012}. This family of inequalities was used to demonstrate that almost two bits of randomness can be extracted from a quantum system prepared in a weakly entangled state. This property makes the tilted Bell inequality an ideal candidate to test our method, due to the hardness to certify nonlocality in such case.

\section{Error propagation}\label{error_prop}

This section provides a general expression for the experimental error obtained by measuring the Bell inequality value. In particular, we show how errors in the photon counting number due to finite statistics propagates to $\Delta \mathcal{Q}$. Consider the following general expression for a Bell inequality: 
\begin{equation}\label{S1}
 \mathcal{Q} =\sum\limits_{x,y = 0}^{m - 1}  \sum\limits_{a,b = 0}^{d-1} s_{xy}^{ab} p(ab|xy) + \sum\limits_{x = 0}^{m - 1}  \sum\limits_{a = 0}^{d-1} s_{x}^{a} p_{A}(a|x) + \sum\limits_{y = 0}^{m - 1}  \sum\limits_{b = 0}^{d-1} s_{y}^{b} p_{B}(b|y).
\end{equation}
\noindent Here we include both joint and marginal probability distributions. As typical, the marginal probabilities are calculated from the joint probabilities, and we average over all possible $x$ (or $y$), i.e. $ p(a|x) = m^{-1} \sum_{y = 0}^{m-1}\sum_{b = 0}^{d-1} p(ab|xy)$ and $  p(b|y) = m^{-1}\sum_{x = 0}^{m-1}\sum_{a = 0}^{d-1} p(ab|xy)$. Replacing these quantities in Eq.\eqref{S1} and rewriting $\mathcal{Q}$ in term of the coincidence count $c(ab|xy)$ we get
\begin{equation}\label{S2}
   \mathcal{Q}=\sum\limits_{x,y = 0}^{m - 1}  \sum\limits_{a,b = 0}^{d-1}  \dfrac{c(ab|xy)}{\sum_{\alpha \beta} c(\alpha \beta|xy)} \left [ s_{xy}^{ab} + \dfrac{s_{x}^{a}}{m}    +  \dfrac{s_{y}^{b}}{m}   \right ],
\end{equation}
\noindent \noindent with $p(ab|xy) = c(ab|xy)\left( \sum_{\alpha \beta} c(\alpha \beta|xy)\right)^{-1}$. Finally, Gaussian error propagation and the Poisson statistics of the recorded coincidence count are considered to calculate $\Delta \mathcal{Q}$. The Possonian nature of the coincidence counts gives squared error  $(\Delta c(ab|xy))^2=c(ab|xy)$. The general expression for the experimental error is then
\begin{equation}\label{S3}
    \Delta \mathcal{Q} = \sqrt{ \sum_{abxy}\left( \dfrac{\partial \mathcal{Q}}{\partial c(ab|xy)} \right)^2 c(ab|xy)},
\end{equation}
\noindent and straightforward calculation leads to 
\begin{multline}\label{S4}
    \dfrac{\partial \mathcal{Q}}{\partial c(a'b'|x'y')} = \dfrac{1 }{ \left( \sum_{ab} c(ab|x'y') \right)^2 } \left[  \left( S_{x'y'}^{a'b'}  + \dfrac{1}{m}\left( S_{x'}^{a'} + S_{y'}^{b'} \right) \right)\sum_{ab} c(ab|x'y') \right. - \\ \left.  \sum_{ab}\left( S_{x'y'}^{ab}  + \dfrac{1}{m}\left( S_{x'}^{a} + S_{y'}^{b} \right) \right) c(ab|x'y')   \right].
\end{multline}

\section{Tilted Bell inequality}\label{sec:tilted} 
Certification of quantum nonlocality from experimental data is a challenging task in general. In a recent work,  a genuine random number generation protocol based on quantum nonlocality was experimentally demonstrated \cite{Gomez2019}. Here, authors certified randomness from a physical system prepared in quantum states with concurrences $C=0.986$, $C=0.835$ and $C=0.582$. However, certification failed for $C=0.193$ and $C=0.375$, due to the high sensibility of certifying quantum violation under the presence of errors. The inequality in question is the tilted Bell inequality \cite{Acin2012}. This inequality is well suited for our study, since it allows us to test our method for different levels of entanglement, from maximally entangled to separable, by just modifying the parameter $\alpha$: 
\begin{equation}\label{belltilted}
\alpha \left[ p_{A}(0|0)-p_{A}(1|0) \right]+ \sum_{x,y=0}^{1}\sum_{a,b=0}^{1} (-1)^{xy}\left[ p(a=b|xy)-p(a \neq b|xy)\right] \leq \mathcal{C}_{\alpha},
\end{equation}
where $p_A(a|x)$ is the marginal probability distribution for Alice,  $\mathcal{C}_{\alpha}=\alpha+2$ and $\mathcal{Q}_{\alpha}=\sqrt{8+2\alpha^2}$, $\alpha\in[0,2]$. The  quantum state producing the maximal violation in \eqref{belltilted} is given by $|\psi(\theta)\rangle=\cos(\theta)|00\rangle+\sin(\theta)|11\rangle$,
where $\theta=2^{-1}\arcsin \left( \sqrt{\frac{1-(\alpha/2)^{2}}{1+(\alpha/2)^{2}}}\right) $. Optimal settings are provided by eigenvector bases of the following observables:
\begin{equation}
    A_{0}=\sigma_{z}, \quad A_{1}=\sigma_{x}, \quad B_{0}= \cos(\mu)\sigma_{z} + \sin(\mu)\sigma_{x} \quad \text{and} \quad B_{1} = \cos(\mu)\sigma_{z} - \sin(\mu)\sigma_{x} \nonumber 
\end{equation}
\noindent where $\mu=\arctan \left( \sqrt{\frac{1-(\alpha/2)^{2}}{1+(\alpha/2)^{2}}}\right)$ and $\sigma_x = \sigma_1$ and $\sigma_z = \sigma_3$ the Pauli Matrices defined in Eq. \eqref{PauliMatrices}.

\begin{figure}[h!]
\centering
\includegraphics[scale=0.675]{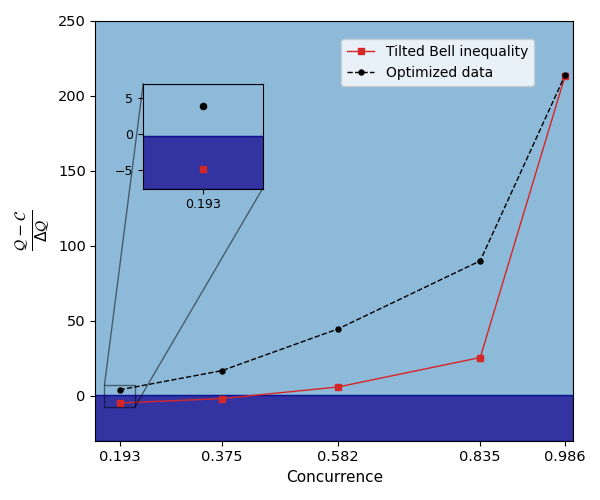}
\caption[NL: Nonlocality certification for different concurrences]{Number of standard deviation (SDN) for the gap between the LHV and quantum values, as a function of concurrence. For each value of concurrence, the optimization procedure \eqref{max} provides a  Bell inequality that increases the number of standard deviations of the quantum-LHV value gap. SDN is calculated for two cases: the original tilted Bell inequality [red squares] and an inequality arising from optimization \eqref{max} [black circles]. In both cases, we consider experimental data, where quantum nonlocality can be certified in the light blue region. For concurrences $C=0.375$ and $C=0.193$, there is no quantum violation of the tilted Bell inequality, whereas inequalities arising from optimization \eqref{max} produce a violation for both concurrences. This result is obtained with the same statistical data that did not produce a violation with the tilted Bell inequality, i.e. a new experiment was not required to certify quantum nonlocality.}
\label{Fig1}
\end{figure}

We tested optimization \eqref{max} for a photonic Bell inequality experiment, see Appendix B in Ref. \cite{Gomez_2021} for experimental details. The experiment consisted of a high-quality source of polarization-entangled Bell states of the form $\ket{\psi(\theta)}$.  Entanglement of this optimal quantum state can be characterized by its concurrence, given by $C=\sin(2 \theta)$. By using optimization \eqref{max}, we improve the experimental violation of the tilted Bell inequality for high concurrences and, more importantly, we successfully demonstrate quantum nonlocality for low values of concurrence, i.e. $C=0.375$ and $C=0.193$, something that failed to be proven  when considering the tilted Bell inequality \eqref{belltilted} \emph{from the same statistical data}. In Fig. \ref{Fig1} we show the number of standard deviations of the quantum-LHV gap for the tilted Bell inequality and the one resulting from our optimization procedure \eqref{max} (more details in Table \ref{table:one}).

\begin{table}[h!]
\centering
   \caption[Summary of results for non-locality certification]{Summary of results. The concurrence was obtained from quantum state tomography. $\mathcal{Q}_{\alpha}$  refers to the experimental values of the Tilted Bell inequality. $\mathcal{C}_{\alpha}$ is the LHV bound for each $\alpha$. After implementing the method mentioned in section II, the results obtained are presented in the values $\mathcal{Q}$ and $\mathcal{C}$.}
\resizebox{1.0\textwidth}{!}{\begin{minipage}
{\textwidth} \begin{center}
        \begin{tabular}{ || c c c c c c ||}
        \hline 
        Concurrence & $0.193$ & $0.375$ & $0.582$  &  $0.835$ & $0.986$ 
        \\
        \hline \hline
        $\mathcal{Q}_{\alpha}$   & $3.890$ & $3.686$ & $3.418$ & $3.108$ & $2.819$  \\
        $(\pm) \Delta \mathcal{Q}_{\alpha}$ & $0.006$ & $0.008$ & $0.008$ & $0.007$ & $0.004$\\ 
        $\mathcal{C}_{\alpha}$ & $3.921$ & $3.702$ & $3.372$ & $2.937$ & $2.002$ \\ 
        $\frac{\mathcal{Q}_{\alpha}-\mathcal{C}_{\alpha}}{\Delta \mathcal{Q}_{\alpha}}$ & $-4.85$ & $-1.93$ & $5.83$ & $25.44$ & $213.06$ \\
        \hline
        $\mathcal{Q}$   & $1.417$  & $1.505$ & $1.314$ & $1.557$ & $1.819$ \\
        $(\pm) \Delta \mathcal{Q}$ & $0.001$ & $0.007$ & $0.007$ & $0.006$ & $0.004$  \\ 
        $\mathcal{C}$ & $1.412$ & $1.381$ & $1.000$ & $1.000$ & $1.000$ \\ 
        $\frac{\mathcal{Q}-\mathcal{C}}{\Delta \mathcal{Q}}$ &  $3.88$ & $16.68$ & $44.53$ & $89.87$ & $213.57$         \\    
        \hline \hline
       
        \end{tabular} 
        \end{center}
	\end{minipage}}
	\label{table:one}
\end{table}

This result can lead to interesting practical applications. For instance, note that the amount of bits of randomness that can be extracted from the protocol \cite{Gomez2019} are identical to the number of bits that can be extracted by considering our optimized Bell inequality. This is so because both schemes consider the same experimental data, representing the same physical phenomena. 

\newpage
\subsubsection{Inequalities}    
\vspace{-2.25mm}
For each value of concurrence in table \ref{table:one}, a set of coefficients $s^{ab}_{xy}$ was obtained by optimizing Eq. \eqref{max}. Here we show those coefficients by writing the Bell inequalities:\\

\noindent \textit{Concurrence = 0.193}
\begin{equation}
\begin{aligned}
 & \;\;\;\; 0.9424 p(00|00)-0.2156 p(01|00)-0.3479 p(10|00) + 0.8025 p(11|00) \\ 
 &  + 0.5397 p(00|01)-0.3703 p(01|01)-0.9992 p(10|01) + 0.3785 p(11|01) \\ 
 &  + 0.1495 p(00|10)-0.9801 p(01|10)-0.1586 p(10|10) + 0.9992 p(11|10) \\ 
 &  + 0.0928 p(00|11) + 0.9992 p(01|11) + 0.4009 p(10|11)-0.9773 p(11|11) \\ 
 &  -0.3122 p_A(0|0) + 0.2084 p_A(1|0) \leq 1.4122.
\end{aligned}
\nonumber
\end{equation}
\noindent \textit{Concurrence = 0.375}
\begin{equation}
\begin{aligned}
 & \;\;\;\; 0.9784 p(00|00)-0.9989 p(01|00)-0.9993 p(10|00) + 0.9970 p(11|00) \\ 
 &  + 0.9980 p(00|01)-0.9932 p(01|01)-0.9990 p(10|01) + 0.9833 p(11|01) \\ 
 &  + 0.9996 p(00|10)-0.9967 p(01|10)-0.9807 p(10|10) + 0.9965 p(11|10) \\ 
 & -0.9976 p(00|11) + 0.9936 p(01|11) + 0.9827 p(10|11)-0.9997 p(11|11) \\ 
 &  -0.5965 p_A(0|0) -0.5953 p_A(1|0) \leq 1.3819.
\end{aligned}
\nonumber
\end{equation}
\noindent \textit{Concurrence = 0.582}
\begin{equation}
\begin{aligned}
 & \;\;\;\; p(00|00)- p(01|00)- p(10|00) +  p(11|00) +  p(00|01)- p(01|01) \\
 & - p(10|01) + p(11|01) + p(00|10)- p(01|10)- p(10|10) + p(11|10) \\ 
 & - p(00|11) + p(01|11) + p(10|11)- p(11|11)  - p_A(0|0) - p_A(1|0) \leq 1
\end{aligned}
\nonumber
\end{equation}
\noindent \textit{Concurrence = 0.835}
\begin{equation}
\begin{aligned}
 & \;\;\;\;  p(00|00)- p(01|00)- p(10|00) + p(11|00)  +  p(00|01)- p(01|01) \\
 & - p(10|01) +  p(11|01) +  p(00|10)- p(01|10)- p(10|10) +  p(11|10) \\ 
 & - p(00|11) + p(01|11) + p(10|11)- p(11|11)  - p_A(0|0) - p_A(1|0) \leq 1.
\end{aligned}
\nonumber
\end{equation}
\noindent \textit{Concurrence = 0.986}
\begin{equation}
\begin{aligned}
 & \;\;\;\; p(00|00)- p(01|00)- p(10|00) + p(11|00)  + p(00|01)- p(01|01) \\
 & - p(10|01) + p(11|01) + p(00|10)- p(01|10)- p(10|10) + p(11|10) \\ 
 & - p(00|11) + p(01|11) +  p(10|11)- p(11|11)  -p_A(0|0) - p_A(1|0) \leq 1.
\end{aligned}
\nonumber
\end{equation}

\section{Closing the detection loophole}\label{sec:loophole}

Since the work conducted by Aspect et al. in 1981 \cite{Aspect_1982}, which signified the experimental verification of Bell's work, physicists have been concerned about \textit{loopholes} in the design of experiments to certify nonlocality. Free loopholes experiments are important, for example, in quantum key distribution protocols for secure communication \cite{Acin2007}. Experiments to certify nonlocality has to guarantee that no messages are being interchanged between the parts involved; if they are not sufficiently distant apart, then signals can travel, at the speed of light, from one measurement location to the other and generate correlations. This is called the \textit{locality loophole} \cite{bell_aspect_2004}. 

Another source of loopholes is ``missing" detections.  If the \textit{efficiency} of the detector is too low, then some of the particles, e.g. photons, might end up not being registered by the detector and the quantum correlations could be reproduced by a local hidden variables model \cite{Santos_1992}; this is called the \textit{detection loophole}. Thus, one might wonder what is the \textit{minimum} detection efficiency that an experiment can tolerate to consider that a Bell inequality is genuinely being violated. For example, for the CHSH inequality a detector needs to have an efficiency larger than 82.8\% to close the detection loophole considering maximally entangled states \cite{Brunner_2007}. A loophole-free Bell experiment was reported in 2015 \cite{Hensen_2015}. In this section, we show that the optimization procedure \eqref{max}, apart from increasing the quantum-LHV value gap, also reduces the detection efficiency required to close the detection loophole for a fixed set of statistical data. 

Suppose that Alice and Bob have detector efficiencies $\eta_A$ and $\eta_B$, respectively. The minimal efficiencies required to violate a Bell inequality are given by the following procedure \cite{Pironio2010}: first, a two-outcome Bell inequality has to be written in a canonical form. Namely, it has to consider only one of its outcomes per party (we choose $a=b=0$), as we associate the other outcomes ($a=b=1$) to the cases where detectors do not fire correctly. 

To transform any bipartite Bell inequality with $m$ settings and two outcomes to its canonical form, i.e. depending on outputs $a=b=0$ only, the following identities have to be considered for local
\begin{eqnarray}
p_A(1|x)&=&1-p_A(0|x),\nonumber\\
p_B(1|y)&=&1-p_B(0|y),\nonumber
\end{eqnarray}
and joint probabilities
\begin{eqnarray}
p(0,1|x,y)&=&p_A(0|x)-p(0,0|x,y),\nonumber\\
p(1,0|x,y)&=&p_B(0|y)-p(0,0|x,y),\nonumber\\
p(1,1|x,y)&=&1-p_A(0|x)-p_B(0|y)+p(0,0|x,y)\nonumber,
\end{eqnarray}
\noindent for every $x,y=0,m-1$. Applying these identities, it is simple to show that any bipartite Bell inequality \eqref{Bellineq} with $m$ settings per side and two outcomes can be written as: 
\begin{equation}\label{effi1a}
\sum_{x,y=0}^{m-1} \tilde{s}_{xy}^{00} p(0,0|x,y)+\sum_{x=0}^{m-1} \tilde{s}_{Ax}^{\hspace{6px}0} p_A(0|x)+\sum_{y=0}^{m-1} \tilde{s}_{By}^{\hspace{6px}0} p_B(0|y) \leq \mathcal{C},
\end{equation}
\noindent where 
\vspace{-2mm}
\begin{eqnarray}
\tilde{s}^{00}_{xy}&=&\sum_{a,b=0}^1 (-1)^{a+b}s^{ab}_{xy},\nonumber\\
\tilde{s}_{Ax}^{\hspace{6px}0}&=&s_{Ax}^{\hspace{6px}0}-s_{Ax}^{\hspace{6px}1}+\sum_{y=0}^{m-1}\sum_{a=0}^1 (-1)^a s^{a1}_{xy},\nonumber\\
\tilde{s}_{By}^{\hspace{6px}0}&=&s_{By}^{\hspace{6px}0}-s_{By}^{\hspace{6px}1}+\sum_{x=0}^{m-1}\sum_{b=0}^1 (-1)^b s^{1b}_{xy}.\nonumber
\end{eqnarray}
For instance, for the tilted Bell inequality (\ref{belltilted}), we have the following canonical form:
\begin{equation}\label{canonical_tilted}
(\alpha/2-1) p_A(0|0)- 
 p_B(0|0)+p(0,0|0,0) + p(0,0|0,1) +p(0,0|1,0)-
 p(0,0|1,1)\leq \tilde{\mathcal{C}}_{\alpha},
\end{equation}
where $\tilde{\mathcal{C}}_{\alpha}=(\mathcal{C}_{\alpha}+\alpha-2)/4$. Second, due to imperfect detectors, probabilities have to be transformed according to the the rule $p(0,0|x,y)\rightarrow \eta_A\eta_B p(0,0|x,y)$, $p(0|x)\rightarrow \eta_A\, p(0|x)$ and $p(0|y)\rightarrow \eta_B\, p(0|y)$, for every setting $x,y$. Therefore, the lower bound for the minimal efficiencies required to detect a quantum violation satisfies:
\begin{equation}\label{effi1b}
\eta_A\eta_B\sum_{x,y=0}^{1} \tilde{s}_{xy}^{00} p(0,0|x,y)+\eta_A\sum_{x=0}^{1} \tilde{s}_{Ax}^{\hspace{6px}0} p_A(0|x)+\eta_B\sum_{y=0}^{1} \tilde{s}_{By}^{\hspace{6px}0} p_B(0|y)=\mathcal{C}(s).
\end{equation}
We observe that optimization \eqref{max} reduces the detection efficiency required to close the detection loophole, with respect to a target Bell inequality, for a fixed set of probability distributions $p(a,b|x,y)$, $p_A(0|x)$ and $p_B(0|y)$. Here, we show evidence of this fact by improving detection efficiencies for the experimental 

\begin{figure}[h]
    \centering
     \subfloat[\label{Fig3a}]{%
       \includegraphics[scale=0.76]{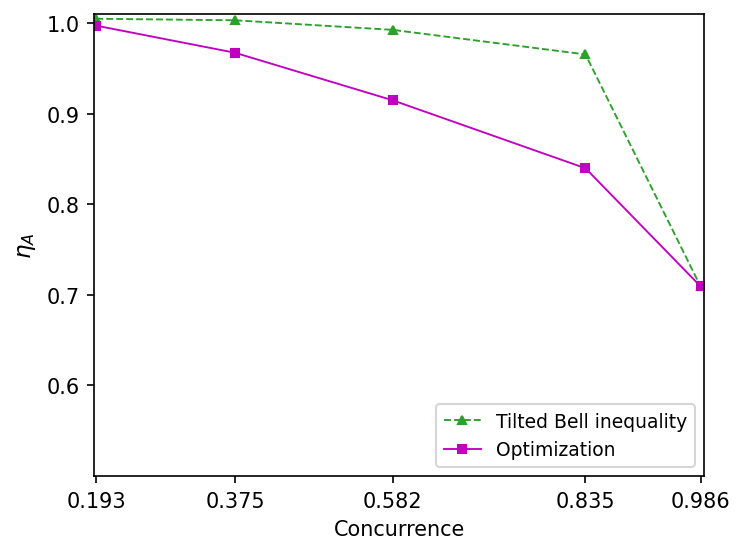}
     }
     
     \subfloat[\label{Fig3b}]{%
       \includegraphics[scale=0.76]{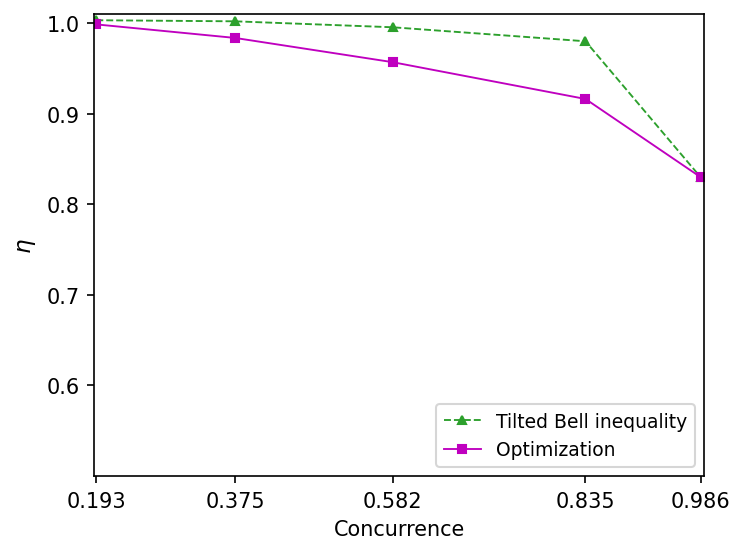}
     }
     \caption[NL: Efficiencies ]{Minimum efficiencies required to close the loophole for the (a) Asymmetric case, where $\eta_B=1$, and (b) Symmetric case, where $\eta_A=\eta_B=\eta$. We compare efficiencies associated with the tilted Bell inequality and the optimization proposed in (\ref{max}). In both cases, we consider the optimal experimental data obtained for the tilted inequality \eqref{canonical_tilted}.}
     \label{efficiency}
\end{figure}
\clearpage 
\noindent statistical distributions associated to the tilted Bell inequality (\ref{canonical_tilted}), see Figure \ref{efficiency}.

\section{Conclusions}
Violation of Bell inequalities is at the heart of quantum physics and defines a cornerstone for a wide range of quantum information protocols with real-world appeal.  We have shown how an ``experiment-inspired" optimization procedure can be applied to the search of Bell inequalities that increase the quantum-LHV gap with respect to a target Bell inequality, for a given set of experimental data (see Prop. \ref{prop2}). Furthermore, we demonstrated that the gap provided by our optimization procedure is the largest possible among the entire set of Bell inequalities having LHV equal to 1, something that can be assumed without loss of generality (see Prop. \ref{prop2}). When nonlocality certification from a given set of experimental data fails, our method provides a ``second chance" to succeed, without requiring to perform any additional measurement. Furthermore, our method also provides a gain in the minimal detection efficiency required by a fixed statistical set, a crucial ingredient to maximize the chances to close the detection loophole.  We illustrated our technique by considering  experimental data associated to the maximal violation of the so-called \emph{tilted Bell inequality}. Here, we considerably increased some previously obtained gaps, a fact that allowed us to certify nonlocality when considering weakly entangled quantum states. This certification was not possible to do with the tilted Bell inequality, i.e. the inequality that motivated the experiment, see Section \ref{sec:tilted}. Furthermore, we also showed how the detection efficiencies required to close the detection loophole can be reduced after implementing our optimization procedure, see Section \ref{sec:loophole}. Our technique finds application in device-independent protocols, random number generation, communication complexity and any practical application based in quantum nonlocality.

    \chapter{Quantum Marginal Problem}\label{cap4}

In this chapter, we study the quantum marginal problem. Here, we introduce an operator to tackle the problem of compatibility between parts of a quantum system and the whole. We also present an algorithm that applies iteratively this operator to find a global quantum state, if it exists, compatible with a prescribed set of quantum marginals and spectra.

\section{Introduction}

The Quantum Marginal Problem (QMP) is basically a problem of compatibility: given a set of marginals we want to find a \emph{compatible} global density matrix, where compatible means that the marginals can be recovered from the global state using the partial trace. The fermionic version of the QMP is known as the \emph{N-Representability problem}, terminology introduced in 1963 by  A. J. Coleman \cite{coleman1963}, and represents one of the biggest challenges in quantum chemistry and quantum many body physics. Knowing the compatibility conditions would, for example, reduce the computational cost of calculating the ground state of many-body quantum systems \cite{coleman1963}. Several works have also been devoted to establish the connections between the QMP and quantum entanglement \cite{Yu2021ACH, Navascus2020EntanglementMP}. 


The QMP, in its most general form, is hard; it was shown in Ref. \cite{QMA_2007} that the $N$-Representability problem belongs to the Quantum Merlin-Arthur complete complexity class, which is the quantum analogous of the NP-complete class. Even the case of two-body marginals, of particular interest since most of Hamiltonians describing $N$-fermion systems, such as atoms and molecules, contain only two-body interactions \cite{coleman1963}, is believed to be intractable.

Despite the complexity of the QMP, tremendous progress has been made. Klyachko \cite{Klyachko_2006} completely solved the problem for the case of one-body marginals. He showed that the compatibility conditions are given by a set of linear inequalities involving the spectra of the one-body marginals. For the case of symmetric states (states describing bosons), necessary and sufficient conditions are given in Ref. \cite{Aloy_2021}, where the compatibility conditions are expressed as an SDP problem. Other works have focused on establishing whether a global state is uniquely determined by a given set of marginals among the pure states or among all the states. For example, only two $2$-body marginals are necessary to uniquely determine, in most of cases, $3$-body pure states \cite{Lajos_2004}. Wyderka and co-authors \cite{Wyderka2017} showed that most of $n$-body pure states are uniquely determined, among pure states, by only three of their $(n-2)$-body marginals.

\section{Imposing Marginals}
\label{section_4.1}
In this section, we introduce the main tool of our study and describe its properties. We first need to distinguish two cases of the QMP: overlapping and non-overlapping. Let us consider a $N$-body quantum system formed by parties labeled $A, B, C, \ldots$, which form the ordered set $\mathcal{I} = \{ A,B,C, \ldots \}$. Let us also consider sets $\mathcal{J}_i  \subset \mathcal{I}$. We say that $\mathcal{J}_i$ and $\mathcal{J}_j$ are \textit{overlapped} if $\mathcal{J}_i \cap \mathcal{J}_j \neq \emptyset $ for $i \neq j$. Similarly, $\mathcal{J}_i$ and $\mathcal{J}_j$ are \textit{non-overlapped} if $\mathcal{J}_i \cap \mathcal{J}_j = \emptyset $ for $i \neq j$.  In Fig. \ref{overlapping_no_overlapping} we illustrate the overlapping and non-overlapping cases for a 4-body system. 

To simplify the notation, we denote along this chapter the maximally mixed state of the system $\mathcal{J}$ by $\mathbb{I}_{\mathcal{J}}$. Now, we define the following operator:

\begin{figure}
    \centering
     \subfloat[\label{fig:overlapping}]{%
       \includegraphics[scale=0.565]{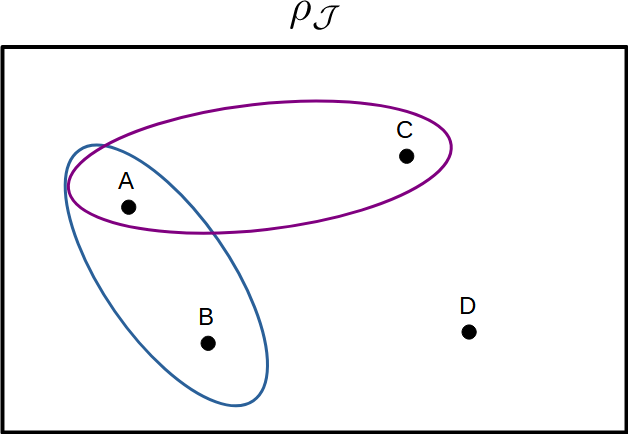}
     }\hspace{5mm}
     \subfloat[\label{fig:no_overlapping}]{%
       \includegraphics[scale=0.565]{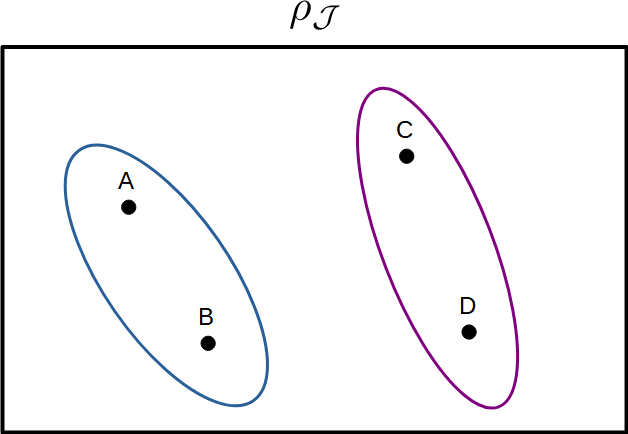}
     }
     \caption[QMP: Overlapping and non-overlapping quantum marginals]{Overlapping and non-overlapping cases for a 4-body quantum system with 2-body marginals (enclosed by the coloured islands). a) \textit{Overlapping:} $\mathcal{I}_1 \cap \mathcal{I}_2 \neq \emptyset $ with $\mathcal{I}_1 = \{A,B\}$ and $\mathcal{I}_2 = \{A,C\}$  b) \textit{Non-overlapping:} $\mathcal{I}_1 \cap \mathcal{I}_2 = \emptyset $ with $\mathcal{I}_1 = \{A,B\}$ and $\mathcal{I}_2 = \{C,D\}$.}
     \label{overlapping_no_overlapping}
\end{figure}

\begin{defi}\label{def:qmo}
Let $\rho_{\mathcal{I}}$ and  $\sigma_{\mathcal{J}}$ be multipartite quantum states associated to a set of parties $\mathcal{I}$ and $\mathcal{J}\subset\mathcal{I}$, respectively. We define the marginal imposition operator as:
\begin{equation}\label{qmo}
Q_{\sigma_{\mathcal{J}}}(\rho_{\mathcal{I}}):=\rho_{\mathcal{I}}-\rho_{\mathcal{J}}+\sigma_{\mathcal{J}},
\end{equation}
where $\rho_{\mathcal{J}}=\mathrm{Tr}_{\mathcal{J}^c}[ \rho_{\mathcal{I}}]$ and $\mathcal{J}^c$ denotes the complement of $\mathcal{J}$ with respect to $\mathcal{I}$. 
\end{defi}

Here, reductions are implicitly tensored with maximally mixed states, e.g. $\sigma_{\mathcal{J}}$ means $\sigma_{\mathcal{J}}\otimes\mathbb{I}_{\mathcal{J}^c}$. As its name suggests, the operator $Q_{\sigma_{\mathcal{J}}}(\rho_{\mathcal{I}})$ imposes the marginal $\sigma_{\mathcal{J}}$ into the density matrix $\rho_{\mathcal{I}}$. To illustrate this, let us consider the case with $\mathcal{I} = \{A,B,C\}$ and $\mathcal{J} = \{A,B\}$. Thus, Eq. \eqref{qmo} becomes:
\begin{equation}\label{qmo_3parties}
Q_{\sigma_{AB}}(\rho_{ABC})=  \rho_{ABC}-\rho_{AB} \otimes \mathbb{I}_C +\sigma_{AB}\otimes \mathbb{I}_C ,
\end{equation}
\noindent where $\rho_{AB} = \mathrm{Tr}_{C}[\rho_{ABC}]$. By partial tracing Eq. \eqref{qmo_3parties} over $C$, one obtain $\sigma_{AB}$. Let us formalize this observation in the following proposition:

\vspace{3mm}
\begin{prop}\label{qmo:imposition}
Let $\rho_{\mathcal{I}}$ and  $\sigma_{\mathcal{J}}$ be multipartite quantum states associated to a set of parties $\mathcal{I}$ and $\mathcal{J}\subset\mathcal{I}$, respectively. Then $\mathrm{Tr}_{\mathcal{J}^c}\left[ Q_{\sigma_{\mathcal{J}}}(\rho_{\mathcal{I}})  \right] = \sigma_{\mathcal{J}}$.
\end{prop}
\begin{proof}
Let us compute $\mathrm{Tr}_{\mathcal{J}^c}\left[ Q_{\sigma_{\mathcal{J}}}(\rho_{\mathcal{I}})  \right]$ 
\begin{eqnarray}\label{imposed_marginal}
\mathrm{Tr}_{\mathcal{J}^c}\left[Q_{\sigma_{\mathcal{J}}}(\rho_{\mathcal{I}}) \right] &=& \mathrm{Tr}_{\mathcal{J}^c}\left[\rho_{\mathcal{I}}\right]- \mathrm{Tr}_{\mathcal{J}^c}\left[\rho_{\mathcal{J}}\right]+\mathrm{Tr}_{\mathcal{J}^c}\left[\sigma_{\mathcal{J}} \right],
\end{eqnarray}
\noindent from Def. \ref{def:qmo} $\rho_{\mathcal{J}}=\mathrm{Tr}_{\mathcal{J}^c}[ \rho_{\mathcal{I}}]$. Also, $\mathrm{Tr}_{\mathcal{J}^c}\left[\rho_{\mathcal{J}}\right] = \rho_{\mathcal{J}}$ and $\mathrm{Tr}_{\mathcal{J}^c}\left[\sigma_{\mathcal{J}} \right] = \sigma_{\mathcal{J}}$. Thus, Eq. \eqref{imposed_marginal} becomes $\mathrm{Tr}_{\mathcal{J}^c}\left[Q_{\sigma_{\mathcal{J}}}(\rho_{\mathcal{I}}) \right] = \rho_{\mathcal{J}}- \rho_{\mathcal{J}}+\sigma_{\mathcal{J}} = \sigma_{\mathcal{J}}$.
\end{proof}

The action of $Q_{\sigma_{\mathcal{J}}}$ over $\rho_{\mathcal{I}}$ removes all the information about $\mathcal{J}$ and outputs an hermitian matrix containing the marginal $\sigma_{\mathcal{J}}$. Furthermore, the action of $Q_{\sigma_{\mathcal{J}}}$ does not affect the information about the subsystems in $\mathcal{I}\setminus\mathcal{J} $. Let us formalize this last statement in the following proposition:

\vspace{3mm}
\begin{prop}\label{obs1}
Let $\rho_{\mathcal{I}}$ and $\sigma_{\mathcal{J}}$ be multipartite quantum states, with $\mathcal{J}\subset\mathcal{I}$ and $\mathcal{K}\subset\mathcal{I}$ such that $\mathcal{J}\cap\mathcal{K}=\emptyset$. Therefore,
$\mathrm{Tr}_{\mathcal{K}^c}[Q_{\sigma_{\mathcal{J}}}(\rho_{\mathcal{I}})]=\mathrm{Tr}_{\mathcal{K}^c}[\rho_{\mathcal{I}}]=\rho_{\mathcal{K}}$.
\end{prop}
\begin{proof}
Let us compute $\mathrm{Tr}_{\mathcal{K}^c}\left[ Q_{\sigma_{\mathcal{J}}}(\rho_{\mathcal{I}})  \right]$ 
\begin{eqnarray}\label{imposed_marginal_k}
\mathrm{Tr}_{\mathcal{K}^c}\left[Q_{\sigma_{\mathcal{J}}}(\rho_{\mathcal{I}}) \right] &=& \mathrm{Tr}_{\mathcal{K}^c}\left[\rho_{\mathcal{I}}\right]- \mathrm{Tr}_{\mathcal{K}^c}\left[\rho_{\mathcal{J}}\right]+\mathrm{Tr}_{\mathcal{K}^c}\left[\sigma_{\mathcal{J}} \right],
\end{eqnarray}
\noindent by definition \eqref{reducedSystem} $\mathrm{Tr}_{\mathcal{K}^c}[\rho_{\mathcal{I}}]=\rho_{\mathcal{K}}$. Also, $\mathcal{J}\cap\mathcal{K}=\emptyset$ implies that $\mathcal{J} \subseteq	\mathcal{K}^c$ and therefore $\mathrm{Tr}_{\mathcal{K}^c}\left[\rho_{\mathcal{J}}\right] = \mathrm{Tr}_{\mathcal{K}^c}\left[\sigma_{\mathcal{J}}\right] = \mathbb{I}_{\mathcal{K}^c}$. Thus, \eqref{imposed_marginal_k} becomes $\mathrm{Tr}_{\mathcal{K}^c}[Q_{\sigma_{\mathcal{J}}}(\rho_{\mathcal{I}})]=\mathrm{Tr}_{\mathcal{K}^c}[\rho_{\mathcal{I}}]=\rho_{\mathcal{K}}$.
\end{proof}

We can verify Prop. \ref{obs1} for Eq. \eqref{qmo_3parties} by choosing $\mathcal{K} = \{C\}$. We obtain $\mathrm{Tr}_{AB}[\rho_{ABC}] = \mathrm{Tr}_{AB}[Q_{\sigma_{AB}}(\rho_{ABC})] = \rho_{C} $. 

In the following proposition we show that the output of Eq. \eqref{qmo} is not always a positive semidefinite hermitian matrix (PSD). Henceforth, to simplify the notation, we will use $Q_{\sigma_{\mathcal{J}}} = Q_{\mathcal{J}}$ and $\rho = \rho_{\mathcal{I}}$.

\begin{prop}\label{obs:CPTP}
The operator \eqref{qmo} is not always positive semidefinite. 
\end{prop}
\begin{proof}

Let us consider a $2$-qubit quantum system. Let us define the following density matrix
\begin{equation}\label{np_pAB}
    \rho_{AB} = 
    \begin{pmatrix}
      \alpha & 0 & 0 & 0 \\
      0 & 0 & 0 & 0 \\
      0 & 0 & 0 & 0 \\
      0 & 0 & 0 & \beta
     \end{pmatrix},
\end{equation}
\noindent with $\alpha + \beta = 1$ and $0 \leq \alpha, \beta \leq 1$, and calculate $\rho_A$ from \eqref{np_pAB}
\begin{equation}
    \rho_{A} = \mathrm{Tr}_{B}[\rho_{AB}] = 
    \begin{pmatrix}
      \alpha & 0 \\
      0 & \beta
     \end{pmatrix}.
\end{equation}

\noindent Let us also consider the following $1$-body marginal 
\begin{equation}
    \sigma_{A} = 
    \begin{pmatrix}
      \gamma & 0 \\
      0 & \delta
     \end{pmatrix},
\end{equation}
\noindent with $\gamma + \delta = 1$ and $0 \leq \gamma, \delta \leq 1$. Replacing $\rho_{AB}$, $\rho_A$ and $\sigma_A$ in \eqref{qmo} we obtain



\begin{eqnarray} \label{non_positive}
Q_{A}(\rho_{AB}) &= &  \rho_{AB}-\rho_{A} \otimes \dfrac{\mathbb{I}}{2} +\sigma_{A}\otimes  \dfrac{\mathbb{I}}{2} \nonumber \\
&=& 
\dfrac{1}{2}
\begin{pmatrix}
  \alpha + \gamma & 0 & 0 & 0 \\
  0 & -\alpha + \gamma & 0 & 0 \\
  0 & 0 & -\beta+ \delta & 0 \\
  0 & 0 & 0 & \beta + \delta
\end{pmatrix}.
\end{eqnarray}

\noindent Thus, \eqref{non_positive} is not positive semidefinite when either $\alpha >  \gamma$ or $\beta > \delta$
\end{proof}
We address the issue of positivity in the next sections. Density matrices are trace one hermitian matrices. It turns out that \eqref{qmo} preserves this property:
\vspace{3mm}
\begin{prop}
Let $\rho$ and  $\sigma_{\mathcal{J}}$ be multipartite quantum states, with $\mathcal{I}$ and $\mathcal{J}\subset\mathcal{I}$. The operator \eqref{qmo} is a trace one hermitian matrix.
\end{prop}
\begin{proof}
$Q_{\mathcal{J}}(\rho)$ is the sum of hermitian matrices, therefore it is also hermitian. Let us calculate the trace of \eqref{qmo}:
\begin{equation}
    \mathrm{Tr}[ Q_{\mathcal{J}}(\rho) ] = \mathrm{Tr}[ \rho ] -  \mathrm{Tr}[ \rho_{\mathcal{J}} ] + \mathrm{Tr}[ \sigma_{\mathcal{J}} ],
\end{equation}

\noindent but $\mathrm{Tr}[ \rho ] = \mathrm{Tr}[ \rho_{\mathcal{J}} ] = \mathrm{Tr}[ \sigma_{\mathcal{J}} ] = 1$. Thus, $\mathrm{Tr}[ Q_{\mathcal{J}}(\rho) ] = 1$.
\end{proof}

Another property of operator \eqref{qmo} is that the composition of two of them commute. 

\begin{prop}\label{Qcommute}
$Q_{\mathcal{J}_2} \circ Q_{\mathcal{J}_1}(\rho) = Q_{\mathcal{J}_1} \circ Q_{\mathcal{J}_2}(\rho) $ for any subsets $\mathcal{J}_1, \mathcal{J}_2 \subset\mathcal{I}$ and any quantum states $\rho,\sigma_{\mathcal{J}_1}$ and $\sigma_{\mathcal{J}_2}$.
\end{prop}
\begin{proof}
Let us first consider that the given marginals are compatible with some global density matrix $\sigma$, this is $\sigma_{\mathcal{J}_1} = \mathrm{Tr}_{\mathcal{J}_1^c}\left[ \sigma \right]$ and $\sigma_{\mathcal{J}_2} = \mathrm{Tr}_{\mathcal{J}_2^c}\left[ \sigma \right]$. Let us compute $Q_{\mathcal{J}_2} \circ Q_{\mathcal{J}_1}(\rho) = Q_{\mathcal{J}_2}(Q_{\mathcal{J}_1}(\rho) )$
\begin{multline}\label{commu}
Q_{\mathcal{J}_2} \circ Q_{\mathcal{J}_1}(\rho) = \rho - \mathrm{Tr}_{\mathcal{J}_2^c}\left[ \rho \right] - \mathrm{Tr}_{\mathcal{J}_1^c}\left[ \rho \right] + \mathrm{Tr}_{\mathcal{J}_2^c}\mathrm{Tr}_{\mathcal{J}_1^c}\left[ \rho \right] + \sigma_{\mathcal{J}_1} + \sigma_{\mathcal{J}_2} \\
+ \mathrm{Tr}_{\mathcal{J}_2^c}\mathrm{Tr}_{\mathcal{J}_1^c}\left[ \sigma \right].
\end{multline}
\noindent Note that \eqref{commu} is invariant under the interchange $\mathcal{J}_1 \leftrightarrow \mathcal{J}_2$. Therefore, 
\begin{equation}\label{commu1}
    Q_{\mathcal{J}_2} \circ Q_{\mathcal{J}_1}(\rho) = Q_{\mathcal{J}_1} \circ Q_{\mathcal{J}_2}(\rho).
\end{equation}
\end{proof}
We can use \eqref{commu1} to prove that the composition $Q_{\mathcal{J}_m}\circ\dots\circ Q_{\mathcal{J}_1}$ contains the set $\{ \sigma_{\mathcal{J}_i} \}_{i=1}^m$ of imposed quantum marginals. To do that, we use the notation $Q_{\mathcal{J}_1,\dots,\mathcal{J}_m}=Q_{\mathcal{J}_m}\circ\dots\circ Q_{\mathcal{J}_1}$. Here, we are considering overlapping and non-overlapping marginals.

\vspace{3mm}
\begin{prop}\label{prop:composition}
Given a set $\{ \sigma_{\mathcal{J}_i}\}_{i=1}^m$ of quantum marginals such that $\mathrm{Tr}_{\mathcal{J}_i^c}\left[ \sigma \right] = \sigma_{\mathcal{J}_i}$ for a density matrix $\sigma$, then the composition 
\begin{equation}\label{composition_qmo}
 Q_{\mathcal{J}_1,\dots,\mathcal{J}_m}(\rho),  
\end{equation}
 satisfy $   \mathrm{Tr}_{\mathcal{J}_i^c}\left[  Q_{\mathcal{J}_1,\dots,\mathcal{J}_m}(\rho) \right] = \sigma_{\mathcal{J}_i}$, for $i = 1 \ldots m.$
\end{prop}
\begin{proof}
By taking $\mathcal{J} = \mathcal{J}_2$ and $\rho_{\mathcal{I}}  = Q_{\mathcal{J}_1}(\rho)$ from Prop. \ref{qmo:imposition}, we know that $Q_{\mathcal{J}_2} ( Q_{\mathcal{J}_1}(\rho) )$ satisfies
$$\mathrm{Tr}_{\mathcal{J}_2^c}\left[ Q_{\mathcal{J}_2} \left( Q_{\mathcal{J}_1}(\rho) \right)  \right] = \sigma_{\mathcal{J}_2}.$$ Also, because of Prop. \ref{Qcommute} $\mathrm{Tr}_{\mathcal{J}_1^c}\left[ Q_{\mathcal{J}_2} \left( Q_{\mathcal{J}_1}(\rho) \right)  \right] = \mathrm{Tr}_{\mathcal{J}_1^c}\left[ Q_{\mathcal{J}_1} \left( Q_{\mathcal{J}_2}(\rho) \right)  \right] = \sigma_{\mathcal{J}_1}.$ In general, 
$$Q_{\mathcal{J}_1,\dots,\mathcal{J}_{m}} = Q_{\mathcal{J}_m}( Q_{\mathcal{J}_1,\dots,\mathcal{J}_{m-1}}) = \ldots = Q_{\mathcal{J}_{m-1}}( Q_{\mathcal{J}_1,\dots,\mathcal{J}_{m}}),$$
\noindent and therefore $\mathrm{Tr}_{\mathcal{J}_i^c}\left[  Q_{\mathcal{J}_1,\dots,\mathcal{J}_m}(\rho) \right] = \sigma_{\mathcal{J}_i}$, for $i=1,\ldots,m$.
\end{proof}

\newpage
\subsubsection{Derivation of expressions}\label{app:maps}
Here, we show some expressions obtained from applying \eqref{composition_qmo}. We consider the set of all $k$-body marginals for a global state formed by $N$ bodies. For example, for $N=3$, with labels in $\mathcal{I} = \{A,B,C\}$, the possible reductions to $2$ bodies are $\sigma_{AB}, \sigma_{AC}$ and $\sigma_{BC}$. Thus, from \eqref{composition_qmo} 
\begin{equation}
\begin{aligned}
Q_{AB,AC,BC}(\rho_{ABC}) &= \rho_{ABC} - \left( \rho_{AB} + \rho_{AC} + \rho_{BC} \right) + \left( \rho_A + \rho_B + \rho_C \right) \\ 
&+ \left( \sigma_{AB} + \sigma_{AC} + \sigma_{BC} \right) - \left( \sigma_A + \sigma_B + \sigma_C \right),
\end{aligned}
\end{equation}
\noindent with $\rho_{ABC} = \rho$. 

Expressions obtained in this way can be simplified by choosing $\rho = \mathbb{I}_{\mathcal{I}}$. Next we list some of them for different values of $N$ and $k$:

\begin{itemize}
\item $N = 2$, $k = 1$
\begin{equation}
    \label{N2k1}
    Q_{A,B}(\mathbb{I}_{\mathcal{I}}) = \sigma_A + \sigma_B - \mathbb{I}_{\mathcal{I}}.
\end{equation}
\item $N = 3$, $k = 1$
\begin{equation}
    \label{N3k1}
    Q_{A,B,C}(\mathbb{I}_{\mathcal{I}}) = \sigma_A + \sigma_B + \sigma_C - 2 \mathbb{I}_{\mathcal{I}}.
\end{equation}
\item $N = 4$, $k = 1$
\begin{equation}
    \label{N4k1}
    Q_{A,B,C, D}(\mathbb{I}_{\mathcal{I}}) = \sigma_A + \sigma_B + \sigma_C + \sigma_D - 3\mathbb{I}_{\mathcal{I}}.
\end{equation}
\item $N = 3$, $k = 2$
\begin{equation}
    \label{N3k2}
    Q_{AB,AC,BC}(\mathbb{I}_{\mathcal{I}}) = \mathbb{I}_{\mathcal{I}} + \left( \sigma_{AB} + \sigma_{AC} + \sigma_{BC} \right) - \left( \sigma_A + \sigma_B + \sigma_C \right).
\end{equation}
\item $N = 4$, $k = 2$
\begin{equation}
    \label{N4k2}
\begin{aligned}
Q_{AB, \ldots, CD}(\mathbb{I}_{\mathcal{I}}) = 3\mathbb{I}_{\mathcal{I}} + &\left( \sigma_{AB} + \sigma_{AC} +  \sigma_{AD} + \sigma_{BC} + \sigma_{BD}+ \sigma_{CD} \right) \\ 
& - 2 \left( \sigma_A + \sigma_B + \sigma_C + \sigma_D \right). \\
\end{aligned} 
\end{equation}\label{N5k2_formula}
\item $N = 5$, $k = 2$
\begin{equation}
\label{N5k2}
\begin{aligned}
Q_{AB, \ldots, DE}(\mathbb{I}_{\mathcal{I}}) = 6\mathbb{I}_{\mathcal{I}} &+ \left( \sigma_{AB} + \sigma_{AC} + \sigma_{AD} + \sigma_{AE} + \sigma_{BC} + \sigma_{BD} + \sigma_{BE} \right.\\ 
& \left. + \sigma_{CD} + \sigma_{CE} + \sigma_{DE} \right)  - 3\left( \sigma_A + \sigma_B + \sigma_C + \sigma_D + \sigma_E\right). \\
&
\end{aligned}
\end{equation}
\item $N = 4$, $k = 3$
\begin{equation}
\label{N4k3}
\begin{aligned}
Q_{ABC,\ldots, BCD}(\mathbb{I}_{\mathcal{I}}) &= \left( \sigma_{ABC} + \sigma_{ABD} + \sigma_{ACD} + \sigma_{BCD} \right)  - \left( \sigma_{AB} + \sigma_{AC} + \sigma_{AD} \right. \\  
& \left. + \sigma_{BC} + \sigma_{BD}+ \sigma_{CD} \right) + \left( \sigma_A + \sigma_B + \sigma_C + \sigma_D \right) - \mathbb{I}_{\mathcal{I}}
\end{aligned}
\end{equation}
\end{itemize}

\noindent Some of the expressions above suggest the following proposition:
\begin{prop}
Let $\mathcal{I}=\{ A, B, C, \ldots, \mathcal{I}_N \}$ and $\mathcal{J}=\{ AB, AC, \ldots, \mathcal{J}_{m} \}$ be the sets of labels associated to all one and two party marginals, respectively, with $|\mathcal{I}|=N$ and $|\mathcal{J}| =  2^{-1}N!/(N-2)!$. Given the set $\{ \sigma_j\}_{j=1}^{m}$ of all possible bipartite quantum marginals, the following formula holds
\begin{equation}\label{kNmarginals}
    Q_{AB, AC \ldots, \mathcal{J}_{m}}(\mathbb{I}_{\mathcal{I}})= \underset{  j \in \mathcal{J}  }{ \sum } \sigma_j - (N-2)\underset{ i \in \mathcal{I} }{ \sum } \sigma_i + \dfrac{(N-1)(N-2)}{2}\mathbb{I}_{\mathcal{I}}.
\end{equation}
\end{prop}
\begin{proof}
Let us show by induction this proposition. We saw above that \eqref{kNmarginals} holds for $N=3,4$ and $5$. Let us suppose that Eq. \eqref{kNmarginals} is valid for any $N$ and show that it is also true for $N+1$. Let us consider an $(N+1)$-qudit system and all its 2-party marginals and let $\mathcal{I}'$ and $\mathcal{J}'$ be the ordered sets of labels for the $N+1$ qudits and the $2^{-1}(N+1)!/(N-1)!$ 2-body marginals, respectively. In general, we can write 
\begin{equation}\label{kNplus1marginals}
    Q_{AB, AC \ldots, \mathcal{J}'_{\ell}}(\mathbb{I}_{\mathcal{I}'})=  \underset{  j \in \mathcal{J}'  }{ \sum } \alpha_j \sigma_j +  \underset{ i \in \mathcal{I}' }{ \sum } \beta_i \sigma_i + \gamma  \mathbb{I}_{\mathcal{I}'}.
\end{equation}
\noindent Because of the commutation property (see Prop. \ref{Qcommute}), the 2-body marginals in \eqref{kNplus1marginals} have to appear equally weighted, this is $\alpha = \alpha_j$. Also, since the 1-body marginals result from partial tracing an equal number of 2-body marginals, $\beta = \beta_i$. Thus, Eq. \eqref{kNplus1marginals} becomes
\bigskip
\begin{equation}\label{kNplus1marginals2}
    Q_{AB, AC \ldots, \mathcal{J}'_{\ell}}(\mathbb{I}_{\mathcal{I}'})= \alpha \underset{  j \in \mathcal{J}'  }{ \sum } \sigma_j+ \beta \underset{ i \in \mathcal{I}' }{ \sum } \sigma_i  + \gamma \mathbb{I}_{\mathcal{I}'}.
\end{equation}
\noindent Let us take partial trace of \eqref{kNplus1marginals2} over the last element in the set $\mathcal{I}'$, denoted $\mathcal{I}'_{N+1}$. We obtain 
\begin{eqnarray}\label{kNplus2marginals}
\mathrm{Tr}_{\mathcal{I}'_{N+1}} \left[ Q_{AB, AC,\ldots, \mathcal{J}'_{\ell}}(\mathbb{I}_{\mathcal{I}'}) \right]&=& \alpha \underset{  j \in \mathcal{J}  }{ \sum } \sigma_j + \alpha \underset{ i \in \mathcal{I} }{ \sum } \sigma_i  + \beta \underset{ i \in \mathcal{I} }{ \sum } \sigma_i + \beta \mathbb{I}_{\mathcal{I}} + \gamma \mathbb{I}_{\mathcal{I}}, \nonumber \\
&=& \alpha \underset{  j \in \mathcal{J}  }{ \sum } \sigma_j  + (\alpha + \beta) \underset{ i \in \mathcal{I} }{ \sum } \sigma_i + (\beta + \gamma) \mathbb{I}_{\mathcal{I}}.
\end{eqnarray}
\noindent Using the fact that $\mathrm{Tr}_{\mathcal{I}'_{N+1}} \left[  Q_{AB, AC,\ldots, \mathcal{J}'_{\ell}}(\mathbb{I}_{\mathcal{I}'}) \right] = Q_{AB, AC, \ldots, \mathcal{J}_{m}}(\mathbb{I}_{\mathcal{I}})$, we obtain $\alpha = 1$, $\beta = -(N+1-2)$ and $\gamma = (N+1 - 1)(N+1 -2)/2$. Thus, Eq. \eqref{kNplus1marginals} becomes
\begin{equation}\label{kNplus1marginals3}
    Q_{AB, AC \ldots, \mathcal{J}'_{\ell}}(\mathbb{I}_{\mathcal{I}'})= \underset{  j \in \mathcal{J}'  }{ \sum } \sigma_j - (N+1-2) \underset{ i \in \mathcal{I}' }{ \sum } \sigma_i + \dfrac{(N+1 - 1)(N+1 -2)}{2} \mathbb{I}_{\mathcal{I}'},
\end{equation}
\noindent and therefore \eqref{kNmarginals} holds. 

\end{proof}
\section{Numerical Study}\label{sec:numerical_study}

In the previous section we showed that $Q_{\mathcal{J}_1,\dots,\mathcal{J}_m}(\rho_0)$ is a trace one hermitian matrix containing the imposed marginals. Numerically, we have found that for some set $\{ \sigma_{\mathcal{J}_i} \}_{i=1}^m$ and initial state $\rho_0$, the composition $Q_{\mathcal{J}_1,\dots,\mathcal{J}_m}(\rho_0)$ is PSD. However, finding the correct initial state $\rho_0$, or even more, determining whether it exists for a given set of marginals, is not easy. Here, we want to show the occurrence of PSD matrices when choosing $\rho_0 = \mathbb{I}_{\mathcal{I}}$ in Eq. \eqref{composition_qmo} for different number of parties and local dimension $d$. First, we define the \emph{generator state}:

\begin{defi}[Generator state]\label{gen:qmp}
Given a set $\{ \mathcal{J}_i \}$ of $m$ labels, such that $\mathcal{J}_i \subset \mathcal{I}$, with $\mathcal{I}$ an alphabet of $N$ symbols, a quantum state $\rho_{gen}$ is called a generator state if $\mathrm{Tr}_{\mathcal{J}_i^c}[ \rho_{gen} ] = \sigma_{\mathcal{J}_i}$, for all $i = 1, \ldots, m$.
\end{defi}


For this study we considered $|\mathcal{J}_i| = k$, for $i=1, \ldots, m$. This is, all the quantum marginals are formed by $k$ bodies, each of them with local dimension $d$. Thus, the number of possible $k$-body marginals is 
\begin{equation}
M(N,k) =  \frac{N!}{k!(N-k)!}. 
\end{equation}


In Fig. \ref{positive_cases} we show the number of PSD matrices (denoted \textit{NPM}) resulting from $Q_{\mathcal{J}_1,\dots,\mathcal{J}_m}(\mathbb{I}_{\mathcal{I}})$ vs the number $m$ of quantum marginals, obtained from numerical simulations for different values of $N$ and $k$. In each of the plots in Fig. \ref{positive_cases} we compare results for $d=2$ (magenta solid line) and $d=3$ (blue dotted line). We see that NPM grows with the local dimension $d$ (see also Fig. \ref{fig:N4_k3_m4}, which shows NPM vs $d$). We found from these numerical simulations that for mixed states $\rho_{gen}$ chosen at random according to the Hilbert-Schmidt measure,  $Q_{\mathcal{J}_1,\dots,\mathcal{J}_m}(\mathbb{I}_{\mathcal{I}})$ outputs a full rank density matrix containing the given marginals, which turned out to be different to the generator state (mixed states are not uniquely determined by their quantum marginals). Thus, if it is known in advance that the marginals are compatible with a full rank global state, then we can easily find, with high probability, a global state  by simply applying the operator $Q_{\mathcal{J}_1,\dots,\mathcal{J}_m}$ to $\mathbb{I}_{\mathcal{I}}$. The chances are even greater for large local dimensions $d$, as numerical simulations suggest. The probability becomes smaller for quantum states close to the pure states, and for pure states $\rho_{gen}$ randomly drawn from the Haar measure \cite{Mezzadri_2007}, the occurrence of PSD matrices is zero for most of cases.



\begin{figure}[h!]
    \centering
    
    \subfloat[\label{fig:N4_k2}]{%
       \includegraphics[scale=0.418]{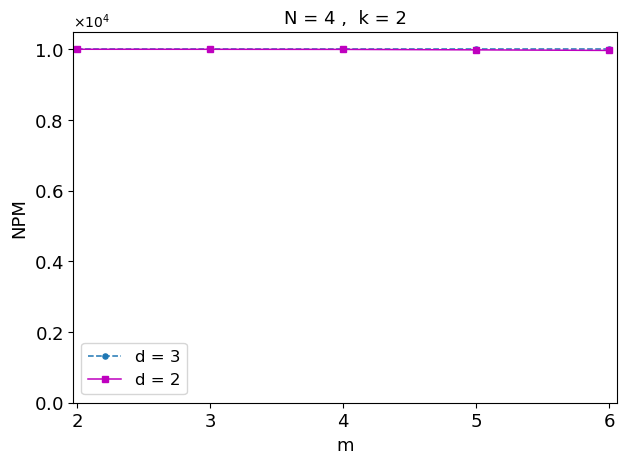}
     }
     \subfloat[\label{fig:N4_k3}]{%
       \includegraphics[scale=0.418]{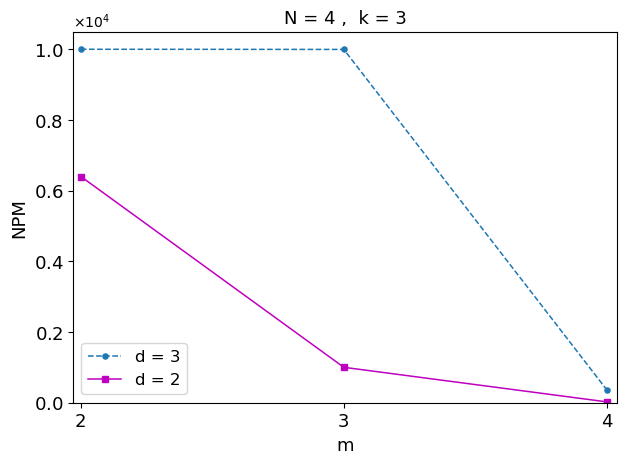}
     }
     
     \subfloat[\label{fig:N5_k2}]{%
       \includegraphics[scale=0.418]{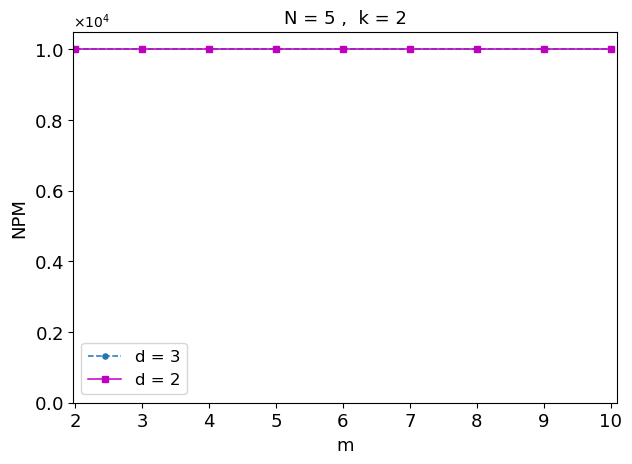}
     }
     \subfloat[\label{fig:N5_k3}]{%
       \includegraphics[scale=0.418]{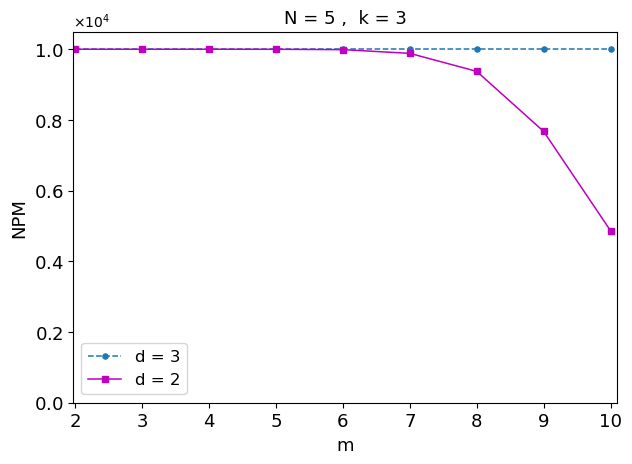}
     }
     
     \subfloat[\label{fig:N6_k4}]{%
       \includegraphics[scale=0.418]{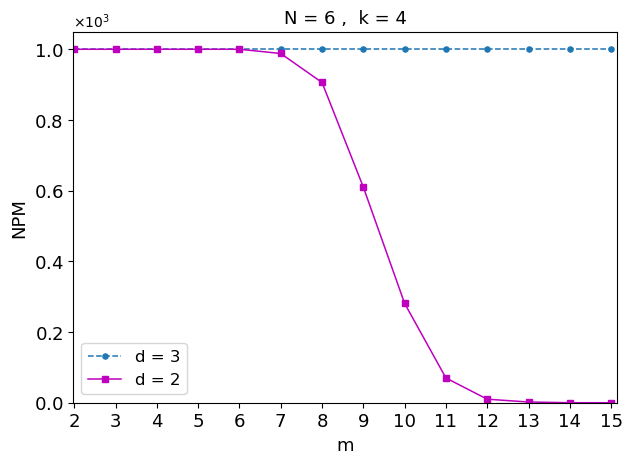}
     }
     \subfloat[\label{fig:N6_k5}]{%
       \includegraphics[scale=0.418]{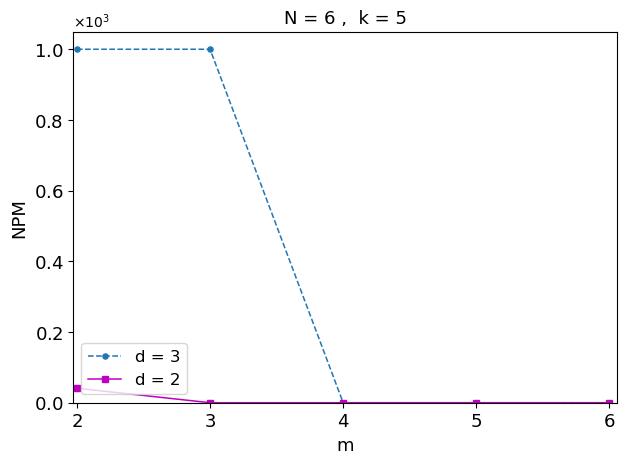}
     }

     \caption[QMP: Number of positive semidefinite cases vs number of input marginals]{ \textit{Number of positive semidefinite matrices (NPM)} vs the number $m$ of quantum marginals, for different values of $N$, $k$ and $d$, obtained from the composition $Q_{\mathcal{J}_1,\dots,\mathcal{J}_m}(\mathbb{I}_{\mathcal{I}})$. These results were obtained from a 10000 (1000 for the cases with $N=6$) full rank generator states drawn at random under the Hilbert-Schmidt measure \cite{Karol_2011}. For each $\rho_{gen}$, $m$ out of $M(N,k)$ possible quantum marginals are chosen at random.}
     \label{positive_cases}
\end{figure}


\begin{figure}
    \centering
    \includegraphics[scale = 0.47]{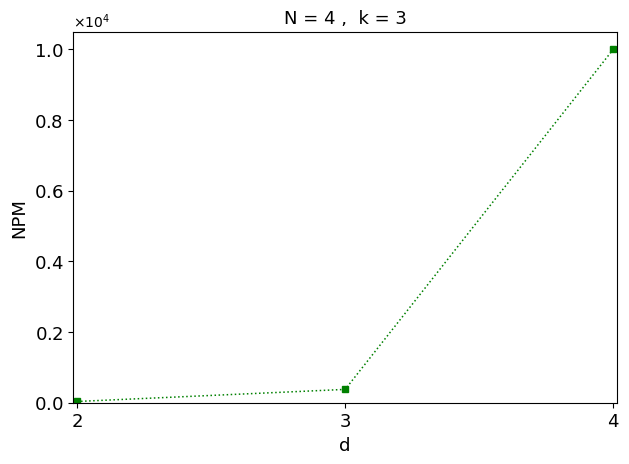}
    \caption[QMP: Number of positive semidefinite cases vs local dimension]{NMP vs $d$ considering all the possible 3-body marginals. For $d=2$, only 35 out of 10000 are PSD. For $d=4$, all the output matrices are PSD.  }
    \label{fig:N4_k3_m4}
\end{figure}

\section{Algorithm}
 
In section \ref{sec:numerical_study}, we saw that given a set of all possible $k$-body marginals calculated from a mixed generator state, the chances of finding a full rank PSD matrix compatible with those marginals from the composition $Q_{\mathcal{J}_1,\dots,\mathcal{J}_m}(\mathbb{I}_{\mathcal{I}})$ are, in many cases, high. In this section, we study a more restrictive version of the same problem. Let us consider an $N$-body quantum system of dimensions $d$ each. Given the set $\{ \sigma_{\mathcal{J}_i} \}_{i=1}^m$ of quantum marginals describing $m$ reduced part of the system, each of them formed by $k_i$ bodies, with $k_i = |\mathcal{J}_i|$, and given $\Vec{\lambda} = (\lambda_0, \ldots, \lambda_{d^N - 1})$, with $\lambda_i \geq 0$ and $\sum_{i=0}^{d^N} \lambda_i = 1$, we want to study the following question:
\begin{quote}
    \textit{Given a set $\{ \sigma_{\mathcal{J}_i} \}_{i=1}^m$ of quantum marginals and spectra $\Vec{\lambda}$, is there any density matrix $\rho$ compatible with this information?}
\end{quote}
Here we are considering both, overlapping and non-overlapping sets $\mathcal{J}_i$. Instead of the spectra, we may also consider rank constraint. Thus, the problem consists in finding a density matrix of rank equal or smaller than $r$, if it exists, compatible with the given marginals. 

In Prop. \ref{prop:composition} was shown that applying the composition $Q_{\mathcal{J}_1,\dots,\mathcal{J}_m}(\rho)$ produces a trace one hermitian matrix $\rho'$, compatible with the $m$ marginals. Also, we saw that $\rho'$ is not necessarily PSD, which can be checked by computing its \textit{spectral decomposition} and looking for negative eigenvalues. The spectral decomposition of $\rho'$ consists of the following factorization
\begin{equation}\label{spectral_decomposition}
    \rho' = U D U^{\dagger},
\end{equation}
with $U$ a unitary matrix and $D = diag(\mu_0, \ldots, \mu_{d^N - 1})$ a diagonal matrix whose entries are the \textit{eigenvalues} of $\rho'$. We can \textit{impose} the prescribed eigenvalues $\Vec{\lambda}$ into $\rho'$ simply by substituting $\Lambda  = diag(\lambda_0, \ldots, \lambda_{d^N - 1})$ in place of $D$. This is
\begin{equation}
    \rho' \longrightarrow \rho'' = U \Lambda U^{\dagger}.
\end{equation}
For the rank constraint case, we only need to keep the $r$ largest positive eigenvalues, set the rest of them to zero and then normalize the resulting matrix. Either way, \textit{part} of the information that was previously imposed about the marginals is lost; this is, we will probably find that $\sigma_{\mathcal{J}_i} \neq \mathrm{Tr}_{\mathcal{J}_i^c}[ \rho'' ]$. Then, we can apply $Q_{\mathcal{J}_1,\dots,\mathcal{J}_m}(\rho'')$ to impose the marginals $\{ \sigma_{\mathcal{J}_i} \}_{i=1}^m$ in $\rho''$, but it will cause the loss of part of the previously imposed eigenvalues. Intuitively, we realize that the loss of the marginals and spectra will reduce each time the whole process described above is repeated, until it finally converges to a matrix, if it exists, with the prescribed marginals and spectra. We summarize this procedure in algorithm \ref{alg:qmp}. \\

\vspace{4mm}
\begin{algorithm}[H]\caption{QMP algorithm.}\label{alg:qmp}
\begin{algorithmic}
\Require $\{ \sigma_{\mathcal{J}_i} \}_{i=1}^m$ and  $\vec{\lambda}$ (or prescribed rank)
\Ensure $\rho'' \in B(\mathcal{H})$ with spectra $\vec{\lambda}$, satisfying  $\sigma_{\mathcal{J}_i} = \mathrm{Tr}_{\mathcal{J}_i^c}[ \rho'' ]$
\State{$\rho_{0}$ (Random density matrix of dimension $d^N$) \\ $n = 0$}\\
\noindent \textbf{repeat}\\
\hspace*{0.5cm}$\rho' \leftarrow Q_{\mathcal{J}_1,\dots,\mathcal{J}_m}(\rho_0) = U D U^{\dagger}$\\
\hspace*{0.5cm}$\rho''  \leftarrow  U \Lambda U^{\dagger}$ \\
\hspace*{0.5cm}$\rho_{0} \leftarrow \rho''$\\
\noindent \hspace*{0.5cm}$n \leftarrow n + 1$\\
\noindent \textbf{until} $ \mathcal{D}_{T} \leq\epsilon$ or $n = \mathrm{Max\;number\;of\;iterarions}$ 
\end{algorithmic}
\end{algorithm}
\vspace{4mm}

\newpage
To study the convergence of algorithm \ref{alg:qmp} we consider the following metrics:
\begin{itemize}
\item[i)] \textit{Marginals Distance $\mathcal{D}_M$:}
\begin{equation}\label{eq7}
\mathcal{D_{M}} = \left( \dfrac{1}{m} \sum_{i=1}^{m}\mathfrak{D}(\sigma_{ \mathcal{J}_i }, \sigma''_{ \mathcal{J}_i })^2 \right)^{1/2},
\end{equation}
\noindent where $\mathfrak{D}(\sigma_{ \mathcal{J}_i }, \sigma''_{ \mathcal{J}_i }) =  \sqrt{\mathrm{Tr}\bigl[(\sigma_{ \mathcal{J}_i } - \sigma''_{ \mathcal{J}_i })(\sigma_{ \mathcal{J}_i } - \sigma''_{ \mathcal{J}_i })^{\dagger} \bigr]}$ is the Hilbert-Schmidt distance between the given marginal $\sigma_{ \mathcal{J}_i }$ and $ \sigma''_{ \mathcal{J}_i } = \mathrm{Tr}_{\mathcal{J}_j^{c}}\bigl( \rho''  \bigr)$. 
\item[ii)] \textit{Spectral Distance $\mathcal{D}_{\lambda}$:}
\begin{equation}
\mathcal{D}_{\lambda} = \| \vec{\mu} - \vec{\lambda} \|,
\end{equation}
\noindent where $\mathcal{D}_{\lambda}$ is the euclidean distance between $\vec{\mu} = (\mu_0, \ldots, \mu_{d^N - 1})$ and $\vec{\lambda} = (\lambda_0, \ldots, \lambda_{d^N-1})$.
\item[iii)]\textit{Overall Distance $\mathcal{D}_{T}$:}
\begin{equation}
   \mathcal{D}_{T} =  \sqrt{ \mathcal{D}_M^2 + \mathcal{D}_{\lambda}^2}.
\end{equation}
\end{itemize}

$\mathcal{D}_M$ and $\mathcal{D}_{\lambda}$ tell us how close the algorithm is getting to a state with the prescribed marginals and spectra, respectively. Algorithm \ref{alg:qmp} will run until the overall distance $\mathcal{D}_{T}$ is smaller than the given tolerance $\epsilon$, or until an specified number of iterations is reached. 

In Fig. \ref{fig:N4d3} we show the convergence for a pure $4$-party density matrix with local dimensions $d=3$. The vertical axis corresponds to the metrics, given in log scale, whereas the horizontal axis shows the number $n$ of iterations. In addition to $\mathcal{D}_M$ and $\mathcal{D}_{\lambda}$, the Hilbert-Schmidt distance $\mathcal{D}_{G}$ between $\rho''$ and $\rho_{gen}$ is shown. We see that $\mathcal{D}_{G}$ tends to zero, which means that $\rho''$ approaches to $\rho_{gen}$ in each iteration. Since, in this case, the marginals are compatible only with the generator state,  $\rho''$ will always approach to $\rho_{gen}$ for any initial seed $\rho_0$. Fig. \ref{fig:N4d3} shows the solution when $r = 1$ and Fig. \ref{fig:rN4d3} for $r=2$. Fig. \ref{fig:rN4d3} reflects the fact that a rank-2 density matrix with the same marginals than a rank-1 $\rho_{gen}$ does not exist. The nonexistence of this state produces instability in the algorithm, as the rank-2 restriction puts $\rho''$ away from a density matrix with the desired marginals. Despite the fact that a rank-1 solution exists, i.e. the generator state, it cannot be reached when considering a seed taken at random among the entire set of density matrices. Such a kind of instability is only observed when there is no quantum state of rank $r$ with the imposed marginals. 

\begin{figure}[ht!]
    \centering
     \subfloat[\label{fig:N4d3}]{%
       \includegraphics[scale=0.465]{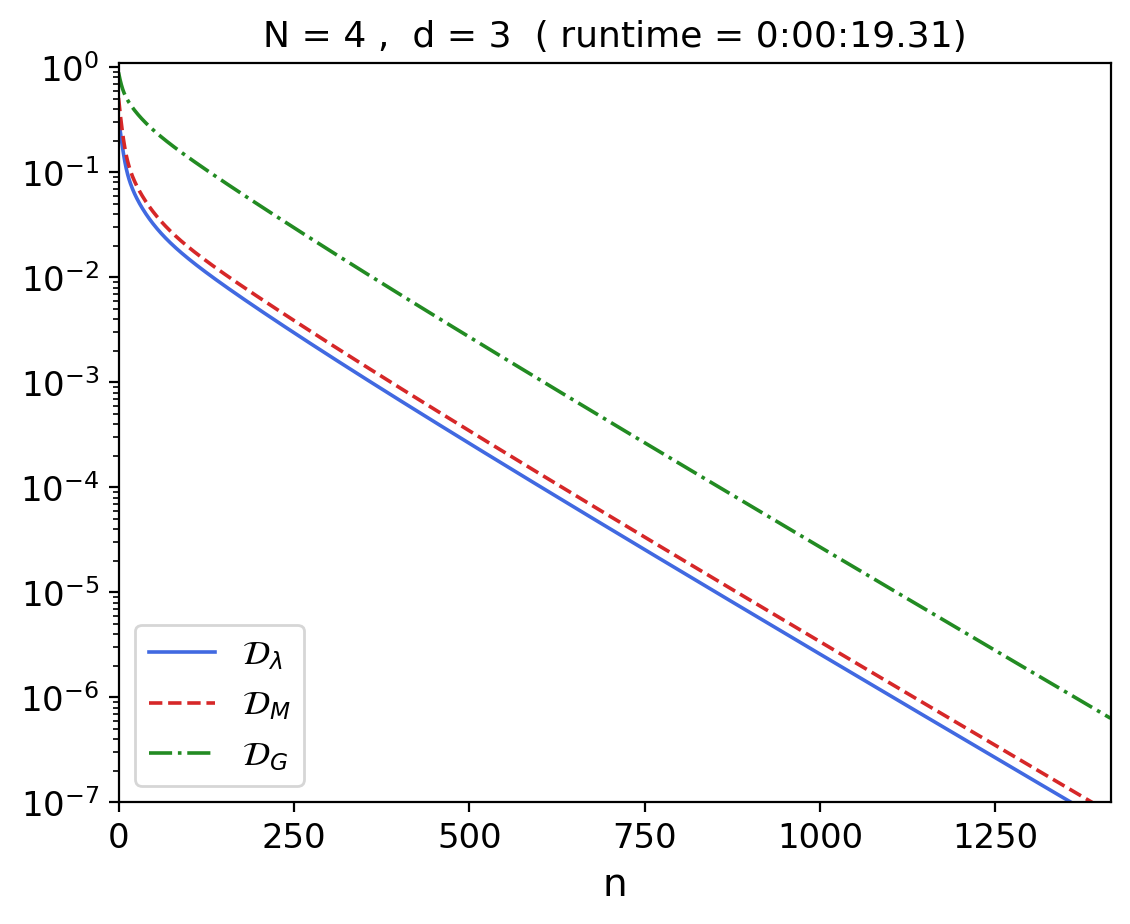}
     }
     \subfloat[\label{fig:rN4d3}]{%
       \includegraphics[scale=0.465]{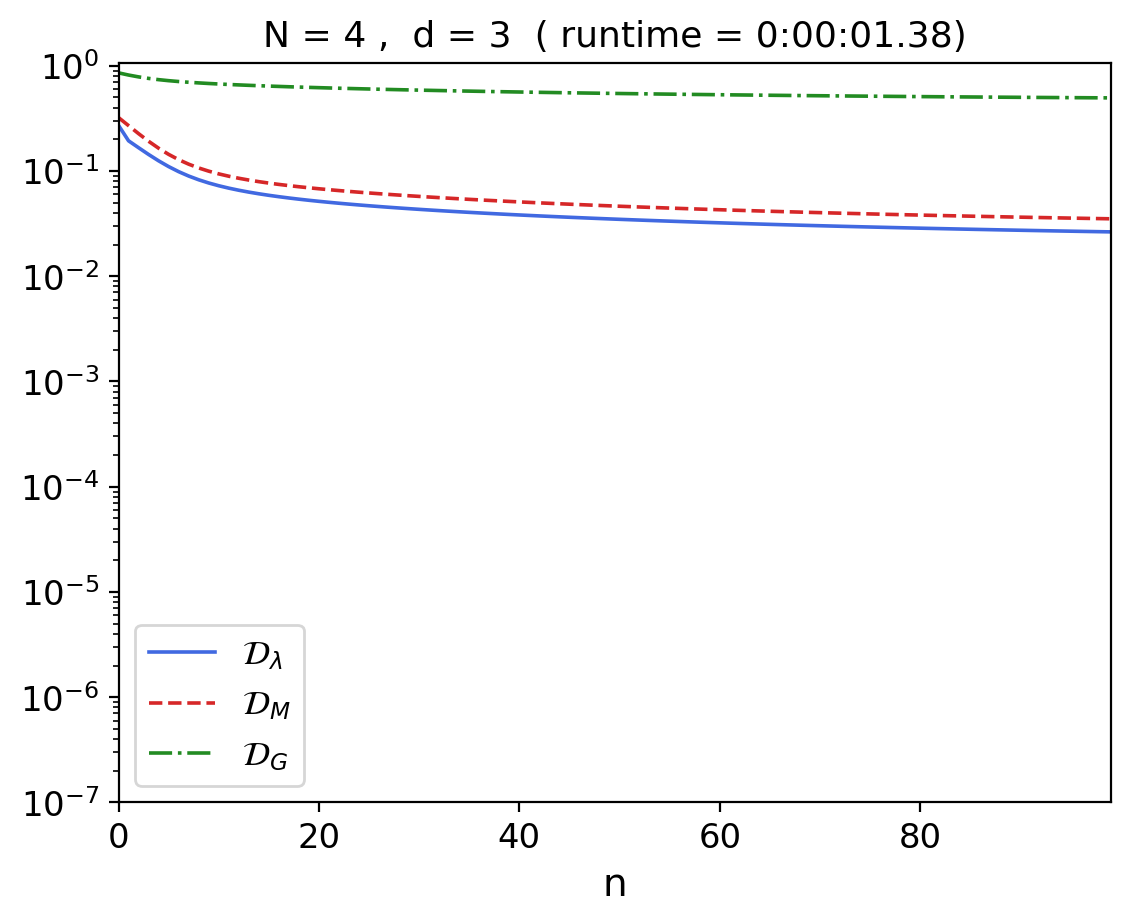}
     }
     \caption[QMP: Convergence of algorithm \ref{alg:qmp} with rank constraint]{Convergence of a $4$-party case, each party of dimension $d = 3$. The prescribed marginals consist of all possible reductions to 2 parties calculated from a pure generator state $\rho_{gen}$. Here, $\mathcal{D}_{\lambda}$, $\mathcal{D}_{M}$ and $\mathcal{D}_{G}$ vs $n$ (the number of iterations) were plotted. We show the convergence when a) a rank one is prescribed and b) the behavior of algorithm \ref{alg:qmp} when rank two is prescribed.  $\mathcal{D}_{G}$ is the Hilbert-Schmidt distance between $\rho_{gen}$ and $\rho''$. The runtime is given in format \textit{hh:mm:ss}.}
\end{figure}

The input marginals can have distinct number of parties. In Fig. \ref{fig:mixed_3p2p} we show the convergence of a $5$-party density matrix $\rho_{ABCDE}$ of rank 4, with inputs marginals $\sigma_{ABD}$ and $\sigma_{BC}$, which were computed from a full rank generator state of $5$ parties and local dimensions $d = 3$. The distance $\mathcal{D}_{G}$ experiences very small oscillations (unnoticeable in the plot) around some value, but the closer $\mathcal{D}_{T}$ gets to zero the smaller the amplitude of these oscillations.

In Fig. \ref{fig:convergence_full_rank} the convergence for a case with prescribed spectra is shown. The eigenvalues and the marginals $\sigma_{AC}, \sigma_{AD}, \sigma_{BC}$ and $\sigma_{CD}$ were computed from a full rank mixed $\rho_{gen}$. Algorithm \ref{alg:qmp} converges to a state compatible with the inputs, but the state found is different to $\rho_{gen}$.
\begin{figure}[ht!]
    \centering
     \subfloat[\label{fig:mixed_3p2p}]{%
      \includegraphics[scale = 0.465]{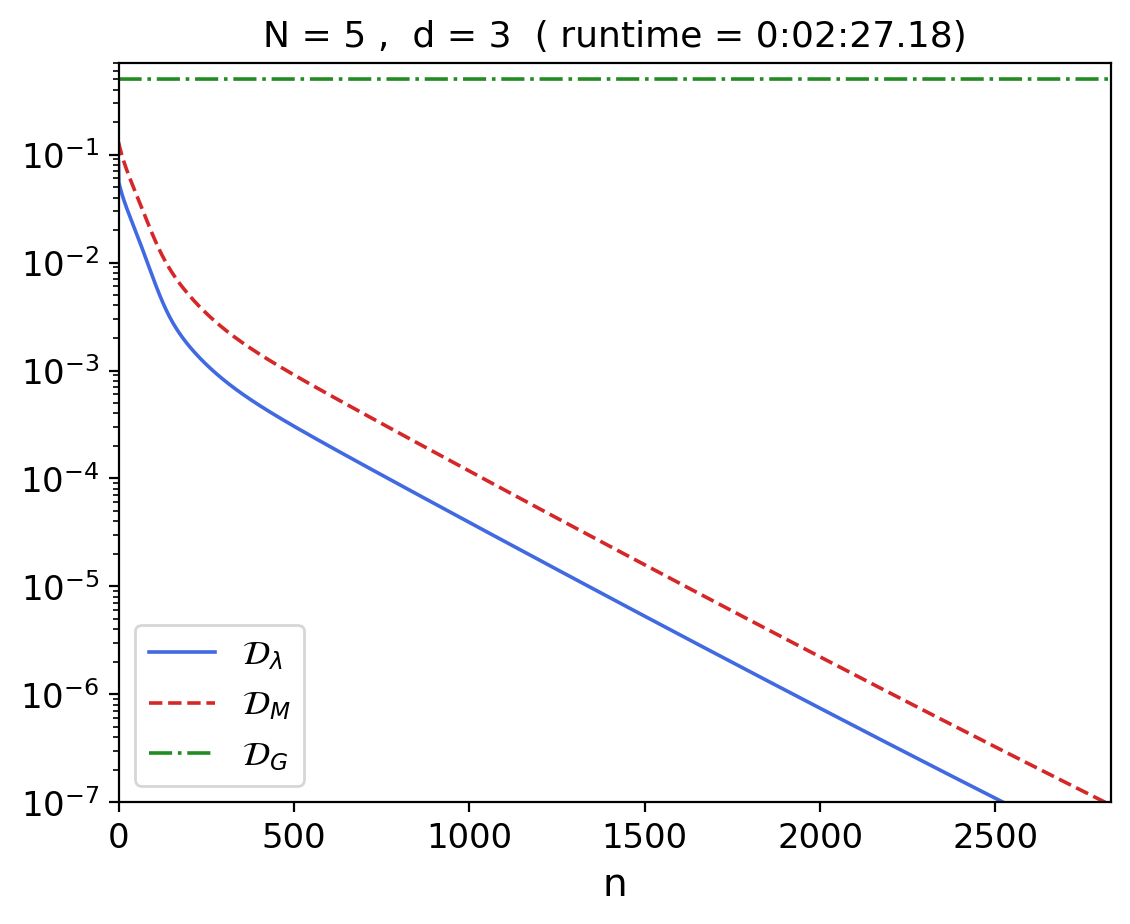}
     }
     \subfloat[\label{fig:convergence_full_rank}]{%
      \includegraphics[scale = 0.465]{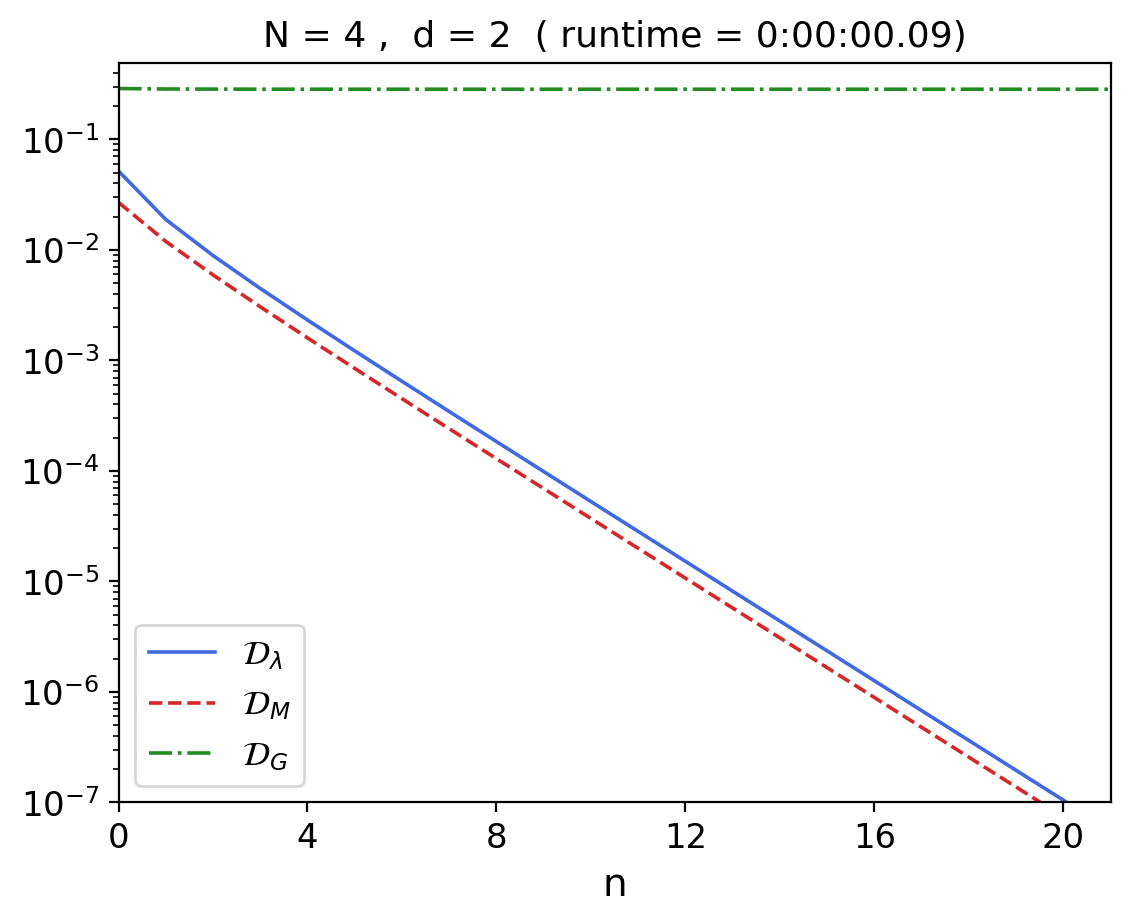}
     }
     \caption[QMP: Convergence of algorithm 2 with prescribed spectra]{a) Convergence of a 5-body density matrix $\rho_{ABCDE}$ of rank four, with input marginals $\sigma_{ABD}$ and $\sigma_{BC}$ calculated from a full rank generator state, and local dimension $ d = 3$. b) Convergence of algorithm \ref{alg:qmp} for a prescribed spectra case. The prescribed eigenvalues are the same as those of the generator sate.}
\end{figure}

In general, there are some seeds $\rho_0$ for which algorithm \ref{alg:qmp} converges in fewer iterations than for other seeds. However, to determine the specific conditions on $\rho_0$ for faster convergence is not easy. 



\subsubsection{Absolutely Maximally Entangled States}

A pure quantum state of $N$ parties with local dimension $d$ is called \textit{Absolutely Maximally Entangled (AME)}, denoted $AME(N,d)$, if all its possible $\lfloor{  N/2 \rfloor}$-body marginals are maximally mixed , with $\lfloor{  \ldots \rfloor}$ the \textit{floor} function. AME states have been used for quantum error-correcting codes \cite{Pastawski_2015}, quantum secret sharing \cite{Helwig_2012} and teleportation protocols \cite{helwig2013absolutely}. The question of whether AME states exist for given $N$ and $d$ is still an open problem. A summary of the existence of AME states is given in \cite{AMEs}.

Here, we test algorithm \ref{alg:qmp} for some AME states. In Fig. \ref{fig:ame42} we show the case with $N=4$ and $d=2$. As expected, it cannot find an AME(4,2) state; Higuichi et al. showed in \cite{HIGUCHI2000213} that such a state does not exist. The convergence of the AME(4,3) is shown in \ref{fig:ame43}. The \textit{plateau} seen in this last case is very common for AME states. For more parties and larger local dimensions, this plateau is much more prolonged, resulting in very large runtimes. For example, for cases such as the AME(4,5) and the AME(4,6), we allowed the algorithm to run for about two weeks, apparently going under the plateau stage, but it was unable to converge. In Fig. \ref{fig:accl_ame44} (solid line) the convergence for the AME(4,4) is shown; this case also exhibits a plateau. It took to the algorithm about 6 and half minutes to find the AME(4,4), whereas for other seeds it was unable to converge to this quantum state after 12 hours of running. 



\begin{figure}[h!]
    \centering
     \subfloat[\label{fig:ame42}]{%
       \includegraphics[scale = 0.465]{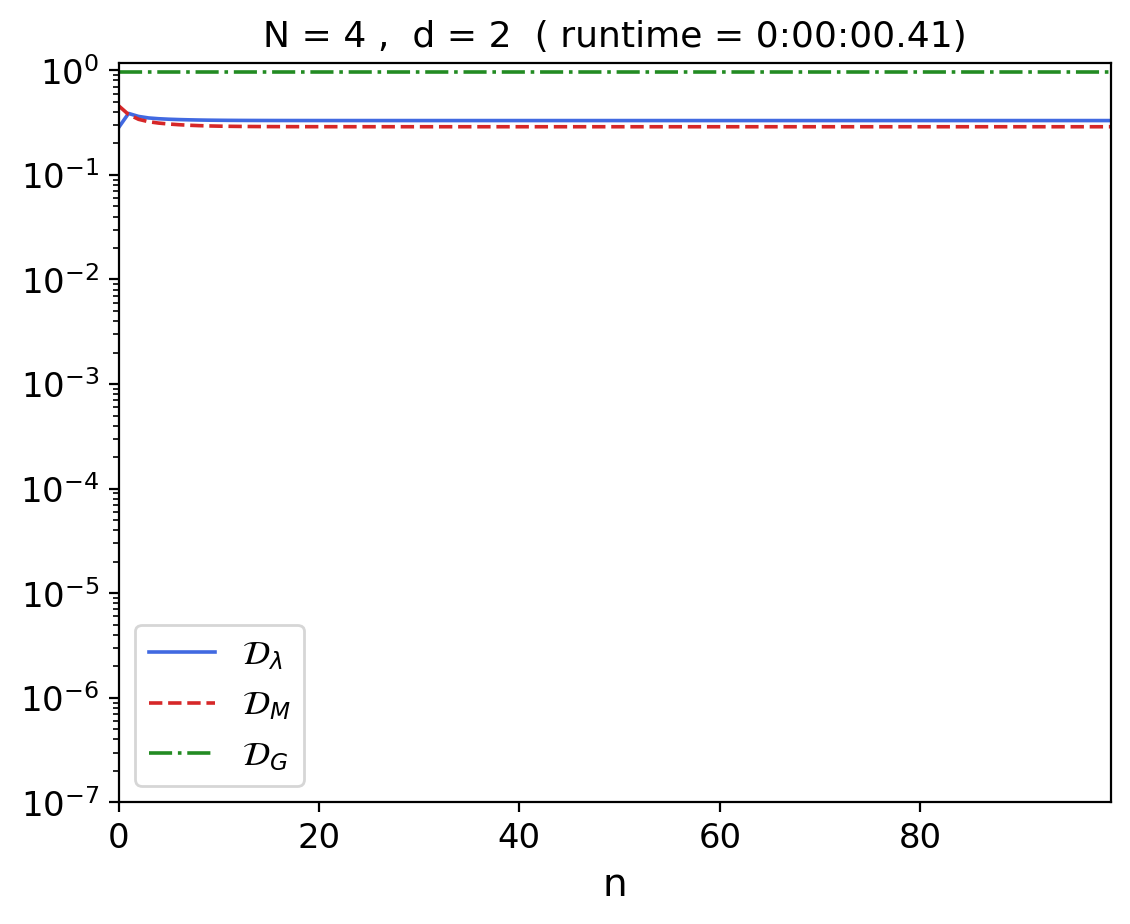}
     }
     \subfloat[\label{fig:ame43}]{%
       \includegraphics[scale = 0.465]{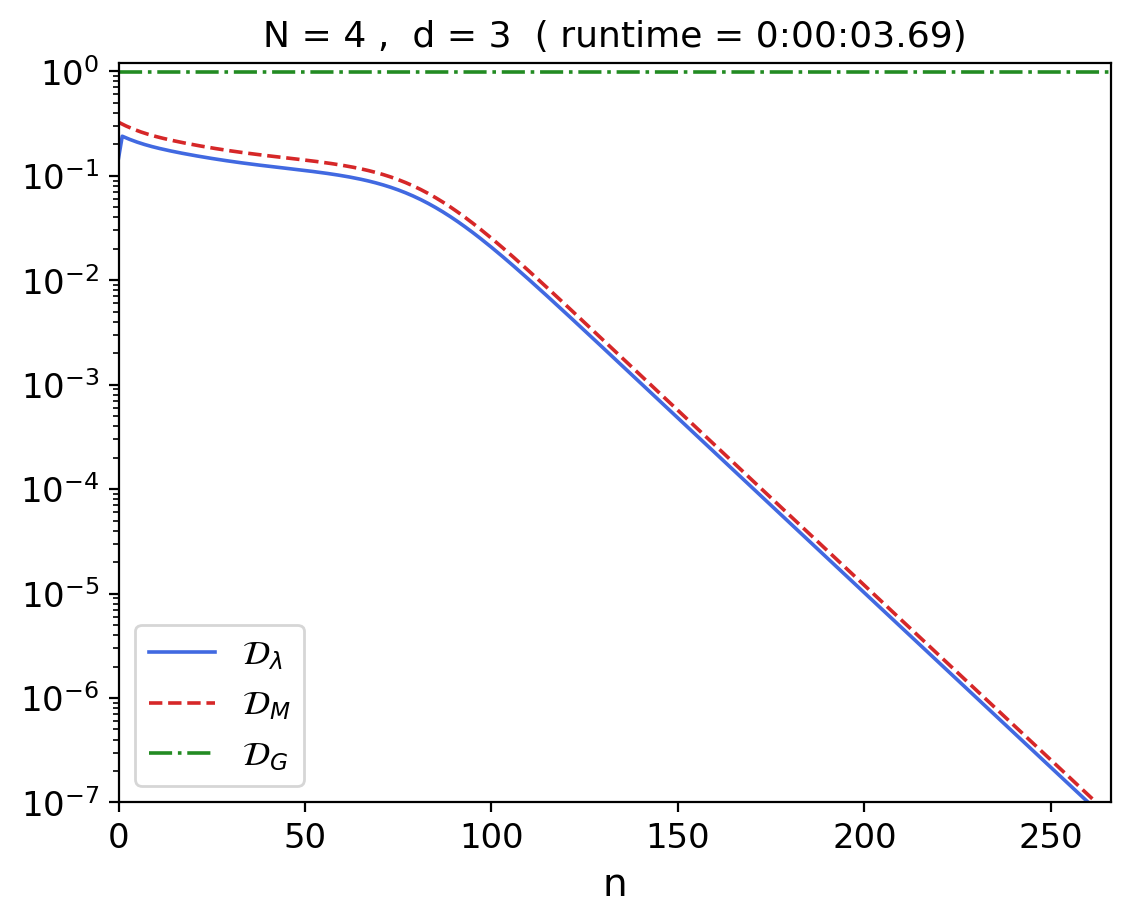}
     }
     \caption[QMP: Convergence of algorithm 2 for AME(4,3)]{ a) 100 iterations of algorithm \ref{alg:qmp} considering all the possible 2-body marginals, all of them equal to $\mathbb{I}/4$. In this case, the algorithm is trying to find a $4$-parties pure quantum state, with local dimension $d=2$, compatible with the given marginals; such a state does not exist. b) Convergence for the AME(4,3) state.}
\end{figure}

\subsubsection{Acceleration of the algorithm}
Algorithm \ref{alg:qmp} is a heuristic approach, which can be seen as a \textit{fixed point algorithm.} A well-known algorithm for fixed point problems is the one developed by Halpern \cite{halpern_1967}. Let $T$ be a nonexpansive mapping on a Hilbert space $\mathcal{H}$. For $x_0 \in \mathcal{H}$, the Halpern algorithm consists of the sequence
\begin{equation}\label{halpern}
    x_{n + 1} := \alpha_n x_0 + (1 - \alpha_n)T(x_n), \quad n \in \mathbb{N},
\end{equation}
\noindent where $\alpha_n \in \mathbb{R}$ satisfies the conditions: 
\begin{equation}\label{conditions_halpern}
 lim_{n \rightarrow \infty} \alpha_n = 0, \;\; \sum_{n=0}^{\infty} \alpha_n = \infty \;\; \text{and  }\sum_{n = 0}^{\infty}|\alpha_{n+1} - \alpha_n | < \infty.   
\end{equation}
\noindent Halpern showed that \eqref{halpern} strongly converges to a fixed point. In Ref. \cite{Sakurai_2014} an strategy to accelerate the Halpern algorithm was introduced, which consists of the modified sequence:
\begin{align}\label{accelerated_halpern}
    z_{n+1} &:= \dfrac{1}{\alpha}(T(x_{n}) - x_n) + \beta_n z_n, \nonumber \\
    y_n & := x_n + \alpha z_{n+1}, \\
    x_n &:= \mu \alpha_n x_0 + (1 - \mu \alpha_n)y_n. \nonumber
\end{align}
\noindent with $z_0 =\alpha^{-1}(T(x_{0}) - x_0)$, $\mu \in (0,1]$ and $\alpha > 0$. Besides conditions \eqref{conditions_halpern}, $\beta_n \leq \alpha_n^2$ also has to be satisfied. When $\beta_n = 0$ and $\mu = 1$, \eqref{accelerated_halpern} becomes \eqref{halpern}.  

We took some of the ideas from Ref. \cite{Sakurai_2014} and implemented them to accelerate algorithm \ref{alg:qmp}; we show this changes in the algorithm \ref{alg:qmp_accelerated}. Any strategies to update $\alpha_n$ and $\beta_n$ are possible, as long as they satisfy the conditions above. For our numerical simulations we chose $\alpha_n = (n/10^5 + 1)^{-\alpha}$, which is a slightly modified version of $\alpha_n$ from a numerical example in Ref. \cite{Sakurai_2014}, and $\beta_n = \alpha_n^2$. With $\beta_n = 0$, $\mu = 1$ and $\alpha_n = 1$, algorithm \ref{alg:qmp_accelerated} reduces to algorithm \ref{alg:qmp}.

In Fig. \ref{fig:accerated_qmo} we compare algorithms \ref{alg:qmp} and \ref{alg:qmp_accelerated} when starting from the same $\rho_0$. The lines correspond to the overall distance $\mathcal{D}_{T}$. The dashed and dotted lines show the convergence of algorithm \ref{alg:qmp_accelerated} for different choices of $\alpha$ and $\mu$. Fig. \ref{fig:accl_ame44} corresponds to the AME(4,4); at first, the lines are converging at about the same rate, but suddenly algorithm \ref{alg:qmp_accelerated} starts decaying faster. Fig. \ref{fig:accl_pure_44} corresponds to a pure state, randomly drawn from the Haar measure. 

\vspace{4mm}
\begin{algorithm}[H]\caption{QMP algorithm.}\label{alg:qmp_accelerated}
\begin{algorithmic}
\Require $\{ \sigma_{\mathcal{J}_i} \}_{i=1}^m$ , $\vec{\lambda}$ (or prescribed rank), $\alpha$ and $\mu$.
\Ensure $\rho'' \in B(\mathcal{H})$ with spectra $\vec{\lambda}$, satisfying  $\sigma_{\mathcal{J}_i} = \mathrm{Tr}_{\mathcal{J}_i^c}[ \rho'' ]$.
\State{$\rho_{0}$ (Random density matrix of dimension $d^N$)}\\

\noindent $n = 0$\\

\noindent \textbf{Repeat}\\

\noindent \hspace*{0.5cm}$\alpha_n \leftarrow (n/10^{5} + 1)^{-\alpha}$\\

\noindent \hspace*{0.5cm}$\beta_n \leftarrow \alpha_n^2$\\

\noindent \hspace*{0.5cm} \textbf{For} $i = 1, \ldots, m$:\\
\noindent \hspace*{0.99cm} $z_i \leftarrow \alpha^{-1} \left(Q_{\mathcal{J}_{i} }(\rho_0) - \rho_0 \right)   \nonumber$ \\
\noindent \hspace*{0.99cm} $z_{i + 1} \leftarrow z_i +  \beta_n z_{i}  \nonumber$ \\

\noindent \hspace*{0.99cm} $y_{i} \leftarrow \rho_0 + \alpha z_{i+1} $ \\

\noindent \hspace*{0.99cm}  $\rho_0 \leftarrow \mu \alpha_n \rho_0 + (1 - \mu \alpha_n) y_{i} $\\

\noindent \hspace*{0.5cm} \textbf{end for}\\

\noindent \hspace*{0.5cm} $\rho' \leftarrow \rho_0 = U D U^{\dagger}$\\

\noindent \hspace*{0.5cm} $\rho''  \leftarrow  U \Lambda U^{\dagger}$ \\

\noindent \hspace*{0.5cm}$\rho_{0} \leftarrow \rho''$\\

\noindent \hspace*{0.5cm}$n \leftarrow n + 1$\\

\noindent \textbf{until} $ \mathcal{D}_{T} \leq\epsilon$ or $n = \mathrm{Max\;number\;of\;iterarions}$ 
\end{algorithmic}
\end{algorithm}
\vspace{4mm}

\begin{figure}[h!]
    \centering
     \subfloat[\label{fig:accl_ame44}]{%
      \includegraphics[scale = 0.466]{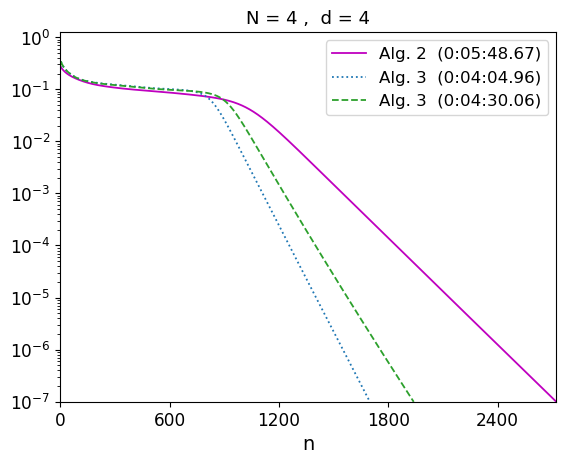}
     }
     \subfloat[\label{fig:accl_pure_44}]{%
      \includegraphics[scale = 0.466]{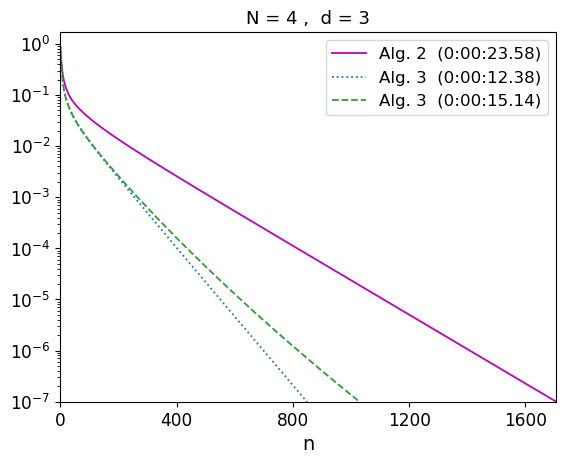}
     }
     \caption[QMP: Convergence of algorithm 3 for the AME(4,4)]{ Comparison between algorithm \ref{alg:qmp} (dashed line) and \ref{alg:qmp_accelerated} (dashed and dotted lines).  For the dotted line $\alpha = 1$ and $\mu = 10^{-10}$ and for the dashed line $\alpha = 50$ and $\mu = 5\times10^{-5}$. a) Convergence for the AME(4,4) and b) convergence for a pure random quantum state, with $N=4$ and $d=3$. }
     \label{fig:accerated_qmo}
\end{figure}

The parameters $\alpha$, $\beta_n$, $\mu$ and $\alpha_n$ can be sensitive to the local dimension, the number of bodies in the global system and the number of bodies in the marginals. Thus, one might have to tune these parameters for different values of $N$, $d$ and $k$. Also, the performance for a set of parameters might depend on the initial seed $\rho_0$. Nonetheless, for given $N$, $d$ and $k$, we have noticed improvement in the performance for a wide range of quantum states with the same choice of parameters. 
The algorithms were implemented in Python, see Appendix  \hyperref[app:C]{C}, and are available on Github \cite{qmp_codes}. All the calculations in this chapter were made on an AMD Ryzen 5 4500U laptop with 16GB of RAM.


\vspace{-3mm}
\section{Conclusions}

In this chapter, we addressed the problem of compatibility in the Quantum Marginal Problem. We introduced an operator to impose a set of quantum marginals in a global matrix. We saw that, in many cases, for quantum marginals generated from a full rank mixed state, drawn from a uniform distribution, the composition of these operators outputs a full rank density matrix compatible with the marginals. We applied the operator to develop a heuristic algorithm for finding a global density matrix, if exists, compatible with a prescribed spectra and marginals. For quantum states randomly taken from uniform distributions, numerical simulations shows that the algorithm is able to find a solution for most of cases. It was also able to find many cases of AME states whose existence is known. However, for large number of parties and local dimension, the algorithm failed to find AME states in a reasonable amount of time, even though we showed that it is possible to accelerate it. For open AME cases (see Ref. \cite{AMEs} ), such as the AME(8,7) and AME(11,3), it was not possible to run the algorithm due to the required memory resources. 

\newpage
\noindent \text{}
\thispagestyle{plain}

    

    
    
    
    
    
    
    
    

    \makeatletter
        \def\toclevel@chapter{-1}
        \def\toclevel@section{0}
    \makeatother
    

\appendix 

\addcontentsline{toc}{chapter}{Appendix}

    \chapter*{Appendix}\label{appC}
    \pagestyle{plain}

\vspace{-1cm}
\section*{A: Codes for the quantum state estimation algorithm}
\label{app:A}
Here we show the codes for the quantum state estimation algorithm presented in \ref{cap2}. The algorithm was implemented in Python and the code is available in Github \cite{pio_gitub}. \\
\begin{lstlisting}[language=Python, basicstyle=\scriptsize]
import numpy as np
import time
from scipy.optimize import root
import enum
import qiskit.quantum_info as qi
import itertools as it
import matplotlib.pyplot as plt
import functools as ft
import random

def QSE(dimension, Measurements,Probabilities,maxNumIterations,min_hsDistance):
    
    """ 
    Estimate a density matrix for "Probabilities" obtained from "Measurements" 
    
    """
    
    start_time_pio = time.time()
    rho_0 = np.eye(dimension)/dimension
    hsDistance = 2
    n = 0
    rho_estimate = 0
    while hsDistance > min_hsDistance and n < maxNumIterations:
        rho_1 = rho_0
        rho_0 = applyCompositionT(rho_1,Measurements,Probabilities)
        rho_0 = projectToDensityOperator(rho_0,dimension)
        rho_estimate += rho_0
        hsDistance = HS_distance(rho_0,rho_1)
        n += 1
        
    rho_estimate = rho_estimate/np.trace(rho_estimate)
    dt_pio = time.time() - start_time_pio

    return rho_estimate, dt_pio
    

def bootStrapping(rho_target,dimension,measurements,numOfSeeds,numOfCopies):
    
    Fidelities = []
    white_noise_level = 0.1
    rhoNoisy = (1 - white_noise_level)*rho_target + white_noise_level*np.identity(dimension)/dimension
    Probabilities = probability_distributions(measurements, dimension, rhoNoisy, numOfCopies)
    numOfbases = len(measurements)
    
    min_hsDistance = 1e-5
    maxNumIterations = 5
    
    for _ in range(numOfSeeds):
        Sprobabilities = boos_probs(Probabilities,numOfCopies,numOfbases,dimension)
        rho_estimate,_ = QSE(dimension, measurements,Sprobabilities,maxNumIterations,min_hsDistance)
        fidelity = qi.state_fidelity(qi.DensityMatrix(np.array(rho_target)), 
                                             qi.DensityMatrix(np.array(rho_estimate)))
        Fidelities.append(fidelity)
    
    return Fidelities
    
    
def boos_probs(probs,countings,numOfbases,dimension):
    
    probabilities = np.zeros((numOfbases,dimension))
    for i in range(numOfbases):          
        pdi  = np.array( [0] + list(probs[i]) )
        pdi, _ = np.histogram(list(np.random.random_sample(countings)), np.cumsum(pdi))
        probabilities[i] = pdi/np.sum(pdi)

    return probabilities
        
    
def born_rule_probabilities(bases, rho, numberOfBases,dimension):
    probabilities=np.zeros((numberOfBases,dimension))
    for i,basis in enumerate(bases):
        probabilities[i] = np.diag(np.conj(basis).T @ rho @ basis).real
    return probabilities

                                    
def probability_distributions(bases, dimension, rho, no_countings ):
    
    probabilities = np.zeros((len(bases),dimension))
    for i,basis in enumerate(bases):     
        pdi = np.diag(np.conj(basis).T @ rho @ basis).real
        pdi  = np.array( [0] + list(pdi) )
        pdi, _ = np.histogram(list(np.random.random_sample(no_countings)), np.cumsum(pdi))

        probabilities[i] = pdi/np.sum(pdi)
    
    return probabilities
         
                                    
def pauli_bases(num_of_qubits):
    
    basisZ = np.array([[1,0],[0,1]])
    basisX = np.array([[1,1],[1,-1]])/np.sqrt(2)
    basisY = np.array([[1,1],[1j,-1j]])/np.sqrt(2)
    
    qubitBasis = [basisZ,basisX,basisY]
    
    if num_of_qubits == 1:
        return qubitBasis
    else:
        pauliIndexes = list(it.product(range(3),repeat=num_of_qubits))
        return [ft.reduce(lambda x, y: np.kron(x,y),[qubitBasis[i] for i in index]) for index in pauliIndexes]


class PauliBases:
    basisZ = np.matrix([[1,0],[0,1]])
    basisX = np.matrix([[1,1],[1,-1]])/np.sqrt(2)
    basisY = np.matrix([[1,1],[1j,-1j]])/np.sqrt(2)
    
    qubitBases = [basisZ,basisX,basisY]
    
    no_qubits = 2
    
    def __init__(self,numberOfQubits):
        self.no_qubits=numberOfQubits
    
    def _generatorRecursion(self,num_of_qubits):
        if num_of_qubits == 1:
            yield from self.qubitBases
        else:
            yield from (np.einsum('ik,jl', prod, aBasis).reshape(2**num_of_qubits,2**num_of_qubits)
                         for aBasis in self.qubitBases 
                         for prod in self._generatorRecursion(num_of_qubits-1))
    
    def generator(self):
        return self._generatorRecursion(self.no_qubits)


def T(rho_0, measurement , probabilities ):
    aBasis=measurement
    rotatedState = np.conj( aBasis ).T @ rho_0 @ aBasis 
    stateWithImposedInformation = rotatedState - np.diag( np.diag( rotatedState ) ) + np.real( np.diag( probabilities ) )
    return aBasis @ stateWithImposedInformation @ np.conj( aBasis ).T


def applyCompositionT(rho_0,Measurements,Probabilities):
    idxs = [i for i in range(len(Measurements))]
    random.shuffle(idxs)
    rho_1 = rho_0
    for j in idxs:
        rho_1 = T(rho_1, Measurements[j], Probabilities[j] )
    return rho_1


def projectToDensityOperator(rho,dimension):
    
    eivals,u = np.linalg.eigh(rho)
    positiveEivals =  eivals.real[eivals.real>0]
    fullSpectra = np.concatenate(( np.zeros(dimension - len(positiveEivals)), positiveEivals ))
    rho = u@np.diag(fullSpectra)@np.conj(u).T
    
    return rho/np.trace(rho)


def HS_distance(rho,sigma):
    return np.linalg.norm(rho - sigma)
\end{lstlisting}

\section*{B: Codes for certification of quantum nonlocality}
\label{app:B}
The method described in chapter \ref{cap3} for quantum nonlocality certification was implemented in Python. Here, we present the code of that implementation. The code is available in Github \cite{code_nlc}. \\
\begin{lstlisting}[language=Python, basicstyle=\scriptsize]
import numpy as np
from scipy.optimize import minimize
from scipy.optimize import NonlinearConstraint
import matplotlib.pyplot as plt
import os
import itertools

from datetime import datetime
cwd = os.getcwd()


def optimal_value(filename, m, num_of_outcomes, num_of_trials, F, marginals_A, 
                  marginals_B, disp = False, save = False):
    
    """
    This routine optimizes R 
    
    Arguments:
    m: number of settings.
    filename: Name of the file where the experimental data is stored.
    marginals_A: a list containing the marginals (expressions like A_i x I ) 
                 in the Alice's side that have to be considered for the 
                 calculations.
    marginals_B: a list containing the marginals (expressions like  I x B_i ) 
                 in the Bob's side that have to be considered for the 
                 calculations.
    disp: 'True' for displaying the partial results.
    save: 'True' to save the coefficients S, Sax and Sby as numpy tensors in 
           the 'npy_data' folder.
                         
    Returns: 
    Tensors "S", "Sax" and "Sby" containing the coefficients 
    """
    
    foundResults = {}
    SolutionFound = False
    dir1 = ( cwd + '\\npy_data\\' + 's_' + 
        filename + datetime.today().strftime('_%Y-%m-%d') )
    dir2 = ( cwd + '\\npy_data\\' + 'sax_' + 
        filename + datetime.today().strftime('_%Y-%m-%d') )
    dir3 = ( cwd + '\\npy_data\\' + 'sby_' + 
        filename + datetime.today().strftime('_%Y-%m-%d') )
    dir4 = ( cwd + '\\npy_data\\' + 'p_' + 
        filename + datetime.today().strftime('_%Y-%m-%d') )
    dir5 = ( cwd + '\\npy_data\\' + 'pax_' + 
        filename + datetime.today().strftime('_%Y-%m-%d') )
    dir6 = ( cwd + '\\npy_data\\' + 'pby_' + 
        filename + datetime.today().strftime('_%Y-%m-%d') )
    dir7 = ( cwd + '\\npy_data\\' + 'Q_C_Dq_Gap' + 
        filename + datetime.today().strftime('_%Y-%m-%d') )
    
    print('n: number of iterations')
    print(60*'_')
    print(f'    n       R         Q         C         \u0394Q         Gap')
    print(60*'_')
    p, c, _ = load_data(filename, F, m, num_of_outcomes, 
                        marginals_A, marginals_B)
    s0, bounds = initial_guess(m, num_of_outcomes, 
                               marginals_A, marginals_B) # initial guess for the optimizer
    n, target = 1, -1
    nlc = NonlinearConstraint( lambda x: -target  
                              + R( x, p, c, m, num_of_outcomes, marginals_A, marginals_B), 
                                      -np.inf, -target ) # constraints

    while n <= num_of_trials:
        res = minimize( R , s0 ,  args=(p, c, m, num_of_outcomes, marginals_A, marginals_B), 
                        method='SLSQP', tol = 1e-20, 
                        options={ 'ftol':1e-20, 'maxiter':5000, 'disp': False}, 
                        bounds=bounds, constraints=[nlc])
        s0 = (res.x + s0)/2  # average between the target guess 's0' and the current solution
                             # R = -res.fun. Thus, 'res.fun < target' => 'R > -target'
        if  res.fun < target:
            SolutionFound = True
            solution = res # keeps the solution that satisfies the constraints
            target = res.fun  # updates the target to the best R value found so far
            nlc = NonlinearConstraint( lambda x: -target  
                                      + R( x, p, c, m, 
                                          num_of_outcomes, marginals_A, 
                                          marginals_B), 
                                      -np.inf, -target ) # update the constraints
            
            Q, C, Dq, gap = results( solution.x, p, c, 
                                    m, num_of_outcomes, 
                                    marginals_A, marginals_B)
            s, sax, sby = vector_to_tensor(solution.x, p,
                                           m, num_of_outcomes, 
                                           marginals_A, marginals_B)
            # save solutions at the 'npy_data' folder
            pabxy, pax, pby = p
            if save:
                np.save(dir1, s)
                np.save(dir2, sax)
                np.save(dir3, sby)
                np.save(dir4, pabxy)
                np.save(dir5, pax)
                np.save(dir6, pby)
                np.save(dir7, np.array([Q, C, Dq, gap]) )

            # display results
            if disp:
                print( "%5d    %4.5f   %4.5f   %4.5f    %4.5f    %4.5f" 
                      %(n, -solution.fun, Q, C, Dq, gap) )    
               
        n += 1
    
    if SolutionFound:
        foundResults['Coefficients'] = (s, sax, sby)
        foundResults['values'] = (Q, C, Dq, gap)
        return foundResults
    else:
        print("\nNo solution has been found ...")
    

    
def probs_local_hidden_model(num_of_outcomes , m):
    
    outcomes = [1] + [0 for i in range(num_of_outcomes-1)]
    rows = ( list(i) for i in set( itertools.permutations(outcomes)) )
    
    for i in itertools.product(rows, repeat = m):
        yield np.array(i).T 

        
def lhv_value(s, m, num_of_outcomes ):
    """
    Inputs:
    s: tuple containing the tensors (S, Sax, Sby)
    num_of_outcomes: number of outcomes
    m: number of settings
    
    Returns: 
    local bound for a 2-party scenario, with "m" settings 
    and "num_of_outcomes" possible outcomes
    """
    s, sax, sby = s
    Cmax = -1e10
    for pax in probs_local_hidden_model(num_of_outcomes , m):
        for pby in probs_local_hidden_model(num_of_outcomes , m):
            C =  ( np.sum(sax*pax) + np.sum(sby*pby)  
                  + np.sum( s*np.einsum('ax,by->abxy', pax,pby) ) )
            if C > Cmax:
                Cmax = C
    return Cmax


def quantum_max_violation(s, p):
    """
    Inputs:
    s: tuple containing the tensors (S, Sax, Sby)
    p: tuple containing the tensors (p, pax, pby)
    
    Returns:
    Quantum Value
    """
    s, sax, sby = s
    p, pax, pby = p
    return np.sum(sax*pax) + np.sum(sby*pby) + np.sum(s*p)


def error_quantum_value( s, c, m, num_of_outcomes):

    c,cax,cby = c
    s,sax,sby = s
    Delta_Q, Delta_Q1, Delta_Q2 =  0, 0, 0
    for i in itertools.product(range(m), repeat = 2):
        for l in itertools.product(range(num_of_outcomes), repeat = 2):
            x,y = i
            a,b = l
            dQ = (s[a,b,x,y]*np.sum(c[:,:,x,y]) 
                  - np.sum(s[:,:,x,y]*c[:,:,x,y]))/(np.sum(c[:,:,x,y])**2)
            if False:
                dQ1 = (sax[a,x]*np.sum(cax[:,x]) 
                       - np.sum(sax[:,x]*cax[:,x]))/(np.sum(cax[:,x])**2)
                dQ2 = (sby[b,y]*np.sum( cby[:,y]) 
                       -np.sum(sby[:,y]*cby[:,y]))/(np.sum(cby[:,y])**2)
                Delta_Q1 += dQ1**2*cax[a,x]
                Delta_Q2 += dQ2**2*cby[b,y]
            else:
                dQ = dQ + (1/m)*(sax[a,x]*np.sum(c[:,:,y,x]) 
                         - np.sum(sax[:,x]*c[:,:,x,y]))/(np.sum(c[:,:,x,y])**2)
                dQ = dQ + (1/m)*(sby[b,y]*np.sum(c[:,:,y,x]) 
                         - np.sum(sby[:,y]*c[:,:,x,y]))/(np.sum(c[:,:,x, y])**2)
            Delta_Q += dQ**2*c[a,b,x,y]
            
    return np.sqrt( Delta_Q  + Delta_Q1 + Delta_Q2 )



def initial_guess(m, num_of_outcomes, marginals_A, marginals_B):    

    nop = (num_of_outcomes**2)*(m**2) + num_of_outcomes**2*(len(marginals_A) + len(marginals_B))
    lb, ub = -1, 1
    bounds =  tuple( [(lb, ub) for i in range(nop)] )
    s0 = np.array(  [(ub - lb) * np.random.random() + lb   for i in range(nop)] )
    
    return s0, bounds


def R( s, p, c, m, num_of_outcomes, marginals_A, marginals_B):

    k = (num_of_outcomes)**2*(m)**2
    s = vector_to_tensor(s, p, m, num_of_outcomes, marginals_A, marginals_B)
    Q = quantum_max_violation(s, p)
    C = lhv_value(s, m, num_of_outcomes )
    Dq = error_quantum_value( s, c, m, num_of_outcomes)
    
    return  -(Q - Dq + k)/(C + k)


def vector_to_tensor(s, p, m, num_of_outcomes, marginals_A, marginals_B):
    
    o = num_of_outcomes
    n = m**2*o**2
    sax, sby = np.zeros((2,m)), np.zeros((2,m))
    s, s_marginals = np.array( s[:n] ).reshape((o,o,m,m)),  s[n:]
    
    j = 0
    if len(marginals_A) > 0:
        for x in marginals_A:
            for a in range(o):
                sax[a,x] = s_marginals[j]
                j += 1
    if len(marginals_B) > 0:
        for y in marginals_B:
            for b in range(o):
                sby[b,y] = s_marginals[j]
                j += 1
    s[np.nonzero(p[0] == 0)] = np.zeros((s[np.nonzero(p[0] == 0)].shape[0]))
    
    return (s, sax, sby)



def results( s, p, c, m, num_of_outcomes, marginals_A, marginals_B):
    
    o = num_of_outcomes
    n = m**2*o**2
    sax, sby = np.zeros((2,m)), np.zeros((2,m))
    s, s_marginals = np.array( s[:n] ).reshape((o,o,m,m)),s[n:]
    
    j = 0
    if len(marginals_A) > 0:
        for x in marginals_A:
            for a in range(o):
                sax[a,x] = s_marginals[j]
                j += 1
    if len(marginals_B) > 0:
        for y in marginals_B:
            for b in range(o):
                sby[b,y] = s_marginals[j]
                j += 1
    
    s[np.nonzero(p[0] == 0)] = np.zeros((s[np.nonzero(p[0] == 0)].shape[0]))
    s = (s, sax, sby)
    Q = quantum_max_violation(s, p)
    C = lhv_value(s, m, o )
    Dq = error_quantum_value( s, c, m, o)
    
    return  Q, C, Dq, (Q - Dq - C)



def load_data(name_of_file, F, m, num_of_outcomes, marginals_A, marginals_B):
    
    
    filename = cwd + '\\experimental_data\\' + name_of_file + '.txt'
    
    with open(filename,'r') as f:
        data = f.read()

    data = data.split('\n')
    labels_counts, counts = [] , []
    for i in range(len(data)):
        x,y = tuple( data[i].split()) 
        labels_counts.append(x)
        counts.append(float(y))
    
    countings = dict(zip(labels_counts, counts))
    
    
    shape = (num_of_outcomes,)*2 + (m,)*2 
    p = np.zeros( shape )
    c = np.zeros( shape )
    pax, cax =  np.zeros( (2,2) ), np.zeros( ( 2,2) )
    pby, cby =  np.zeros( (2,2) ), np.zeros( ( 2,2) )
    
    measured_marginals = False
    for l in labels_counts:
        a,b,x,y = l
        if a == '-':
            measured_marginals = True
            cax[int(a),int(x)] = countings[l]
        elif b == '-':   
            measured_marginals = True
            cby[int(b),int(y)] = countings[l]
        else:
            a,b,x,y = list( map( int,l ) )
            c[a,b,x,y] = countings[l]

    for l in labels_counts:
        a,b,x,y = l
        a,b,x,y = list( map( int,l ) )
        p[a,b,x,y] = c[a,b,x,y]/(np.sum(c[:,:,x,y]))
    
    p1 = p
    num_of_trials = 10
    output = optimal_settings_KL_divergence(F, p, num_of_outcomes, m, num_of_trials)
    p = get_probs_from_settings(output.x,m)
    
    if measured_marginals:
        for l in labels_counts:
            a,b,x,y = l
            a,b,x,y = list(map(int,l))
            pax[int(a), int(x)] = cax[int(a),int(x)]/(np.sum(cax[:,int(x)]))
            pby[int(b), int(y)] = cby[int(b), int(y)]/(np.sum(cby[:,int(y)]))

    if len(marginals_A) > 0:
        for x in marginals_A:
            for a in range(2):
                pax[a,x] = np.sum(p[a,: ,x ,0])
    if len(marginals_B) > 0:
        for y in marginals_B:
            for b in range(2):
                pby[b,y] = np.sum(p[:,b ,0 ,y])
                 
    return (p, pax, pby) , ( c, cax, cby ), measured_marginals




# Kullback-Leible divergence

def U(x):
    
    beta, gamma, alpha= x[0], x[1], x[2]
    return np.array( [ [ np.exp( -1j*gamma )*np.cos( beta/2 ) , - np.exp( 1j*alpha )*np.sin( beta/2 ) ], 
                    [ np.exp( -1j*alpha )*np.sin( beta/2 ), np.exp( 1j*gamma )*np.cos(beta/2) ] ] )   


def W(theta):
    
    ket0 = np.array([1,0])
    ket1 = np.array([0,1])
    psi = np.cos(theta)*np.kron(ket0,ket0) + np.sin(theta)*np.kron(ket1,ket1)
    
    return np.outer( psi, np.conjugate( psi ) )


def get_probs_from_settings(x, m):
    
    o = 2
    theta, x = x[0],  x[1:]
    rho = W(theta)
    Us = np.array( list( map(U, x.reshape( int(x.shape[0]/3) ,3) ) ) )
    Ua, Ub = Us[:int(Us.shape[0]/2)], Us[int(Us.shape[0]/2):]
    
    pa, pb, p = np.zeros((o,m)), np.zeros((o,m)), np.zeros((o,o,m,m))
    for i in itertools.product(range(m), repeat = 2):
        for l in itertools.product(range(o), repeat = 2):
            x,y = i
            a,b = l
            Pa = np.outer( Ua[x][:,a], np.conjugate( Ua[x][:,a] ) )
            Pb = np.outer( Ub[y][:,b], np.conjugate( Ub[y][:,b] ) )
            p[a,b,x,y] = np.real( np.trace( rho @ np.kron(Pa, Pb) ) )
    
    return  p


def initial_data(m, num_of_parties):
    
    lb , ub = 0, np.pi
    x0 = np.array( [(ub - lb) * np.random.random_sample() 
                    + lb for i in range(3*num_of_parties*m)] )
    
    lb , ub = 0, np.pi
    theta = (ub - lb) * np.random.random_sample() + lb
    x0 = np.insert(x0 , 0, theta )
    
    ### Bounds
    lb , ub = 0, np.pi
    x_bounds = np.array( [(0, np.pi/4)] 
                        + [(lb, ub) for i in range(3*num_of_parties*m)] )
    
    return x0, x_bounds


def kullback_leibe_divergence(x, F, f, num_of_outcomes, m):
    
    """
    f: relative frequencies (original probabilities)
    """
    p = get_probs_from_settings(x, m)
    
    D = 0
    for i in itertools.product(range(num_of_outcomes), repeat = m):
        for j in itertools.product(range(m), repeat = m):
            a,b = i
            x,y = j
            D += F(x,y)*f[a,b,x,y]*(np.log2(f[a,b,x,y])-np.log2(p[a,b,x,y]))
            
    return D



def optimal_settings_KL_divergence(F, probs, num_of_outcomes, 
                                               m, num_of_trials):
    
    """
    This routine minimize the Kullback-Leibe divergence.
    It optimizes over the parameters in the settings.
    """
    
    # one can also use COBYLA in this optimization
    # just comment the next line and uncomment the 
    # line corresponding to COBYLA.
    
    OPTIMIZER = 'SLSQP' 
#     OPTIMIZER = 'COBYLA'

    # initial guess and bounds for the optimizer
    x0, bounds = initial_data(num_of_outcomes, m)
    threshold = 0
    cons = ({'type': 'ineq', 'fun': lambda x: threshold 
             - kullback_leibe_divergence(x, F, probs, num_of_outcomes, m)} )

    res = minimize( kullback_leibe_divergence , x0 ,  
                    args=(F,probs, num_of_outcomes, m), 
                    method=OPTIMIZER, tol = 1e-20, 
                    options={ 'maxiter':250, 'disp': False},
                   constraints=cons
                  )

    previous = res.fun
    result = res
    for i in range(num_of_trials):
        
        x0 = (res.x + x0)/2
        res = minimize( kullback_leibe_divergence , x0 ,  
                    args=(F, probs, num_of_outcomes, m), 
                    method=OPTIMIZER, tol = 1e-20, 
                    options={ 'maxiter':250, 'disp': False},
                    constraints=cons
                  )
        
        # save the best result
        if res.fun < previous:
            previous = res.fun
            result = res
    
    return result



\end{lstlisting}

\section*{C: Codes for the quantum marginal problem algorithm}
\label{app:C}
The following codes correspond to the algorithms presented in chapter \ref{cap4}. The code was written in Python and is available in Github \cite{qmp_codes}.\\
\begin{lstlisting}[language=Python, basicstyle=\scriptsize]
import qiskit.quantum_info as qi
import itertools
import numpy as np
import matplotlib.pyplot as plt
from matplotlib.ticker import MaxNLocator
import time
import h5py
import scipy
from scipy import optimize
from datetime import datetime
from datetime import timedelta
import os
cwd = os.getcwd()



def spectra(self):
    return np.linalg.eigvalsh(self.data)

qi.DensityMatrix.spectra = spectra



class DensityMatrix(qi.DensityMatrix):
    def __init__(self, data, dims=None):
        self._labels = None
        super().__init__(data, dims)
    
    def get_labels(self):
        return self._labels
    
    def set_labels(self, given_labels):
        self._labels = given_labels
    
    def spectra(self):
        return np.linalg.eigvalsh(self.data)

    
    
def create_storing_folers():
    
    path_file_plots = cwd + '\\plots\\'
    path_file_data = cwd + '\\data\\'
    path_file_txtdata = cwd + '\\txtdata\\'
    directories = [path_file_plots, path_file_data, 
                   path_file_txtdata]
    
    created = False
    for pathname in directories:
        try:
            os.mkdir(pathname)
            created = True
        except:
            pass
    if created:
        print("Folders to store data has just been created ...")
        

def partial_trace(rho, antisystem):
    
    n = len( rho.dims() )
    lista = [n-i-1 for i in antisystem]
    lista.sort()
    labels = set(range(n)) - set(antisystem)
    sigma = DensityMatrix( qi.partial_trace(rho, tuple(lista) ), 
                          dims= rho.dims()[:len(labels)] )
    sigma.set_labels(list(labels))
    return sigma


def qmp(d, num_of_qudits, prescribed_marginals, rank = 1, prescribed_spectra = [], params = {}):
    
    start = time.time()
    
    if len(prescribed_spectra) > 0:
        print("prescribed spectra mode..")
    else:
        print("prescribed rank mode..")

    print(33*'_')
    print(f'      n     tdist     gdist')
    print(33*'_')
    
    
    rho_gen = params['rho_gen']
    txt_file_name = (cwd + '\\txtdata\\' + '_N' + 
                     str(num_of_qudits) + 'd' + str(d) + 'data.txt')
    
    if params['save_iter_data']:
        try:
            open(txt_file_name,'w').close()
            f = open(txt_file_name,'a')
        except:
            f = open(txt_file_name,'w')
  
    dn, swapper_d = d**num_of_qudits, swapper(d)
    dims = tuple([d for i in range(num_of_qudits)])
    
    if params['x_0'] != None:
        x_0 = params['x_0']
    else:
        if params['seed'] == 'mixed':
            x_0 = qi.DensityMatrix( qi.random_density_matrix(dims))
        if params['seed'] == 'mms':
            x_0 = DensityMatrix(np.identity(dn)/dn, dims)
        if params['seed'] == 'pure':
            x_0 = qi.DensityMatrix( qi.random_statevector(dims))
    
    mdistance, edistance, gdistance, dtotal = [], [], [], []
    number_of_iterations, total_distance, previous_total_distance = 0, 10, 10
    marginal_hsd, eigenvals_distance = 10, 10

    runtime = 0

    while total_distance > params['dtol'] and number_of_iterations < params['max_iter']:
        
        x_0 = impose_marginals( d, num_of_qudits, x_0, 
                               prescribed_marginals, swapper_d )
        
        if len(prescribed_spectra) > 0:
            eigenvals_distance, x_0 = impose_prescribed_eigenvals(x_0, prescribed_spectra)
        else:
            eigenvals_distance, x_0 = impose_rank_constraint(x_0, rank)
            
        x_0 = qi.DensityMatrix(x_0, dims)

        # metrics
        resulting_marginals, marginal_hsd = compute_marginals_distance( x_0, prescribed_marginals, 
                                                                       num_of_qudits)     
        total_distance = np.sqrt(marginal_hsd**2 + eigenvals_distance**2)
        dg = np.linalg.norm( x_0.data - rho_gen.data)
        gdistance.append( np.real( dg  ) ) 
        edistance.append( eigenvals_distance ) 
        mdistance.append( np.real( marginal_hsd ) ) 
        dtotal.append( total_distance )    

        if number_of_iterations % params['iter_to_print'] == 0:
            print( "%7d  %4.5E  %4.5E" %( number_of_iterations + 1, 
                                         total_distance, dg ) )
        # stores in h5 data
        if params['save_h5']:
            save_partial_data(x_0, mdistance, edistance, 
                              gdistance, stype = params['h5_name'])
        # stores txt data
        end = time.time()
        runtime += end - start
        if params['save_iter_data']:
            f.write(f'%i  %1.10f  %s\n'%(number_of_iterations + 1, total_distance, 
                                         time_format(timedelta(seconds = runtime)) ))
            f.flush()
        
        number_of_iterations += 1
        start = time.time()
        
    
    if params['save_iter_data']:
        f.close()
    
    runtime += time.time() - start
    data = {'rho_n':x_0, 'mdistance': mdistance, 
            'edistance':edistance, 'gdistance':gdistance, 
            'tdistance':dtotal, 
            'runtime':time_format(timedelta(seconds = runtime)) }
    
    return data



def accelerated_qmp(d, num_of_qudits, prescribed_marginals, rank = 1,
                    prescribed_spectra = [], params = {}):
    
    start = time.time()
    
    if len(prescribed_spectra) > 0:
        print("prescribed spectra mode..")
    else:
        print("prescribed rank mode..")
    
    print(33*'_')
    print(f'      n     tdist     gdist')
    print(33*'_')
    
    rho_gen = params['rho_gen']
    txt_file_name = (cwd + '\\txtdata\\' + '_N' 
                     + str(num_of_qudits) + 'd' 
                     + str(d) + 'data.txt')
    
    if params['save_iter_data']:
        try:
            open(txt_file_name,'w').close()
            f = open(txt_file_name,'a')
        except:
            f = open(txt_file_name,'w')
  
    dn, swapper_d = d**num_of_qudits, swapper(d)
    dims = tuple([d for i in range(num_of_qudits)])
    
    if params['x_0'] != None:
        x_0 = params['x_0']
    else:
        if params['seed'] == 'mixed':
            x_0 = qi.DensityMatrix( qi.random_density_matrix(dims))
        if params['seed'] == 'mms':
            x_0 = DensityMatrix(np.identity(dn)/dn, dims)
        if params['seed'] == 'pure':
            x_0 = qi.DensityMatrix( qi.random_statevector(dims))
    
    mdistance, edistance, gdistance, dtotal = [], [], [], []
    number_of_iterations, total_distance, previous_total_distance = 0, 10, 10
    marginal_hsd, eigenvals_distance, threshold = 10, 10, 1e-1

    alfa, alfa_n = params['alfa'], params['alfa_n']
    mu, bt = params['mu'], params['bt']
    beta = alfa_n**2 + bt

    runtime = 0
    if not params['accelerated']:
        alfa, beta, mu = 1, 0, 0 

    while total_distance > params['dtol'] and number_of_iterations < params['max_iter']:
        
        x_0 = accelerated_impose_marginals( d, num_of_qudits, x_0, 
                                           prescribed_marginals, swapper_d, 
                                           alfa_n, alfa, mu, beta)
        if len(prescribed_spectra) > 0:
            eigenvals_distance, x_0 = impose_prescribed_eigenvals(x_0, 
                                                      prescribed_spectra)
        else:
            eigenvals_distance, x_0 = impose_rank_constraint(x_0, rank)
        
        x_0 = qi.DensityMatrix(x_0, dims)
        resulting_marginals, marginal_hsd = compute_marginals_distance( x_0, 
                                               prescribed_marginals,num_of_qudits)     
        total_distance = np.sqrt(marginal_hsd**2 + eigenvals_distance**2)
        dg = np.linalg.norm( x_0.data - rho_gen.data)
        gdistance.append( np.real( dg  ) ) 
        edistance.append( eigenvals_distance ) 
        mdistance.append( np.real( marginal_hsd ) ) 
        dtotal.append( total_distance )    
        
        if params['accelerated']:
            alfa_n = 1/(1e-5*number_of_iterations + 1)**alfa  
            beta = alfa_n**2 + bt
        if number_of_iterations % params['iter_to_print'] == 0:
            print( "%7d  %4.5E  %4.5E" %( number_of_iterations + 1, 
                                         total_distance, dg ) )
        # stores in h5 data
        if params['save_h5']:
            save_partial_data(x_0, mdistance, edistance, 
                              gdistance, stype = params['h5_name'])
        # stores txt data
        end = time.time()
        runtime += end - start
        if params['save_iter_data']:
            f.write(f'%i  %1.10f  %s\n'%(number_of_iterations + 1, total_distance, 
                                         time_format(timedelta(seconds = runtime)) ))
            f.flush()
        number_of_iterations += 1
        start = time.time()
        
    if params['save_iter_data']:
        f.close()
    
    runtime += time.time() - start
    data = {'rho_n':x_0, 'mdistance': mdistance, 
            'edistance':edistance, 'gdistance':gdistance, 
            'tdistance':dtotal, 
            'runtime':time_format(timedelta(seconds = runtime)) }
    
    return data


def swapper(d):
    
    p = 0
    Id = np.identity(d)
    for i in range(d):
        for j in range(d):
            v = np.outer(Id[:,i],Id[:,j])
            u = np.transpose(v)
            p += np.kron(v,u)      
    return p


def kron(*matrices):
    
    m1, m2, *ms = matrices
    m3 = np.kron(m1, m2)
  
    for m in ms:
        m3 = np.kron(m3, m)
    
    return m3


def Pj( in_label, marginal, dl, num_of_qudits, swapper_d ):
    """
    dl: local dimension
    marginal: reduced system with labels given in the tuple "in_label"
    
    """
    
    label = in_label
    n = num_of_qudits - len( marginal.dims() ) 
    dims = tuple( [ dl for i in range( num_of_qudits ) ] )
    swapped_matrix = kron( marginal.data, np.identity( dl**n ) )
    
    all_labels = [ i for i in range( num_of_qudits ) ]
    right_labels = [ i for i in range( list( label )[-1] + 1, num_of_qudits ) ]
    left_labels = [ i for i in range( list( label )[0] ) ]
    
    if left_labels + list( label ) + right_labels == all_labels:
        nl  = list(label)[0] 
        nr = num_of_qudits - nl  - len(label)
        Il, Ir = np.identity( dl**nl ), np.identity( dl**nr )
        swapped_matrix = qi.DensityMatrix( kron( Il, marginal.data, Ir ) , dims )
        swapped_matrix = swapped_matrix/swapped_matrix.trace()       
        return swapped_matrix
    else:
        nl = list(label)[0] 
        nr = num_of_qudits - nl  - len(label)
        Il, Ir = np.identity( dl**nl ), np.identity( dl**nr )
        swapped_matrix = kron( Il, marginal.data, Ir )
        label = tuple( left_labels + list( label ) )
        
    length = len( label )
    remaining = tuple( [ i for i in range( length ) ] )
    
    while length > 0 and label != remaining:
        
        last = label[-1]
        numOfswapps = np.abs( last  - length ) 
        l1, l2 = length - 1, num_of_qudits - ( length + 1 )
        I1, I2 = np.identity( dl**l1 ), np.identity( dl**l2 )
        gate = kron( I1, swapper_d, I2 ) 
        swapped_matrix = gate @ swapped_matrix @ gate
        
        for i in range( numOfswapps ):
            l1, l2 = l1 + 1, l2 - 1 
            I1, I2 = np.identity( dl**l1 ), np.identity( dl**l2 )
            gate = kron( I1, swapper_d, I2 )
            swapped_matrix = gate @ swapped_matrix @ gate

        label = tuple( list( label[:-1] ) )
        length = len( label )
        remaining = tuple( [ i for i in range( length ) ] )
    
    swapped_matrix = qi.DensityMatrix( swapped_matrix, dims )
    swapped_matrix = swapped_matrix/swapped_matrix.trace()
    
    return swapped_matrix


def impose_prescribed_eigenvals(x_0, prescribed_spectra):
    
    dn, _ = x_0.data.shape
    rank = len(prescribed_spectra)
    prescribed_spectra = sorted(prescribed_spectra)
    prescribed_spectra = np.concatenate(( np.zeros((dn-rank)), 
                                         prescribed_spectra ))
    eigenvalues, eigenvects = np.linalg.eigh( x_0.data )    
    eigenvals_distance = np.linalg.norm( eigenvalues - prescribed_spectra )
    rhof = eigenvects @ np.diag( prescribed_spectra ) @ np.conjugate( eigenvects.T )
    
    return eigenvals_distance, rhof/np.trace(rhof)


def impose_rank_constraint(x_0, rank):
    
    dn, _ = x_0.data.shape
    eigenvalues, eigenvects = np.linalg.eigh( x_0.data )
    
    a, b = np.zeros((dn-rank)), eigenvalues[ dn - rank: ]
    if len(eigenvalues[ eigenvalues < 0]) <= rank:
        eigenvals_distance = np.linalg.norm( eigenvalues[ eigenvalues < 0] )
        eigenvalues[ eigenvalues < 0] =  - eigenvalues[ eigenvalues < 0]
    else:
        eigenvals_distance = np.linalg.norm( eigenvalues[:-rank] )
   
    rhof = eigenvects @ np.diag( np.concatenate((a,b)) ) @ np.conjugate( eigenvects.T )

    return eigenvals_distance, rhof/np.trace(rhof)


def impose_marginals(d, num_of_qudits, x_0, prescribed_marginals, swapper_d):

    dn = d**num_of_qudits
    all_systems = set( list( range(num_of_qudits)) )

    for l in list( prescribed_marginals.keys() ):
        antisys = tuple( all_systems - set( l ) )
        tr_rho0_I = partial_trace( x_0 , list( antisys ) )
        x_0 = (x_0  + Pj( l, prescribed_marginals[l], d, num_of_qudits, swapper_d ) -  
               Pj( l, tr_rho0_I, d, num_of_qudits, swapper_d )  )

    return x_0


def accelerated_impose_marginals(d, num_of_qudits, x_0, prescribed_marginals, 
                                 swapper_d, alfa_n, alfa, mu, beta):

    all_systems = set( list( range(num_of_qudits)) )
    dims = tuple( [ d for i in range( num_of_qudits ) ] )
    dn = d**num_of_qudits

    for l in list( prescribed_marginals.keys() ):
        antisys = tuple( all_systems - set( l ) )
        tr_rho0_I = partial_trace( x_0, list(antisys) )
        d_0 = ( Pj( l, prescribed_marginals[l], d, num_of_qudits, swapper_d ).data 
               - Pj( l, tr_rho0_I, d, num_of_qudits, swapper_d ).data )/alfa
        d_1 = d_0 + beta*d_0
        y_0 = x_0.data + alfa*d_1
        x_0 = mu*alfa_n*x_0 + (1 - mu*alfa_n)*y_0
        x_0 = qi.DensityMatrix( x_0, dims )

    return x_0/x_0.trace()


def compute_marginals_distance(rho0, prescribed_marginals,  num_of_qudits):
    
    all_systems = set( list( range(num_of_qudits)) )
    marginal_hsd = 0
    projected_marginals = {}
    
    for l in list( prescribed_marginals.keys() ):
        antisys = tuple( all_systems - set(l) )
        projected_marginals[l] = partial_trace( rho0, list(antisys) )
        marginal_hsd += np.linalg.norm( projected_marginals[l].data 
                                       - prescribed_marginals[l].data)**2
    norm = len( list( prescribed_marginals.keys() ) )
    marginal_hsd = np.sqrt(marginal_hsd/norm)
    
    return projected_marginals, marginal_hsd



def simul_data(d, num_of_qudits, labels_marginals, dtype = "pure" ):

    dn = d**num_of_qudits
    dims = tuple([d for i in range(num_of_qudits)])
    
    if dtype == "AME":
        rho_gen = np.identity(dn)/dn
        rho_gen = qi.DensityMatrix(rho_gen, dims)
    elif dtype == "mixed":
        rho_gen = qi.random_density_matrix(dims).data 
        rho_gen = qi.DensityMatrix(rho_gen, dims)
    elif dtype == "pure":
        rho_gen = qi.DensityMatrix(qi.random_statevector(dn), dims)

    marginals = {}
    all_systems = set( list( range(num_of_qudits)) )

    for s in labels_marginals:
        tracedSystems = tuple( all_systems - set( s ) )
        if len(tracedSystems) > 0:
            marginals[s] = partial_trace(rho_gen, list(tracedSystems))
        else:
            marginals[s] = partial_trace(rho_gen, list(s))
            
    return marginals, rho_gen



def time_format(runtime):

    t = str(runtime)
    h,mi,s = t.split(':')
    s = str(np.round(np.float(s),2))
    
    f,dec = s.split('.')
    if int(f) < 10:
        s = '0' + s

    if int(dec) < 10:
        dec = dec + '0'
    
    return h + ':' + mi + ':' + s



def plot_data(d, num_of_qudits, edistance, mdistance, 
              gdistance, runtime, params = {} ):
    
    cwd = os.getcwd()
    a, h = 7.024, 4.82
  
    plt.figure(figsize=(a,h), dpi = 100 )
    plt.figure().gca().xaxis.set_major_locator(MaxNLocator(integer=True))

    plt.plot( list( range( len(edistance) ) ), edistance, 
             label = r'$\mathcal{D}_{\lambda}$', 
             linestyle='-', linewidth = 1.3,
             markersize = 3.6, mfc='none',  
             color='royalblue')
    plt.plot(  list( range( len(mdistance) ) ), mdistance, 
             label = r'$\mathcal{D}_{M}$', 
             linestyle='--', linewidth = 1.3, 
             markersize = 3.6, mfc='none', 
             color='tab:red')
    if params['plot_global']:
        plt.plot(  list( range( len(gdistance) ) ), gdistance, 
                 label = r'$\mathcal{D}_{G}$', 
                 linestyle='dashdot', linewidth = 1.3, 
                 markersize = 3.6, mfc='none', 
                 color='forestgreen')
        
    plt.title(f'N = {num_of_qudits} ,  d = {d}  ( runtime = {time_format(runtime)})')
    plt.yscale('log')
    plt.xlabel("n", fontsize = 10)

    plt.legend( loc= 'lower left', fontsize = 12)
    
    plt.xlim(0,  list( range( len(edistance) ) )[-1] )
    plt.tick_params(direction='out', axis='y', 
                    which='minor', colors='black')
    
    if params['save_plot']:
        name = ( params['name'] + '_N' + str(num_of_qudits) 
                + 'd' + str(d)  + datetime.today().strftime('_%Y-%m-%d')
                + '.png')
        path_file = cwd + '\\plots\\' + name
        plt.savefig(path_file, format='png',bbox_inches='tight')

    plt.show()


    
def save_partial_data(rho0, mdistance, edistance, gdistance, stype = 'pure'):
    
    Data = {'mhsd': np.array( mdistance), 'edist': np.array(edistance),
            'gdist': np.array(gdistance), 'rho_n': rho0.data}
    
    N = len( rho0.dims() )
    local_dim, *_ = rho0.dims()
    
    file_name = ( cwd + '\\data\\' + stype + 
                 "_rhoN"+ str(N) + "d" 
                 + str(local_dim) + ".h5" )
    try:
        with h5py.File(file_name,'w') as f:
            for l in list( Data.keys() ):
                dtset = f.create_dataset( l, data = Data[l], 
                                         compression="gzip", 
                                         compression_opts=9 )
                f.flush()
    except:
        with h5py.File(file_name,'x') as f:
            for l in list( Data.keys() ):
                dtset = f.create_dataset( l, data = Data[l], 
                                         compression="gzip", 
                                         compression_opts=9 )
                f.flush()
    else:
        with h5py.File(file_name,'a') as f:
            for l in list( Data.keys() ):
                try:
                    del f[l]
                    dtset = f.create_dataset( l, data = Data[l], 
                                             compression="gzip", 
                                             compression_opts=9 )
                except:
                    dtset = f.create_dataset( l, data = Data[l], 
                                             compression="gzip", 
                                             compression_opts=9 )
                f.flush()
                
                
\end{lstlisting}

    

    
    



    
    \renewcommand*{\bibfont}{\small}
    \addcontentsline{toc}{chapter}{Bibliography}
    \printbibliography
 


\end{document}